\documentclass{nws}

\newcommand\bcmdtab{\noindent\bgroup\tabcolsep=0pt%
  \begin{tabular}{@{}p{10pc}@{}p{20pc}@{}}}
\newcommand\ecmdtab{\end{tabular}\egroup}

\title[Network Science]
      {Distance Closures on Complex Networks}

 \author[Tiago Simas \& Luis M. Rocha]
         {Tiago Simas$^1$ and Luis M. Rocha$^{1,2,3}$\\
          $^1$Cognitive Science Program, Indiana University, Bloomington, IN 47406, USA\\
          $^2$Center for Complex Networks and Systems, School of Informatics \& Computing, Indiana University, Bloomington, IN 47406, USA \\
         $^3$Instituto Gulbenkian de Ciencia, Oeiras, Portugal}

\jdate{March 2013}
\pubyear{2013}
\pagerange{\pageref{firstpage}--\pageref{lastpage}}
\doi{S0956796801004857}

\begin{document}

\label{firstpage}

\maketitle

\begin{abstract}
% Text of abstract
To expand the toolbox available to network science, we study the isomorphism between distance and Fuzzy (proximity or strength) graphs. Distinct transitive closures in Fuzzy graphs lead to closures of their isomorphic distance graphs with widely different structural properties. For instance, the All Pairs Shortest Paths (APSP) problem, based on the Dijkstra algorithm, is equivalent to a metric closure, which is only one of the possible ways to calculate shortest paths in weighted graphs.
We show that different closures lead to different distortions of the original topology of weighted graphs. Therefore, complex network analyses that depend on the calculation of shortest paths on weighted graphs should take into account the closure choice and associated topological distortion.
We characterise the isomorphism using the max-min and Dombi disjunction/conjunction pairs. This allows us to: (1) study alternative distance closures, such as those based on diffusion, metric, and ultra-metric distances; (2) identify the operators closest to the metric closure of distance graphs (the APSP), but which are logically consistent; and (3) propose a simple method to compute alternative distance closures using existing algorithms for the APSP. In particular, we show that a specific diffusion distance is promising for community detection in complex networks, and is based on desirable axioms for logical inference or approximate reasoning on networks; it also provides a simple algebraic means to compute diffusion processes on networks. Based on these results, we argue that choosing different distance closures can lead to different conclusions about indirect associations on network data, as well as the structure of complex networks, and are thus important to consider.
\end{abstract}

\tableofcontents

\newtheorem{definition}{Definition}
\newtheorem{theorem}{Theorem}
\newtheorem{lemma}{Lemma}
\newtheorem{corollary}{Corollary}
\newtheorem{example}{Example}
\newtheorem{Algorithm}{Algorithm}

% main text
\section{Introduction}
\label{introduction}

The majority of research on complex networks treats interactions as
binary edges in graphs, even though interactions in real networks
exhibit a wide range of intensities or strengths.  The varying
strength nature of many, if not most, real networks has lead us
towards a more recent drive to study complex networks as weighted
graphs \cite{Newman2001_PREII,Barrat2004PRL,Wang2005a,Goh2005}.
Certainly this shift towards weighted graphs as models of complex
networks is welcomed. However, there is still much to do to bring
decades of research on weighted graphs to bear on the field of
complex networks. One field, in particular, that has accumulated
substantial knowledge about weighted graphs is the field of Fuzzy
Set Theory \cite{Klir1995}.

While the Fuzzy Set community has focused extensively on the
mathematical characteristics of weighted graphs and how to compute
with them \cite{Mordeson2000}, it has not focused much on
developing models of the general principles that explain the
structure and dynamics of complex networks obtained from empirical
data.
Conversely, the complex networks community has paid relatively little attention to the mathematics of weighted
graphs.

We argue that the field of complex networks can particularly profit from learning more about
the algebraic characteristics of various
ways to compute the \emph{transitive closure} of weighted graphs obtained from
real data.
The concept of transitive closure is important because it
allows us to identify not only transitive cliques in a network, but also
indirectly related items; that is, those for which we do not possess direct co-occurrence data, but which may be strongly related via short indirect paths.
Extraction of indirectly related items is important for automatic inference in many problems such as recommender systems, text mining, information retrieval, and prediction of online social behavior. In particular, for these problems, we have previously shown that pairs of items for which we do not have direct co-occurrence data, but which are strongly related via indirect paths possess a
higher probability of direct co-occurrence in the future \cite{Rocha2002,Rocha2005,Simas2012}.
Interestingly, unlike standard binary or crisp graphs, in weighted
graphs there is an infinite number of ways to compute transitive
closure, and therefore, to compute indirect associations in the data.
This means that we should be aware of the effects of different forms of \emph{transitivity} of complex networks modeled as weighted graphs.

Our analysis is based on an \emph{isomorphism} between fuzzy (proximity/strength) and distance graphs, whereby transitive closure is isomorphic to the concept of \emph{distance closure}, out of which many alternative measures of indirect association in network data, including (shortest) path length, ensue.
For instance, Dijkstra's  algorithm
\cite{dijkstra} is ubiquitously used in the field of complex networks to compute shortest paths.
As we show below, this algorithm leads to the very intuitive \emph{metric closure} of a distance graph.
However, via the isomorphism, we show that it is equivalent to a transitive closure based on a pair of logical operations that does not satisfy De Morgan's laws for any involutive complement.
These are undesirable axiomatic features if one is interested in reasoning logically about knowledge represented in complex networks (see
below).

The isomorphism allows us to study how alternative distance/transitive closures impose a different \emph{distortion} of the topology of the original network---e.g. the metric closure (Dijkstra algorithm) enforces a metric topology on a distance graph.
Moreover, it allows us to obtain alternative closures and thus alternative ways to compute indirect associations in complex networks with ideal axiomatic features.
Indeed, different distance closures lead to different ways of computing path length, which is a fundamental building block of the network science methodology, used to compute shortest paths, community structure, etc.
Here, in addition to the metric closure, we study the \emph{ultra-metric} and a \emph{diffusion distance} closure which, unlike the former, possess desirable axiomatic features for reasoning about knowledge stored in networks.
While the ultra-metric closure distorts the original network topology more than the metric closure, the diffusion distance closure is defined such that items in communities are brought closer together, and items in bridges are put relatively further apart as closure is computed.
Therefore, our algebraic algorithm to compute the distance closure of weighted graphs, also provides an alternative and promising means to study diffusion processes on networks.

\section{Background}

\subsection{Complex networks}

In the last few years, much work has been done to understand the
general mechanisms that influence the growth and dynamics of
\emph{complex networks}, understood as systems of variables that are related to one another via some mechanism.
Examples of such relations are: interactions between physical objects (e.g. suppliers and consumers in an electrical grid), social ties (e.g. friendship and trust between people), associations and correlations in data (e.g. gene regulation and phenotypic traits), and many others.
While there are more sophisticated mathematical methods to model multivariate interactions (e.g. hypergraphs and relations \cite{Klamt:2009kx,Klir1995,Mordeson2000}),  the structure and dynamics of
\emph{complex networks} have been mostly studied using graph theory \cite{wasserman, WATTS1998, barabasi-1999-286, pastor_vesp, Dorogovtsev2003, Bornholdt2003}.
Indeed, graphs have been used to model the
Internet \cite{pastor_vesp}, the World Wide Web \cite{albert-2002-74},
collaboration networks \cite{barrat-2004-101,Newman2001},
biological networks \cite{Oltvai2002}, and many other types of multivariate interactions.

%
%Because of the pervasiveness of both the
%small-world phenomenon \cite{WATTS1998} and scale-free networks \cite{barabasi-1999-286} in nature, technology and
%society, there has been extraordinary interest in the
%structure and dynamics of complex networks \cite{pastor_vesp, Dorogovtsev2003, Bornholdt2003}.

%Typically, complex networks are conceptualized as a binary graph where edges
%represent a connection or association between vertices. The
%\emph{degree} $k$ of a vertex is the number of edges connected to
%it. It is convenient to characterize large graphs by their
%\emph{degree distribution}: the probability that the degree of a randomly chosen node is $k$
%\cite{pastor_vesp}. A \emph{power law distribution} is a
%distribution that follows the relation $P(k)\simeq ak^{-\gamma}$
%where $\gamma$ and $a$ are constants. Perhaps the key concept in the
%field of complex networks is that of \emph{scale-free network},
%defined as a graph whose degree distribution follows a power law with $\gamma$
%\cite{newman-2003-45}, typically with $2 < \gamma < 3$. This network structure is
%observed across many biological, social, technological and
%web-related applications
%\cite{albert-2002-74,Newman2001,Yook2002,pastor_vesp,Menczer2004}.

The majority of research on complex networks treats interactions as
binary or crisp edges in graphs, even though interactions in real networks
often exhibit a wide range of intensities or strengths. For instance, the
structure of web site access clearly depends on
heterogeneous amounts of traffic \cite{pastor_vesp}. The same
applies to air-transportation and scientific collaboration networks
\cite{barrat-2004-101,borner_complexity}. The intensity of
friendship (or familiarity) among people was also shown to be a
factor in the speed of epidemic spread \cite{yan-2005-22}.  The
varying strength nature of many, if not most, real networks have
lead towards a more recent drive to study complex networks as
weighted graphs
\cite{Newman2001_PREII,Barrat2004PRL,Wang2005a,Goh2005}.

Certainly this shift towards weighted graphs as models of complex
networks is welcomed. However, there is still much to do to bring
decades of research on weighted graphs to bear on the field of
complex networks. This is particularly true when it comes to
building informatics technology for the Web (e.g. recommender
systems and text mining \cite{Simas2012, verspoor05BMC, abihaidar_GB08}), or predicting social behavior online
\cite{monge&contractor}. Indeed, much work on weighted graphs has
been developed in the past decades in the context of database
research \cite{SHENOI1989}, information retrieval
\cite{Miyamoto1990}, filtering \cite{golbeck}, and social networking
\cite{Pujol}. Fuzzy
Set Theory \cite{Zadeh1965, Klir1995}, in particular, has accumulated
substantial knowledge about Fuzzy graphs \cite{Mordeson2000}, a type of weighted graph we summarize next.

\subsection{Fuzzy Graphs and Transitive Closure}
\label{background_fuzzy}

A $n$-ary relation, $R$, between $n$ sets $X_1, X_2, \cdots, X_n$,
assigns a value, $r$, to elements,  $\mathbf{x} = (x_1, x_2, \cdots,
x_n)$, of the Cartesian product of these sets: $X_1 \times X_2
\times \cdots \times X_n$. The value $r$ signifies how strongly the
elements (or variables) $x_i$ of the $n$-tuple $\mathbf{x}$ are related or
associated to one another (\cite{Klir1995} page 119).
When  $r
\in [0, 1]$, $R$ is known as a \emph{fuzzy relation}
\cite{Klir1995}, and when $n=2$ as a \emph{binary fuzzy relation}.
Binary fuzzy relations, $R(X, Y)$, can be easily represented by
adjacency matrices of dimension $n \times m$,  where $n$ and $m$ are the number
of elements of $X$ and $Y$ respectively. Examples of relevant binary
relations are: keywords $\times$ documents, users $\times$ web
pages, authors $\times$ citations,
%users $\times$ friends (in social
%networking sites), biomedical entities $\times$ documents (e.g. gene
%and protein names in PubMed articles),
etc.

%[TS - Revision (R.2.4 and R2.3)]
Binary fuzzy relations defined on a single set of variables, $R(X, X)$, are also
known as \emph{fuzzy graphs}---a kind of weighted graph where the
edges weights are defined in the unit interval.
In other words, the network of interactions amongst a set of variables $X$ is conceptualized as a binary fuzzy relation of the set with itself.
In general, the weights are unconstrained, but they can also be constrained to accommodate a probability mass function or other restrictions.

% Transitivity and T-Conorms and T-Conorms

A large edge weight between two elements in a fuzzy graph denotes a strong association or interaction between them. But what about a pair of elements that have weak links to one another, but have strong links with the same \emph{other}
elements? Should we infer that the pair of elements is strongly related via indirect associations, that is, from \emph{transitivity}?

To study the transitivity of a fuzzy graph, we need to compute the strength of interaction between any two nodes given all possible indirect paths between them.
There are, however, infinite ways to integrate numerically the weights in the indirect paths.
Menger \cite{Menger} first generalized transitivity criteria in the context of probabilistic metric spaces. To do this, \emph{triangular norm} (T-Norm) binary operations were introduced.
Later, Zadeh imported the concept of T-Norms to generalize logical operations in multi-valued logics such as Fuzzy logic \cite{Zadeh1965,Zadeh1999}.
A \emph{T-Norm} $\land:[0,1] \times [0,1] \rightarrow [0,1]$, is a binary operation with the properties of commutativity ($a \land b = b \land a$), associativity ($a \land (b \land c) = (a \land b) \land c$), and monotonicity ($a \land b \leq c \land  d$ iff $a\leq c$ and $b\leq d$). Moreover, $1$ is its identity element ($ a \land 1 = a$).
In other words, the algebraic structure $([0,1], \land)$ is a monoid \cite{Gondran2007}.
A T-Norm generalizes \emph{conjunction} in logic to deal with real values in the unit interval ($a,b \in [0,1]$), see details in \cite{Klir1995}.
Similarly, a \emph{T-Conorm} $\lor$ generalizes \emph{disjunction} and  has the same properties as a T-Norm, but $0$ is its identity element ($a \lor 0 = a$) \cite{Klir1995}.
Therefore, the algebraic structure $([0,1], \lor )$ is also a monoid \cite{Gondran2007}.
To obtain \emph{dual} T-Norm/T-Conorm pairs, we can derive a T-Conorm from a T-Norm via a generalization of De Morgan's laws: $a \lor b = 1 - ((1-a) \land (1-b))$.

%%% Composition

To integrate all indirect paths between every pair of nodes in a fuzzy graph, we can now use the \emph{composition of fuzzy graphs}, based on a pair of T-Conorm and T-Norm binary operations, $\langle \lor, \land \rangle$, which form
the algebraic structure $([0,1], \land, \lor )$. Notice that this structure is not necessarily a semiring on the unit interval \cite{Gondran2007} and can be more or less constrained to obtain desirable properties (see below).
The composition of fuzzy graphs is done via the logical composition of the graph's adjacency
matrix with itself ($R \circ R$), in much the same way as the algebraic product of
matrices, except that summation and multiplication are substituted by the T-Conorm and T-Norm, respectively \cite{Klir1995,Klement2004}:
For any disjuction/conjunction (T-Conorm/T-Norm) pair $\langle \lor, \land \rangle$, the
general composition of fuzzy graphs is:

 $$R \circ R = \bigvee_{k}
\bigwedge (r_{ik},r_{kj})=r'_{ij}$$

\noindent where $r_{ij}$ denotes $R(x_i, x_j)$,
the weight of the edge between vertices $x_i$ and $x_j$ of fuzzy graph $R$.
The most commonly used operations for disjunction and conjunction
are the \texttt{maximum} and
\texttt{minimum}, respectively.
Thus, the standard composition of fuzzy graphs is referred to as the
\emph{max-min composition}:

 $$R \circ R = \max_{k}
\min (r_{ik},r_{kj})=r'_{ij}$$

%LMR: For TC above, start with the series definition to infinity, then give intuitive definition already above. Then discuss issues of finite computation, then put algorithms for computing in appendix (both the one for max, and the other one I wrote)

The \emph{transitive closure} $R^T(X, X)$ of a fuzzy graph $R(X, X)$ can now be defined as:

\begin{equation}
\label{TC1}
%R^{T}=R \cup R^2 \cup \dots \cup R^k
R^{T} = \bigcup_{n=1}^{\kappa} R^n
\end{equation}

\noindent where $R^n = R \circ R^{n-1}$, for $n = 2, 3, ...$, and $R^1 = R$ \cite{Klir1995}.
Furthermore, the union of two graph adjacency matrices of the same size, $R \cup S$, is defined by the disjunction of their respective entries: $r_{ij} \lor s_{ij},\,\, \forall_{i,j}$, where $\lor$ denotes the same T-Conorm used in the composition.
In the most general case, $\kappa \rightarrow \infty$ \cite{Gondran2007}, but with reasonable constrains (see below), the transitive closure of finite graph converges for a finite $\kappa$.

Since different T-Conorm/T-Norm pairs can be employed in the composition of fuzzy graphs,
different criteria for transitivity can be established---a key concept in our
work.
Let us exemplify with the most commonly used form of
transitivity in fuzzy graphs, using the traditional
disjunction/conjunction (T-Conorm/T-Norm) pair $\langle \vee = \texttt{maximum}, \wedge = \texttt{minimum}\rangle$. A fuzzy graph $R(X, X)$ is \emph{max-min
transitive} iff:

 $$r_{ij} \geq \max_{\forall x_k \in X}
\min [r_{ik},r_{kj}], \forall_{x_i, x_j \in X}$$

This definition generalizes the transitive property of
crisp graphs, which requires that nodes $x_i$ and $x_j$ be linked
($r_{ij}=1$) if $x_i$ is linked to $x_k$ and $x_k$ to $x_j$
($r_{ik}=r_{kj}=1$).
In contrast, the (max-min) fuzzy transitivity requires that edge
$r_{ij}$ is at least as large as the maximum of the weakest links
(minimum edges) in each possible indirect path via some node $x_k$.
In other words, we compute the possible indirect paths between nodes
$x_i$ and $x_j$ via $x_k$, and identify the weakest edge in each
path. Then, from all these indirect paths, we choose the one with
the largest weakest edge.
Given eq. \ref{TC1}, this is done not just for a single intermediary node $x_k$, but for every indirect path of $\kappa-1$ intermediary nodes.

Notice that if the edge weights are not weighted, then all transitive closure criteria established by the possible T-Conorm/T-Norm pairs collapse to the standard transitive closure of crisp graphs. In other words, if $r \in \{0,1\}$, for any acceptable pair $\langle \vee, \wedge \rangle $, $R^T$ given by eq. \ref{TC1} yields a graph where $r^T_{ij} =1$ iff there is a path between $x_i$ and $x_j$ in graph $R$, and $r^T_{ij} =0$, otherwise; if $R$ is a connected graph, then $R^T$ is a complete graph.

When the transitive closure $R^T(X, X)$ uses the T-Conorm $\vee = \texttt{maximum}$, with any T-Norm $\wedge$,
%that obeys the Archimedean property (see appendix 1),
then $\kappa$ in eq. \ref{TC1} is finite and not larger than $|X|-1$
%\cite{Pang2003, Klir1995}.
\cite{Klir1995}.
In other words, the transitive closure converges in finite time and can be easily computed using Algorithm 1 defined in appendix A.
It has also been shown  that if the algebraic structure $([0,1], \land, \lor )$ is a dioid \cite{Gondran2007}, then $\kappa$ in eq. \ref{TC1} is also finite \cite{Han2004,Han2007} (see Appendix A).
In this case, the transitive closure can be computed in finite time using Algorithm 2 defined in appendix A.
Two of the main examples of transitive closure we develop here (metric and ultra-metric closure, see \S \ref{max_closures}) use the T-Conorm $\vee = \texttt{maximum}$, and therefore can be computed in finite time.
The third example we focus on, the diffusion closure (\S \ref{diffusion_section}), is based on an algebraic structure $([0,1], \land, \lor )$ which is not a dioid. However, as we show below,
%
%while the closure (quickly) converges asymptotically to zero distance for all edges in the graph as $\kappa$ increases,
%
the utility of this closure for complex networks resides in the first few $\kappa$ steps, and therefore finite-time convergence is not required.

%
%
%\subsection{Proximity Networks}
%\label{PNetworks}

We say that a fuzzy graph $R(X, X)$ is a \emph{similarity} graph if it is reflexive ($r_{ii} = 1$), symmetric ($r_{ij} = r_{ji}$), and transitive; $R(X, X)$ is
a \emph{proximity} graph if it is reflexive and symmetric \cite{Klir1995}. The
transitive closure of a proximity graph is a similarity graph, but
because there are many ways to define transitivity based on distinct
disjunction/conjunction pairs, there are also many ways to define
similarity.

\subsection{Representing and Fusing Knowledge in proximity networks}
\label{knowledge_networks_section}

To build complex networks from multivariate data we can use a number of measures of the strength of variable interaction.
For instance, we have previously derived proximity graphs from a co-occurrence measure that is a natural weighted extension \cite{Rocha1999} \cite{rochabollen2001} \cite{Popescu2006} of the Jaccard similarity measure \cite{Grefenstette94}, which has been used extensively in computational intelligence \cite{nakamura} \cite{Rocha2005}.
This co-occurrence measure yields proximity graphs which represent the closeness or strength of association of variables interacting in networks (e.g. terms extracted from
documents, or users of a social networking web site).
Proximity graphs can thus be seen as
\emph{associative knowledge networks} that represent how often
elements co-occur in some dataset \cite{Rocha2002,rocha03NATO}.
%
%
%
%Given a generic
%binary fuzzy relation $R$ between sets $X$ (of $n$ elements $x$) and $Y$
%(of $m$ elements $y$), we extract two complementary proximity
%graphs: $XYP$ and $YXP$. The weights of fuzzy graph $XYP$ are denoted by $xyp(x_i, x_j)$ and represent the probability that
%both $x_i$ and $x_j$ are related via $R$ to the same elements $y \in
%Y$ (and only those). Conversely, the weights of fuzzy graph $YXP$ are denoted by $yxp(y_i, y_j)$ and represent the probability
%that both $y_i$ and $y_j$ are related via $R$ to the same elements $x
%\in X$ (and only those). In short, these measures equate proximity
%with co-occurrence; the respective formulas are:
%
%\begin{equation}
%\label{proximity_measures}
%xyp(x_i,x_j)=\frac{\displaystyle\sum_{k=1}^m (r_{ik} \wedge
%r_{kj})}{\displaystyle\sum_{k=1}^m (r_{ik} \vee r_{kj})};\quad \
%yxp(y_i,y_j)=\frac{\displaystyle\sum_{k=1}^n (r_{ki} \wedge
%r_{kj})}{\displaystyle\sum_{k=1}^n (r_{ki} \vee r_{kj})}
%\end{equation}
%

Other co-occurrence measures can be used to capture a degree of
association or closeness between elements of two sets in a binary
relation. In information retrieval, in addition to variations of the
Jaccard measure, it is common to use the cosine \cite{BaezaYates99},
Euclidean \cite{Strehl} and even mutual information measures
\cite{turney01}.
%
%For characterizing closeness in relations, we
%prefer our weighted Jaccard proximity measure because it possesses
%several desirable characteristics. %Mutual information-based measures
%%are not symmetric, therefore they are neither proximity nor
%%similarity measures as defined above.
%The Euclidean measure is a
%similarity measure in the sense defined above (it is transitive for
%most commonly used criteria), but it generates non-sparse matrices,
%since all finite elements of the relation $R$ lead to similarity
%greater than zero. This makes it impractical for very large data
%sets. The cosine proximity measure (which is not transitive for most
%commonly used criteria) is scale-invariant which makes it very
%appealing for text documents of varying size, but may be problematic in other
%domains. The weighted Jaccard measure has aspects of both the
%Euclidean and the cosine measures \cite{Strehl}, and leads to
%sparse matrices.
%
%We, and many others, have used the weighted extension of the
%Jaccard proximity measure (eq. \ref{proximity_measures}) in several applications with good results (see below).
%
Nonetheless, all of the theoretical work we develop below applies to any proximity graph (as defined above), independently of the measure used to obtain it from specific data sets.
Notice also that proximity graphs are symmetrical (undirected). This is desirable because below we study their isomorphic distance graphs---distance is by definition symmetric. However, our work is directly applicable to acyclical directed graphs. It is also extendable to cyclical directed graphs, but transitive closure needs to be computed via an available efficient algorithm in that case (e.g. \cite{Nuutila1994}), because $\kappa$ in eq. \ref{TC1} is not necessarily finite with cyclical graphs (Algorithms 1 and 2 in Appendix A are not guaranteed to halt).

%While the Fuzzy Set community has focused extensively on the
%mathematical characteristics of fuzzy graphs and how to
%manipulate them, it has not focused much on the structure and
%dynamics of real networks obtained from empirical data.
%
%Conversely,
%the complex networks community has paid very little attention to the
%mathematics of weighted graphs. In the field of complex networks,
%there is much to do to understand the \emph{axiomatic}
%characteristics of various ways to build weighted graphs from real
%data. There is also ample need to study the effect of various forms
%of \emph{transitivity} on the structure and dynamics of weighted
%networks.
%
%Towards that goal, let us present in some details the
%necessary background on Fuzzy Graphs.
%
%Let us exemplify such use with our previous work.

%
%We derive our proximity networks using the proximity measures of
%formulae \ref{proximity_measures} computed from binary fuzzy relations
%extracted from large collections of documents, websites, or records
%stored in databases.
%Proximity graphs can be seen as
%\emph{associative knowledge networks} that represent how often
%elements co-occur in some dataset \cite{Rocha2002,rocha03NATO}. As with any other co-occurrence method, the assumption
%is that items that frequently co-occur are associated with a common
%concept, theme, or social community understood by the community of
%users and writers of the documents.

\begin{figure}[!th]
\centerline{\includegraphics[width=0.8\textwidth]{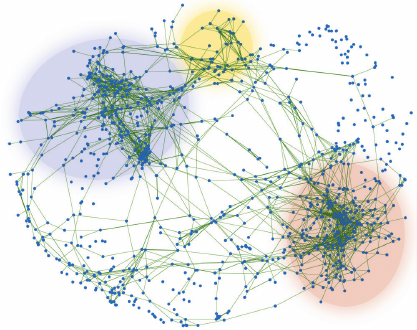}}
\caption{\small Social communities discovered in the proximity
network of journals accessed by users of the \emph{MyLibrary@LANL}
recommender system %\cite{Rocha2005}
. In this proximity network,
journals are closer to one another, if they tend to co-occur in the
same user profile, and only in those. Drawn using the
Fruchterman-Reingold algorithm in Pajek %\cite{pajek}
.}% Figure reprinted from \cite{Rocha2005}. }
\label{myLib_IPP}
\end{figure}

Notice that a proximity graph allows us to capture network
associations rather than just pair-wise interactions. In other
words, we expect concepts or social communities to be organized in
more interconnected sub-graphs, modules, or clusters of items in the
proximity networks.
Figure \ref{myLib_IPP} depicts a
proximity network extracted from the recommender system we developed
for the MyLibrary service of the digital library at the Los Alamos
National Laboratory (LANL)---details in \cite{Rocha2005}. The elements in this
network are scientific journals, and the proximity edge weights were computed
from co-occurrence of journals in user profiles.
The Principal Component
Analysis (PCA) \cite{Wall2003} of this network revealed two main
clusters of journals. The first component (eigen-vector) refers to a set of journals
related to ``Chemistry, Materials science and Physics'' (left,
blue). The second component refers to a set of journals related to
``Computer Science and Applied Mathematics'' (right, orange).
A smaller third cluster in the figure refers to ``Bioinformatics and Computational
Biology'' (top, yellow) \cite{Rocha2005}.
The main clusters discovered in this network capture the research threads pursued at LANL. Being a nuclear weapons
laboratory, much of its research is concerned with Materials Science
and Physics on the one hand, and Simulation and Computer Science on
the other. Thus, the journal proximity network, produced from user
profiles, captured the main communities of scientists (the users of
MyLibray) at Los Alamos, as well as the knowledge associated with
these communities (characterized by the journals in the respective
components). Our user tests of the quality of recommendation
based on the community structure of this network were quite good \cite{Rocha2005}.
This exemplifies how proximity networks obtained from co-occurrence data capture the knowledge traded by social collectives.

Additionally, using proximity networks to capture and extract knowledge in the biomedical literature led to very high performance on various information extraction tasks \cite{verspoor05BMC, abihaidar_GB08, Kolchinsky2010} of the \emph{BioCreative} text mining competition \cite{biocreative_issue}.
We have also tested recommendation of movies based on the clusters of the proximity network of users obtained from the \emph{MovieLens} benchmark with very good results \cite{Simas2012}. This
exemplifies how proximity networks can be seen as effective,
knowledge and social structure representations.
Indeed, the clusters of similar items obtained using this approach are isomorphic to the recently proposed method of \emph{link communities} in complex networks, which were shown to be excellent at uncovering the natural hierarchical organization of networks \cite{Ahn2010}.

Since proximity networks capture knowledge entailed by multivariate data, it would be very useful to be able to ``fuse'' and logically combine networks obtained from distinct data sets or situations. For instance, given the journal network from LANL shown in Figure \ref{myLib_IPP}, we could compute how journals are related by the Los Alamos Community \emph{or} the community of institution $A$ \emph{and not} of institution $B$. In other words, it would be good to be able to make inferences on networks fused via logical expressions.
This \emph{network fusion} is a thread of research that the network science field has not dealt with, but which can be achieved via the \emph{approximate reasoning} methodology of Fuzzy logic and other many-valued logics \cite{Ying:1994}. However, in order to pursue such a \emph{network approximate reasoning}, we must constrain the algebraic structure $([0,1], \land, \lor )$ such that the pair $\langle \land, \lor \rangle$ obeys minimal axiomatic properties such as De Morgan's laws for a negation/complement operation. Below (\S \ref{axiomatics}) we show that the diffusion distance closure we propose is based on a pair $\langle \land, \lor \rangle$ from the Dombi family of T-Norms \cite{dombi82} which is closest to the metric closure (Dijkstra) but, unlike the latter, obeys De Morgan's laws for any involutive complement. Therefore, the T-Norm/T-Conorm pair used for the diffusion distance, is a good candidate to pursue approximate reasoning on networks---the development of which is outside of the scope of this article.

\subsection{Semi-metric behavior in distance networks}
\label{semimetric_background}

%Proximity graphs capture associations amongst elements of a set,
%such as Journals or users, which are directly measured based on
%co-occurrence data.
%
%A high value of proximity means that two items from one set (e.g.
%words) tend to co-occur frequently in another set of objects (e.g.
%web pages). But what about items that do not co-occur frequently
%with one another, but do occur frequently with the same \emph{other}
%elements? In other words, even if two items do not co-occur much,
%they may occur very frequently with a third item (or more). Should
%we infer that the two items are related via indirect associations,
%that is, from \emph{transitivity}?

Here we study transitivity as a general topological
phenomenon of weighted graphs such as proximity networks---where it
can be computed in different ways.
While
%
%the Fuzzy Set community has focused extensively on the
%mathematical characteristics of various possible
%conjunction/disjunction (T-Norm/T-Conorm) pairs to compute
%transitivity \cite{Klement2004a,Klement2004b,Maes2006}, it has not
%focused much on the structure and dynamics of real networks obtained
%from empirical data. Indeed, there is very little work on the
%identification of the most intuitive and appropriate forms of
%transitivity for information retrieval, text mining, or network
%analysis in general. Conversely, while
%
the last decade witnessed
a tremendous amount of scientific production towards understanding
the structure of complex networks,
including the study of their topological features vis a vis the triangle inequality \cite{Serrano2008}, there is still much to be known about the effect of
various forms of transitivity on network structure.
%
%Our approach is an effort to bridge the research advances of these
%two distinct communities.

To build up a more intuitive understanding of transitivity in
weighted graphs, and to be able to relate our results to the most common methods used in the complex network field, we convert our proximity graphs to distance graphs.
Distance can be seen intuitively as the opposite of proximity, and is the most common way to conceptualize (shortest) path length in complex networks, e.g. via the Dijkstra algorithm \cite{dijkstra} (see below).
Various functions can be used to convert one into the other. Perhaps
the most common way to convert a fuzzy proximity graph $R(X,X)$ to a distance graph $D(X,X)$
%
%LMR: convert previous instances of R(X,X) to R(X).
%
is to use the simplest proximity-to-distance conversion
function \cite{Rocha2002}\cite{Strehl}:

% $\varphi$: $distance=\frac{1}{proximity}-1$

 \begin{equation}
\label{distance_measures4}
d_{ij} = \varphi (r_{ij})  = \frac{1}{r_{ij}} - 1,\quad \forall_{x_i,x_j \in X}
\end{equation}

%, which is the Dombi T-Norm generator with $\lambda=1$ (see \cite{klement}).
%From the generic proximity measures $XYP$ and
%$YXP$, obtained from a relation $R$ between sets $X$ and $Y$ using
%formulae \ref{proximity_measures}, we can compute generic distance
%functions among the elements of $X$ or $Y$:
%
%\begin{equation}
%\label{distance_measures4} d_X(x_i, x_j)= \frac{1}{xyp(x_i, x_j)} -
%1,\quad d_Y(y_i, y_j)= \frac{1}{yxp(y_i, y_j)} - 1
%\end{equation}

\noindent where $d_{ij}$ are the entries of the adjacency matrix of the distance graph $D(X,X)$,
and $\varphi: [0,1] \rightarrow [0, \infty]$ is a distance function because it yields
nonnegative, symmetric ($d_{ij}=d_{ji}$), and anti-reflexive ($d_{ii}=0$) values \cite{galvin_shore91}.
%They define weighted graphs $D_X$  and
%$D_Y$, which we refer to as \emph{distance graphs}, whose vertices
%$x_i$ or $y_i$ are the elements of $X$ or $Y$, and the edges are the
%values $d_X(x_i, x_j)$ and $d_Y(y_i, y_j)$, respectively.
%
A small
distance between elements implies a strong association between them.
%
%Distance is by definition symmetric, therefore our analysis here is restricted to undirected graphs.
%%
%However, all the analysis we present below is directly applicable to acyclical directed graphs. It is also easily extendable to cyclical directed graphs, but transitive closure needs to be computed via an available efficient algorithm in that case (e.g. \cite{Nuutila1994})

In general, distance graphs obtained from  data, (e.g. via co-occurrence data) are not entirely metric because, for some pair
of elements $x_i$ and $x_j$, the \emph{triangle inequality} may be
violated: $d_{ij} \geq d_{ik} + d_{kj}$ for some
element $x_k$. This means that the shortest distance between two
elements in $D(X,X)$ is not necessarily the direct edge but rather an
indirect path. Distance functions that violate the triangle
inequality are referred to as \emph{semi-metrics}
\cite{galvin_shore91}.
We say that the edge between a pair of nodes $x_i$ and $x_j$ in a distance graph is semi-metric when there is at least one indirect path between the nodes whose distance is shorter
than the direct edge: $d_{ij} > d_{ik} + \cdots  +  d_{lm} + \cdots + d_{pj}$.
The intensity of \emph{semi-metric behavior} is computed by comparing (e.g. via a ratio) how much shorter the indirect path is in relation to the direct link \cite{Rocha2002}.
Pairs of elements with large semi-metric behavior denote a type of \emph{latent association}  \cite{Rocha2002}.
That is, an association which is not grounded on the direct evidence used to build the distance graph (e.g. co-occurrence data), but rather indirectly implied by the
overall network of associations captured by the graph.
%

%Clearly, semi-metric behavior is a question of degree: some
%semi-metric shortcuts are much shorter than others depending on how
%much the triangle inequality is violated. Thus, to measure a degree
%of semi-metric behavior we can use the \emph{semi-metric}
%and \emph{below average ratios} \cite{Rocha2002}:
%
%\begin{equation}
%\label{semimetric_ratios} s(x_i,x_j)=\frac{d(x_i,
%x_j)}{\underline{d}(x_i, x_j)}, \quad b(x_i,
%x_j)=\frac{\hat{d}_{x_i}}{\underline{d}(x_i, x_j)}
%\end{equation}
%
%
%\noindent where $\underline{d}(x_i, x_j)$ is the shortest indirect
%distance between $x_i$ and $x_j$ in distance graph $D_X$, and
%$\hat{d}_{x_i}$ is the mean direct distance from $x_i$ to all
%other $x_k \in X$ such that $d(x_i, x_k)$ is finite.
%%
%$s(x_i, x_j) > 1$ for semi-metric edges with finite $d(x_i, x_j)>0$;
%
%The below average ratio $b$ is only applicable to semi-metric
%edges $(x_i, x_j)$ where $0 < \underline{d}(x_i, x_j) <
%d (x_i, x_j)$ and it measures how much the shortest indirect
%distance between $x_i$ and $x_j$ falls below the average distance of
%$x_i$ to all its directly associated elements $x_k$. The below
%average ratio is designed to capture semi-metric behavior of edges
%$(x_i, x_j)$ which do not have a finite direct distance $d(x_i,
%x_j)$.
%%
%Note that $b(x_i, x_j) \neq b(x_j, x_i)$.
%
%$b > 1$ denotes a
%below average distance reduction (see \cite{Rocha2002} for more
%details).

\begin{figure}[!th]
\centerline{\includegraphics[width=130mm,height=110mm]{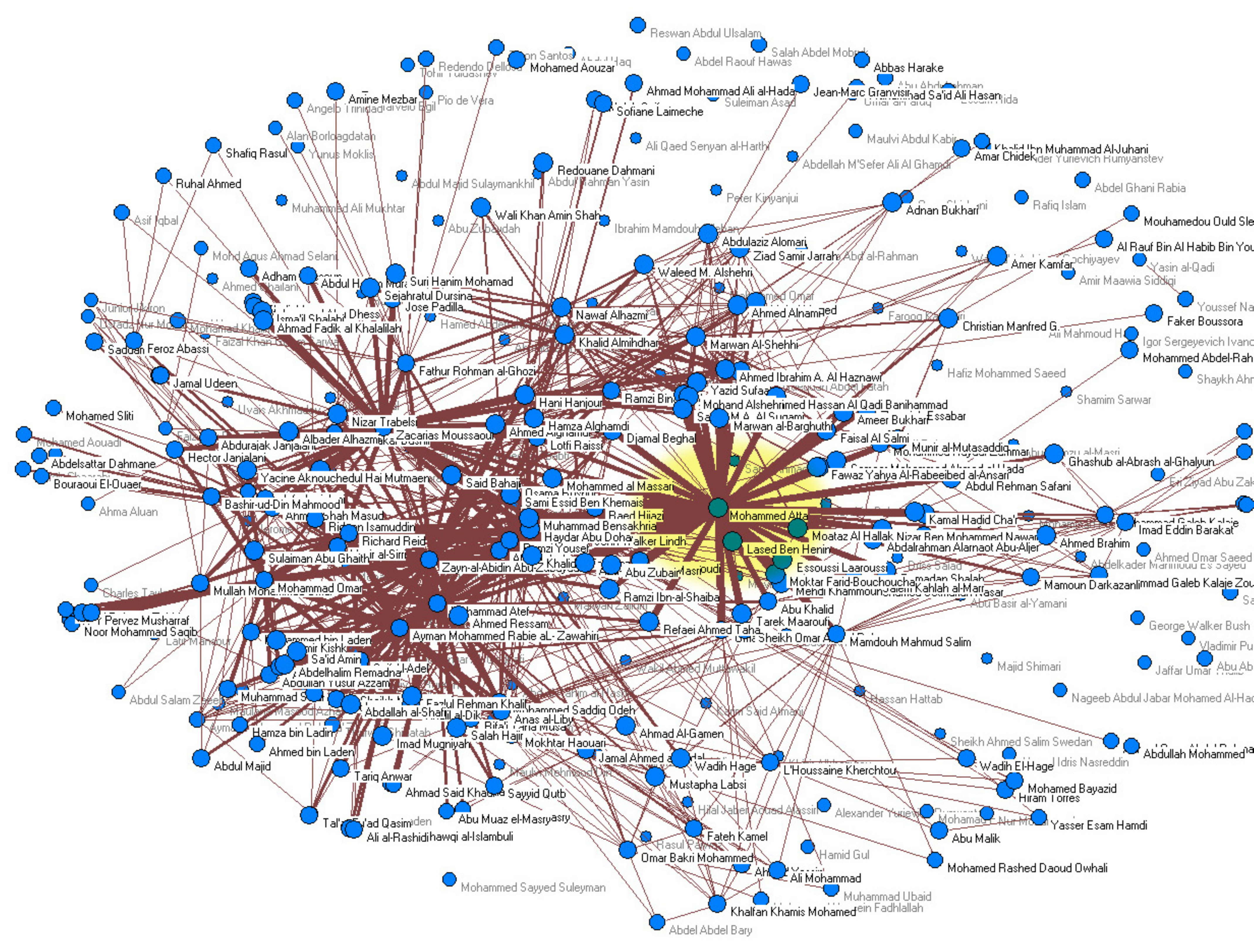}}
\caption{\small Terrorist proximity network obtained from
intelligence data related to the $9/11$ terrorist attacks on New York city and Washington DC; strongly semi-metric edges, shown with thicker lines. The node for Mohammed Atta is highlighted (yellow). The strong links out of this node, denote latent terrorist
associations not directly identified in the public-domain intelligence data, but highly
possible via indirect links picked by the semi-metric ratio. Drawn using the Fruchterman-Reingold algorithm in Pajek
%\cite{pajek}
 } \label{terr_net_semimetric}
\end{figure}

Rocha has proposed that in proximity graphs of keywords
extracted from documents, strong latent associations imply novelty in the temporal evolution of the network, and can thus be used to identify trends \cite{Rocha2002}. We have also
used and tested this idea, with good results, in a recommender
system that was implemented at LANL's digital library \cite{Rocha2005}.  In the
case of this service, a strong semi-metric association in the
journal network (figure \ref{myLib_IPP}) identifies a pair of
journals that hardly co-occur in user profiles, but which are
nonetheless very strongly implied via other journals which
co-occur with the pair.
The methodology also yielded competitive results in the \emph{MovieLens} benchmark \cite{Simas2012}, against the most common recommender system algorithms, and it has been used in the \texttt{givealink.org} project \cite{givealink05,givealink2006}.
We have also tested our method on social networks obtained from public-domain data about social interactions of terrorists associated with the September 11th attacks to the USA \cite{rocha02Terr}, showing that semi-metric information can identify valid, latent associations not directly observed in intelligence data (see Figure \ref{terr_net_semimetric}).

%%%% AQUI

%LMR: don't forget Mariaangles refes

Clearly, semi-metric behavior intuitively captures a form of (geometric) transitivity, but in the distance realm.
Below, we show how it is one of many types of transitivity than can be usefully used in complex networks. But let us first discuss the computational aspects of characterizing the semi-metric behavior of a distance graph.
\subsection{Computing Semi-metric pairs: metric closure}
\label{semimetric_pairs_section}
%
%From a practical standpoint, one is naturally interested in
%identifying the specific pairs of elements that are most
%semi-metric.
%
%These pairs are useful to issue recommendations  \cite{Rocha2005,givealink05,givealink2006,Simas2012}, to identify
%keywords appropriate to classify biological entities
%\cite{verspoor05BMC,abihaidar_GB08}, or social interactions that may
%have a higher chance of occurring in the future \cite{rocha02Terr}. These pairs of
%items are associated not by direct co-occurrence in the data, but
%are rather implied (as a global property) by the transitivity of the
%proximity networks obtained from the same data.
%%

The computation of all the shortest (indirect) paths between every pair of nodes in a distance graph is known as the \emph{All Pairs Shortest Paths} (APSP) problem,  one of the most fundamental algorithmic graph problems \cite{Zwick}.
The complexity of the fastest known algorithm for solving the APSP problem for weighted graphs is $O(mn+n^{2}log$ $n)$, where $n$ and $m$ are, respectively, the number of vertices and edges \cite{Brandes}. The most common approach to the APSP determines the distances of all pairs by calling the \emph{Single-Source Shortest-Path} (SSSP) Dijkstra algorithm $n$ times \cite{Brandes}\footnote{For directed cyclical graphs, before calling Dijkstra, this approach to the APSP uses the Bellman-Ford algorithm for removing all negative cycles and is known as Johnson's algorithm, which, for positive sparse weighted graphs, reduces to a time complexity $O(n^2log$ $n)$ \cite{Siek}. }. Here we refer to this algorithm as the APSP/Dijkstra algorithm. There are other approaches for solving the APSP problem, such as Floyd-Warshall algorithm \cite{Brandes}\cite{Siek}, but all of them fall in the $O(n^{3})$ complexity range \cite{Zwick}.

Notice that after computation of the APSP of a distance graph $D(X,X)$, we obtain its \emph{metric closure}.
In other words, we can construct a distance graph $D^{mc}(X,X)$, whose edges $d^{mc}_{ij}$ between any two elements $x_i$ and $x_j$ are defined by the shortest (direct or indirect) distance between them in $D(X,X)$.
An alternative way to compute the metric closure is to use the algorithm for transitive closure (Algorithm 1 in Appendix A), except that graph composition is done using the pair $(min, +)$  and instead of $\cup = max$ in step 1, we use $\cap = min$.
This method is also known as the \emph{distance product} which, after some simplifications, can reach a complexity of $O(n^{2.575})$, and is another approach to solving the APSP problem based on matrix operations \cite{Zwick}.

%
%
%The final matrix obtained with this process (which, as discussed below, is related to Dijkstra's method \cite{dijkstra}) is the metric
%closure $D^{mc}$ of graph $D$:
%
%: it is equivalent to the metric closure of such graphs.
%
%\begin{Algorithm}{Metric Closure (a.k.a. distance product \cite{Zwick})}
%\label{dca}
%\begin{enumerate}
%  \item $D' = (D \circ D)$
%  \item If $D'\neq D$, make $D = D'$ and go back to step 1.
%  \item  Stop: $D^{mc} = D'$
%\end{enumerate}
%\end{Algorithm}
%
%\noindent where  $[D \circ D]_{i,j} = \min_{k}(d_{ik} + d_{kj})=d'_{ij}$ is the distance product (see below).

Using $D^{mc}$, we identify all semi-metric edges in
$D$, by collecting those edges for which $d_{ij} > d^{mc}_{ij}$ is true. An
%
%Further, we
%choose the most semi-metric pairs edges in $D$, using the
%semi-metric ratios of formulae \ref{semimetric_ratios},
example is shown with a network of terrorists in figure \ref{terr_net_semimetric}, where thickness of edges denotes intensity of semi-metric behavior.
Figure
\ref{computing_semimetric_behavior}, depicts the general process of computing semi-metric behavior given the
proximity-to-distance map \ref{distance_measures4}.

\begin{figure}[!th]
%\centerline{\includegraphics[width=130mm,height=100mm]{computing_semimetric_behavior}}
\centerline{\includegraphics[width=1\textwidth]{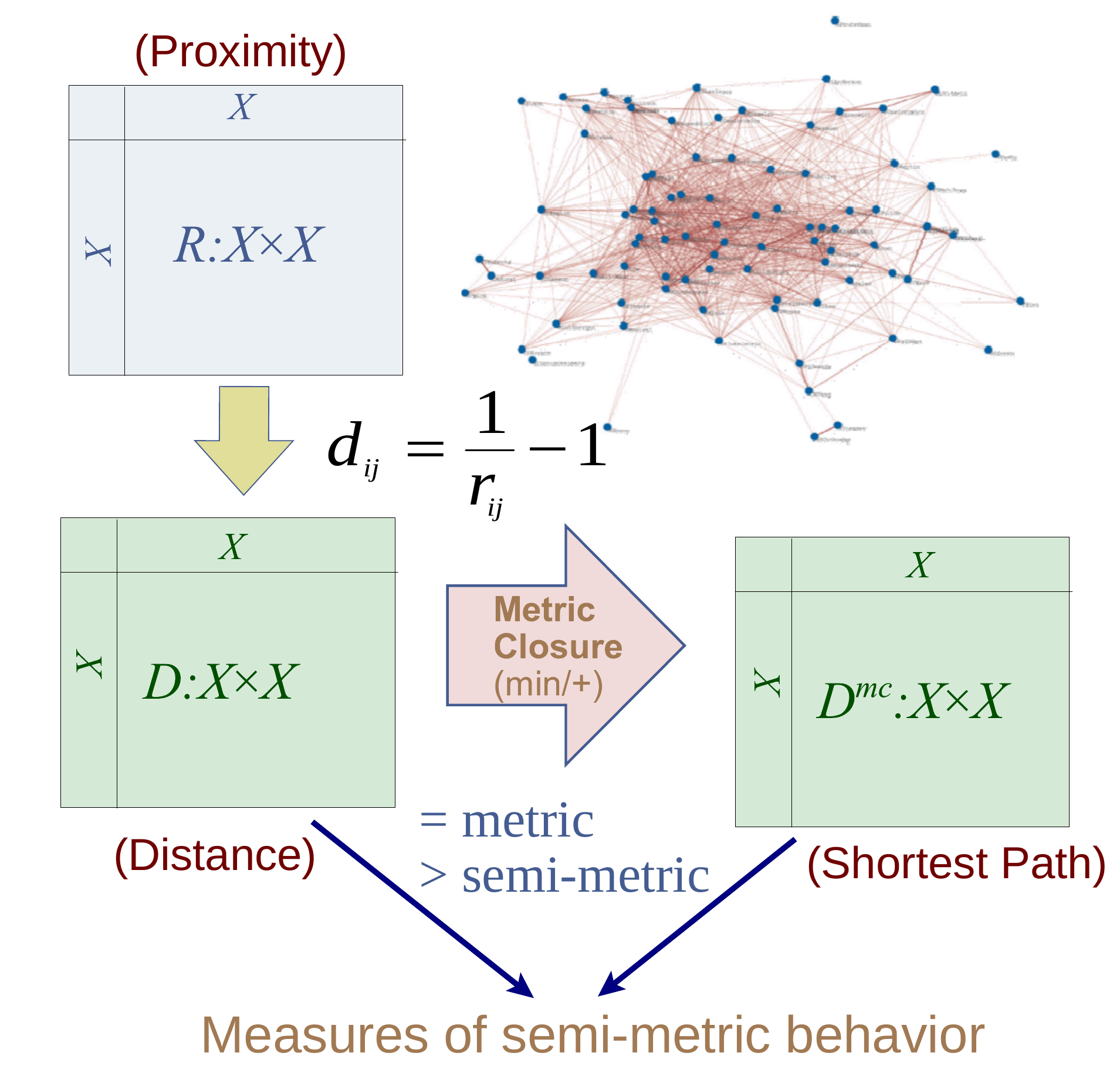}}
\caption{\small Computing semi-metric behavior. First,
a distance matrix/graph $D(X,X)$ is computed using eq.
\ref{distance_measures4}. Then, the metric closure of this matrix,
$D^{mc}$,  is computed using $(min, +)$ composition (distance product) or the APSP/Dijkstra algorithm. Semi-metric pairs are identified as: $d_{ij} > d^{mc}_{ij}$. }
\label{computing_semimetric_behavior}
\end{figure}

When we perform the metric closure, the geometry of the distance graph $D^{mc}$ is a \emph{distortion} of the geometry of the original graph $D$ obtained from data.
In other words, the original semi-metric topology extracted directly from data is forced to become metric (enforcement of the triangle inequality).
This is done by computing shortest paths in the most intuitive manner: summing edges in all paths and selecting the minimum.
However, there are many other possible ways to compute path length.
In the following sections, we define ageneral \emph{distance closure}, which includes the metric closure as a special case, and is shown to be isomorphic to the transitive closure in fuzzy (proximity) graphs.
This isomorphism allows us to use the formal edifice of (generalized) transitive closures of fuzzy graphs, on the theory and practice of complex networks modeled as weighted graphs.
Via this isomorphism we can use distinct transitive closures of fuzzy graphs to produce alternative measures of path length in distance graphs, which result in novel analytical possibilities for complex network models.
Each means of computing path length induces a distortion of the original relational data in a network, based on a specific transitivity criterion---e.g. the metric closure (APSP/Dijkstra algorithm) enforces a metric topology on a distance graph, where the transitivity criterion is the triangle inequality.
Additionally, some criteria are based on better axiomatic characteristics than others, as we discuss below.

The study of the geometry of complex networks has become increasingly relevant. For instance, there has been much interest in the assumption that the underlying geometry of complex networks is hyperbolic \cite{Krioukov2010}. This theory can explain their heterogeneous degree distributions and strong clustering, as simple reflections of the negative curvature of the underlying hyperbolic geometry, and can be useful to model biological \cite{Serrano2012} and technological networks \cite{Boguna2010}.
In our approach, we do not make claims about the underlying geometry of complex networks. Rather, we observe that most networks obtained directly from data via common measures (see \S \ref{knowledge_networks_section}) are strongly semi-metric, but are subsequently distorted via path length measures to become metric. Here, via the concept of transitive closure in fuzzy graphs, we want to study the various types of distortions one can impose on the original topology.

Notice further that our approach is not a generalization of shortest paths into \emph{fuzzy paths}, first introduced by Dubois and Prade \cite{dubois1980}, and extensively studied in the Fuzzy Sets community \cite{FSP1,FSP2,FSP3,FSP4}.
%
%LMR: The below is not needed, but I leave it in case reviewers comment about this and we want to bring it back.
%
%Two main approaches have been proposed to give a degree of fuzziness to the graph edges, \cite{FSP3} \cite{FSP2}. The classical fuzzy shortest path problem the length of a given edge in the graph is attributed a fuzzy number. The second approach the length of a given path is a fuzzy number and each edge in the graph has a membership value. The search for the fuzzy shortest path in the graph can be done using several approaches such as: using a dynamic programming formulation, \cite{FSP3}, methods based on the defuzzification of the fuzzy weights \cite{FSP3} and others \cite{dubois1980,FSP1,FSP4}.
%
Rather than generalizing the concept of shortest path (e.g. assigning fuzzy numbers to graph edges or paths \cite{FSP3,FSP2}), we use algebraic path length measures on distance graphs, which we show to be isomorphic to the generalized transitivity criteria of fuzzy graphs.
%

%The analysis presented here is for undirected graphs used to represent proximity/strength graphs and their isomorphic distance graphs, especially those obtained from co-occurrence data in large data sets as described above which are naturally symmetric.
%%
%However, all the analysis we present below is directly applicable to acyclical directed graphs. It is also easily extendable to cyclical directed graphs, but transitive closure needs to be computed via an available efficient algorithm in that case (e.g. \cite{Nuutila1994}).
%%
%Therefore, the derived concepts are not solely restricted to the undirected proximity and symmetric graphs we generate from co-occurrence data in our methodology, and can be used to study complex networks in general.

%LMR: added sentences above about extension to directed graphs.

\section{General distance closure}
\label{isomorph_section}

Transitive Closure is a well established algorithm in the theory of Fuzzy Graphs, used to calculate a similarity graph, whose edge weights are not weaker (by some transitivity criterion) than any indirect path between the same edge vertices.
%%
%%Transitive closure is also behind many definitions and theorems in the theory of Fuzzy Graphs \cite{Mordeson2000}.
%%
%Complex network science, including the study of such phenomena as small-world and scale-free networks, is heavily based on a related notion:
%%
%%graphs including weighted graphs \cite{Barrat2008}, and
%%
%the computation of shortest paths, typically using the Dijkstra algorithm \cite{dijkstra}.
%%
%This field can certainly profit from concepts already established in Fuzzy Graph Theory, such as generalised transitive closures.
%
In section \ref{semimetric_pairs_section} above, we defined the concept of metric closure, which is related to the APSP.
Metric closure is based on the very intuitive notions of Euclidean geometry, whereby path length is computed by summing constituent edge (distance) weights, and shortest paths are, in turn, picked by choosing the minimum path lengths---typically computed using the Dijkstra algorithm \cite{dijkstra} or the distance product \cite{Zwick,Rocha2002,Rocha2005}.
However, many other closures of distance graphs are possible, which we can easily formulate via an isomorphism $\varphi$ between proximity and distance graphs.
%
%to transitive closures in Fuzzy Graphs (binary fuzzy relations), using the mathematical framework of algebraic structures \cite{Gondran2007,Han2007,Han2004,SRTC1,SRTC2,SRTC3}.

\subsection{Proximity to Distance Isomorphism}
\label{isomorphism_definition}

%In order to explore the counterpart of transitive closure in distance graphs, we need to establish an isomorphism $\varphi$ between proximity and distance graphs.
%
Henceforth, without loss of generality, let us define a weighted graph as $G = (X,E)$, where $X$ is the set of vertices (or variables) and $E$ is the set of edges, which can also be represented by an adjacency matrix $E$ whose entries denote the weights of edges $e_{ij}$ between vertices $x_i$ and $x_j$.
\emph{Proximity graphs}, are fuzzy graphs $G_P(X,P)$ represented by adjacency matrices $P$ whose edge weights $p_{ij} \in [0,1], \forall x_i, x_j \in X$, such that $p_{ij}=p_{ji}$ (symmetry) and $p_{ii}=1$ (reflexivity).
Moreover, the composition of proximity graphs used to compute their transitive closure utilizes the algebraic structure $I=\{[0,1],\lor,\land\}$ where $\lor,\land$ are, respectively, T-Conorm and T-Norm binary operations (see \S \ref{background_fuzzy}).
Similarly, \emph{distance graphs} $G_D(X,D)$ are represented by adjacency matrices $D$ defined by edge weights $d_{ij} \in [0,+\infty], \forall x_i, x_j \in X$, such that $d_{ij}=d_{ji}$ (symmetry) and $d_{ii}=0$ (anti-reflexivity).
An isomorphic composition of distance graphs, leading to a \emph{distance closure} utilizes the algebraic structure $II=\{[0,+\infty],f,g\}$ where $f,g: [0,+\infty] \times [0,+\infty] \rightarrow [0,+\infty]$ are two binary operations.

%%[TS - Revision R2.12]
%Transitive closure as defined above (equation \ref{TC1}) can be seen as the quasi-inverse of an element working on an Algebraic Structure $I=\{S=[0,1],\lor,\land\}$,  where $\lor,\land:S\times S \rightarrow S$ are two binary operations, and $tc$ defined as in equation \ref{TC1}, where the union is defined as the $\lor$. Likewise, \emph{distance closure} can be seen as the quasi-inverse of an element working on the Algebraic Structure $II=\{S'=[0,+\infty],f,g\}$, where $f,g:S'\times S' \rightarrow S'$ are two binary operations on the elements of set $S'$ and $dc$ is the quasi-inverse, where the union is now defined as the binary operator $f$, see appendix 1 for definitions and \cite{Gondran2007}.
%%[TS -Revision R2.12 end]

The map $\varphi: [0,1] \rightarrow [0,+\infty]$, which converts proximity ($p_{ij}$) to distance($d_{ij}$) weights, must satisfy the constraints imposed on $\langle \lor,\land,f,g \rangle$ and consequently on algebraic structures $I$ and $II$, such that the transitive closure of a proximity graph is isomorphic to the distance closure of a corresponding distance graph.
For instance, we show below that $\varphi$ is necessarily a generator function \cite{klement} of the \emph{T-Norm} $\land$ in algebraic structure $I$, when $\lor \equiv \max$, $f \equiv \min$, and $g \equiv +$ (see \S \ref{other_closures}).
There are no linear functions than can satisfy the necessary constraints, because it maps the unit interval $[0,1]$ into the positive real line $[0,+\infty]$. However, there is an infinity of non-linear functions that satisfy the necessary constraints, the simplest of which is the map of formula \ref{distance_measures4}.
As we show below, each non-linear map $\varphi$ that satisfies the isomorphism constraints enforces a particular topological distortion of the original proximity graph used to construct the distance graph, which ultimately determines the way we compute path length and shortest paths.
This poses us with a problem of degeneracy of solutions to computing the distance closure of weighted graphs. Therefore, if we want to understand and make appropriate inferences about path lengths in complex networks, since an infinity of distance closures are possible, we should better understand the space of non-linear functions that enable isomorphism $\varphi$, which we approach below.
Figure ~\ref{fig3} depicts the isomorphism between proximity and distance graphs and algebraic structures $I$ and $II$.

\begin{figure}[h]
\centering
\includegraphics[width=75mm,height=55mm]{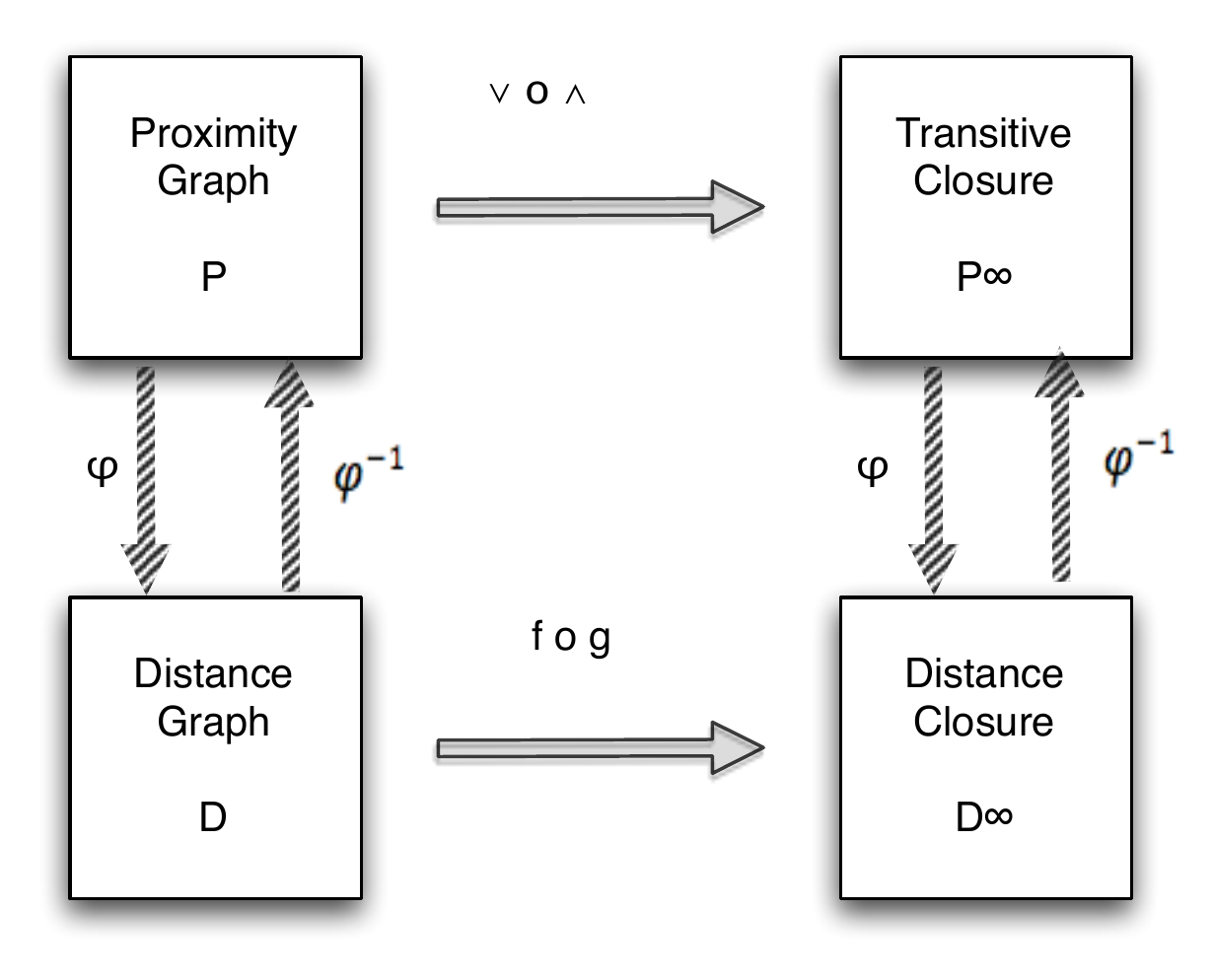}
\caption{Transitive and Distance Closure Isomorphism.\label{fig3}}
\end{figure}

%%LMR: the paragraph below was removed for now as it seems unnecessary given the background.

%%
%The concept of \emph{T-Norms} was introduced by Karl Menger to generalise transitivity in \emph{probabilistic metric spaces} \cite{schweizer83}. The results of Menger and his followers were then use to generalise the concept of Conjunctions (Unions) and Disjunctions (Intersections) in Fuzzy logic (Sets) \cite{Klir1995}.  The %[TS- Revision R2.13]
%All Pairs Shortest Paths (APSP) problem, based on the Dijkstra algorithm, is only one of the possible ways to calculate shortest paths -- via a metric closure. Transitive closure is a thus generalisation of the APSP problem, and, as it is well known in the fuzzy set community, there are infinite solutions to this problem, \cite{Klir1995}. Nonetheless, different \emph{T-Norms} provide lower and upper bounds of the strength of transitivity, where the strongest \emph{T-Norm} is the \emph{minimum} function and the weakest \emph{T-Norm} is the \emph{drastic product}\footnote{The drastic product T-Norm is defined as: $T_D(a,b)= (b \; \textrm{if} \; a = 1) \lor (a \; \textrm{if} \;  b = 1) \lor (0 \; \textrm{otherwise})$} \cite{klement} \cite{Klir1995}.
%%[TS- Revision R2.13 end]
%The ability to sweep the transitivity space that results from the \emph{T-Norm} bounds allows us to control and understand the topological distortion imposed on proximity graphs when, via the non-linear isomorphism, we convert to a distance graph to be fed to the APSP or metric closure.

%

%

%

\begin{definition}{\bf{(Graph Isomorphism)}}
\label{isomorphism}
 Two undirected weighted graphs $G_{1}=(X,E_{1})$ and $G_{2}=(X,E_{2})$ are isomorphic if there is a vertex-preserving bijective edge mapping $\varphi:E_{1} \to E_{2}$, i.e. a bijection $\varphi$ with  \[\forall x_i,x_j \in X : e_{ij} \in E_{1} \Leftrightarrow \varphi(e_{ij})\in E_{2}\]
\end{definition}

\begin{definition}{\bf{(Proximity to Distance Map)}}
\label{def1}
Let $\varphi:[0,1]\rightarrow [0,+\infty]$, $d_{ij}=\varphi(p_{ij})$, be a function that maps the edge weights $p_{ij}\in[0,1]$ of a fuzzy proximity graph $G_{P}=(X,P)$ into the edge weights $d_{ij}\in[0,+\infty]$ of a distance graph $G_{D}=(X,D)$, $\forall x_i,x_j \in X$. Let also $\Phi:[0,1]\times[0,1]\rightarrow[0,+\infty]\times[0,+\infty]$ be the graph function that maps the proximity adjacency matrix into the distance adjacency matrix, $D=\Phi(P)$. We define $\varphi$ and $\Phi$ in the following way:\\
\\(1) $\varphi$ is strictly monotonic decreasing, $\forall a,b \in [0,1]: a>b \Rightarrow \varphi(a) < \varphi(b)$;\\
(2)  $\varphi(0)=\infty$ and $\varphi(1)=0$;\\
%[TS - Revision R2.14]
(3) $\Phi(P)=[\varphi(p_{ij})]$, $\forall x_i,x_j \in X$ (It is a matrix function).
%[TS- Revision R2.14 end]
\end{definition}

Because $\varphi$ is a real valued function and it is strictly monotonic it is also bijective, therefore the graphs $G_{P}$ and $G_{D}$ are isomorphic via map $\Phi$, with the same set of vertices $X$.
The simplest example of such a function is the map of formula \ref{distance_measures4}.
To better understand the constraints of this isomorphism, below we provide a mathematical analysis with a few simple theorems---the proofs of which, unless otherwise specified, are included in appendix B of the supplementary materials.
Should the reader be interested exclusively on the results, the important formulae for the subsequent sections are eq. \ref{DC1} (distance closure), the isomorphism constraints of eq. \ref{eq_constraint_g}-\ref{eq_constraint_and}, and eq. \ref{distortion} (distortion).

\begin{theorem}
\label{theorem1}
 Let $G_{P}=(X,P)$ be a proximity (symmetric and reflexive) graph and $\Phi$ the graph distance function of definition ~\ref{def1}, then $G_{D}=(X,D)$, where $D=\Phi(P)$, is symmetric and anti-reflexive.
\end{theorem}

%\begin{proof}
%Since $G_{P}$ is reflexive then $p_{x,x}=1$ and from definition ~\ref{def1} we have $d_{x,x}=\varphi(p_{x,x})=\varphi(1)=0$, therefore $G_{D}$ is anti-reflexive. Let $x$ and $y$ be two vertices of $G_{P}$, because a proximity graph is symmetric we have $p_{x,y}=p_{y,x}$, since $\varphi$ is bijective $d_{x,y}=\varphi(p_{x,y})=\varphi(p_{y,x})=d_{y,x}$, therefore $G_{D}$ is symmetric.
%\end{proof}

Next we define the pair of binary operations $\langle f,g \rangle$ of algebraic structure $II$, which operate on distance graphs.

\begin{definition}{\bf{(TD-norms and TD-conorms)}}
 \label{def4a}
Let $f,g:[0,+\infty]\times [0,+\infty]\rightarrow [0,+\infty]$, such that for all $a,b,c \in [0,+\infty]$ the following four axioms are satisfied:\\(1)  $f(a,b)=f(b,a)$, $g(a,b)=g(b,a)$ (commutativity). \\ (2) $f(a,f(b,c))=f(f(a,b),c)$, $g(a,g(b,c))=g(g(a,b),c)$  (associativity).
\\(3) $f(a,b)\leq f(a,c)$, $g(a,b)\leq g(a,c)$, whenever $b \leq c$ (monotonicity).
\\(4) $f(a,\infty) = a$, $g(a,0)=a$, with $a\leq \infty$ (boundary conditions).

\noindent We refer to $g$ as a TD-norm and to $f$ as a TD-conorm.
\end{definition}

%[TS - Revision R2.15]
\begin{theorem}
\label{theorem2}
 If $\varphi$ is a distance function as in definition ~\ref{def1}. For every pair of T-Norm/T-Conorm operations $\langle \land,\lor \rangle$, there exists a pair of operations  $\langle f, g \rangle$ a TD-conorm/TD-norm (definition \ref{def4a}) and vice versa, obtained via the following constraints:\\
\\(1) $\varphi(a \land b) = g(\varphi(a),\varphi(b))$;
\\(2) $\varphi(a \lor b)= f(\varphi(a),\varphi(b))$.
\\
\\ Where $a,b \in [0,1]$.
\end{theorem}

\begin{definition} {\bf{(n-Power of Proximity Graph)}}
 \label{def2}
Let $G_{P}=(X,P)$ be a fuzzy proximity graph. We define the n-power of $P$ as \[P^{n}=\underbrace{P \circ P \circ \cdots \circ P}_{n},\] where the composition of proximity graphs is given by (see also \S \ref{background_fuzzy}):
$$P \circ P = \bigvee_{k} \bigwedge (p_{ik},p_{kj})=p^2_{ij}, \forall x_i,x_j,x_k \in X$$
\end{definition}

\begin{definition} {\bf{(Transitive Closure of Proximity Graph)}}
 \label{def3}
The transitive closure $G_{P}^T(X,P^T)$ of a proximity graph $G_P(X, P)$ is given by:
\[P^{T} = \bigcup_{n=1}^{\kappa} P^n \]
\noindent Where $\cup$ is defined by the same T-Conorm used to produce each n-power. In the most general case, $\kappa \rightarrow \infty$ \cite{Gondran2007}, but with reasonable constrains (see below), the transitive closure of a finite proximity graph converges for a finite $\kappa$ (see also \S \ref{background_fuzzy}).
\end{definition}

Next we focus on distance graphs and algebraic structure $II=\{[0,+\infty],f,g\}$.

\begin{definition}{\bf{(n-Power of Distance Graph)}}
 \label{def4}
Let $G_{D}=(X,D)$ be a distance graph. We define the n-power of $D$ as \[D^{n} =\underbrace{D\circ D\circ \cdots \circ D}_{n} \]  where the composition of distance graphs is given by:
 $$D \circ D = \mathop{f}\limits_{k} \, g (d_{ik},d_{kj})=d^2_{ij}, \forall x_i,x_j,x_k \in X$$
\noindent where  $\langle f, g \rangle$ are a TD-conorm/TD-norm pair per definition \ref{def4a}.
\end{definition}

%\begin{definition}{\bf{(Distance Closure)}}
% \label{def6b}
%The distance closure $G_{D}^T(X,D^T)$ of a distance graph $G_D(X, D)$ is given by:
%\[D^{T} = \dot \cap_{n=1}^{\kappa \rightarrow \infty} D^n \]
%\noindent where $\dot \cap$ is defined by the same TD-Conorm $f$ used to produce each n-power of the distance graph.
%\end{definition}

\begin{definition}{\bf{(Distance Closure)}}
 \label{def6b}
The distance closure $G_{D}^T(X,D^T)$ of a distance graph $G_D(X, D)$ is given by:
\begin{equation}
\label{DC1}
%R^{T}=R \cup R^2 \cup \dots \cup R^k
D^{T} = \dot \cap_{n=1}^{\kappa \rightarrow \infty} D^n
\end{equation}
%\[D^{T} = \dot \cap_{n=1}^{\kappa \rightarrow \infty} D^n \]
\noindent where $\dot \cap$ is defined by the same TD-Conorm $f$ used to produce each n-power of the distance graph.
\end{definition}

%%%%%%%%%%%%%%%%%%%%%%%%%%%%%%%%%%%%%%%%%%%%%%%%%%%%%%%
%Important Theorems

\begin{theorem}
\label{theorem5}
If $G_{P}=(X,P)$ is a fuzzy proximity graph and $G_{D}=(X,D)$ is the distance graph obtained from $G_P$ via $D=\Phi(P)$, where $\Phi$ is the isomorphism (distance function) in definition ~\ref{def1}, then the following statements are true:

1) $\Phi (P)\dot \supseteq \Phi (P^{2} )\dot \supseteq \Phi (P^{3} )\dot \supseteq \cdots \supseteq \Phi (P^{\infty } )$ ;

2) $D\dot \supseteq D^{2} \dot \supseteq D^{3} \dot \supseteq \cdots \dot \supseteq D^{\infty } $.
\end{theorem}

\noindent where $\Phi(P^{n})\dot \supseteq \Phi(P^{n+1})$ means that: $\forall x_i,x_j \in X: \varphi(p^{n}_{ij}) \geq \varphi(p^{n+1}_{ij})$, and $D^{n} \dot \supseteq D^{n+1}$ means that: $\forall x_i,x_j \in X: d^{n}_{ij} \geq d^{n+1}_{ij}$.

Proof in appendix B.

\begin{theorem}
\label{theorem11}
Given a proximity graph $G_{P}=(X,P)$, a distance graph $G_{D}=(X,D)$, and the isomorphism $\varphi$ and $\Phi$ of definition  \ref{def1},  for any algebraic structure $I = ([0,1], \land, \lor )$ with a T-Conorm/T-Norm pair $\langle \land, \lor \rangle$ used to compute the transitive closure of $P$, there exists an algebraic structure $II = ([0,+\infty], f, g )$ with a TD-conorm/TD-norm pair $\langle f, g \rangle$ to compute the isomorphic distance closure of $D$, $D^{T} = \Phi(P^{T})$, which obeys the condition:

\[\forall x_i,x_j,x_k \in X:\mathop{f}\limits_{k} (g(\varphi (p_{ik}),\varphi (p_{kj})) ) = \varphi(\mathop{\lor}\limits_{k} ( (p_{ik}\land p_{kj})))\]

\noindent and vice-versa if we fix $\langle f, g \rangle$ (TD-norm/TD-Conorm) and isomorphism $\varphi$, to obtain $\langle \lor, \land \rangle$:

\[ \forall x_i,x_j,x_k \in X: \mathop{\lor}\limits_{k}( \varphi^{-1} (d_{ik}) \land \varphi^{-1} (d_{kj})) = \varphi^{-1} ( \mathop{f}\limits_{k} (g(d_{ik},d_{kj})) ) \]

\noindent where $\varphi ^{-1}$ is the inverse function of $\varphi$.
\end{theorem}

The conditions of this theorem lead to the following constraint equations that isomorphism $\varphi$ enforces on algebraic structures $I$ and $II$ (as shown in the proof for theorem \ref{theorem11} in appendix B):

\begin{equation}
\label{eq_constraint_g}
g(d_{ik},d_{kj})=\varphi(\varphi^{-1}(d_{ik})\land \varphi^{-1}(d_{kj}))
\end{equation}

\begin{equation}
\label{eq_constraint_f}
f(d_{ik},d_{kj}) \equiv \varphi ( \varphi^{-1}(d_{ik})\lor \varphi^{-1}(d_{kj}))
\end{equation}

%From these last equations we can also find $\lor$ and $\land$ given $f$, $g$ and the isomorphism $\varphi$:

\begin{equation}
\label{eq_constraint_or}
p_{ik}\lor p_{ki}=\varphi^{-1}(f(\varphi(p_{ik}),\varphi(p_{ki})))
\end{equation}
\begin{equation}
\label{eq_constraint_and}
p_{ik}\land p_{kj}=\varphi^{-1}(g(\varphi(p_{ik}),\varphi(p_{ki})))
\end{equation}

%kol: the numbers of the theorems need to be worked out. The theorem just above (currently 6), should be Theorem 9---you are using the same tex counter for both which causes problems also! Moreover the Theorems 9,10,11 in appendix should be Theorems 6, 7, and 8 (and remain in appendix). Also, the proof of current theorem 6 appears again with its proof as Theorem 12 in appendix . This makes no sense. The proof should be in appendix but they need to have the same number!  The order of theorem should be the logical order, not the position order in the Latex. Use: http://www.personal.ceu.hu/tex/counters.htm

Since many possible transitive (distance) closures are possible, it is important to measure how much a closure defined by a given T-Norm/T-Conorm pair $\langle \land, \lor \rangle$ (or TD-conorm/TD-norm pair $\langle f,g \rangle$) distorts the original proximity (distance) graph in the isomorphism space of Theorem \ref{theorem11}.
We define \emph{distortion}, $\Delta$, as the sum of the differences between the edges in the original graph and the edges obtained by a given closure.

\begin{equation}
\label{distortion}
\Delta(P)=\sum_i\sum_j|p_{ij}^T-p_{ij}|
\end{equation}

%\noindent which can be normalized as $\frac{2.\Delta(P)}{|X|^2-|X|}$.

Theorem \ref{theorem11} specifies the isomorphism constraint on $\langle f, g \rangle$ given $\langle \lor, \land \rangle$, and $\varphi$, or, alternatively, the constraint on $\langle \lor, \land \rangle$ given $\langle f, g \rangle$, and $\varphi$. This allows us to study several closure scenarios, which lead to different distortions of the original graphs.
%
%Understanding and mapping this constraint, is necessary in order to analyze models of complex networks based on weighted graphs. Any conclusions derived from such models should take into account the distortions one imposes on graph topology when converting proximity/strength into distance graphs to compute path length and shortest path measures.
%
Given this space of possible transitivity criteria, it is reasonable
to ask several questions: for a given proximity-to-distance
isomorphism $\varphi$, what is the equivalent of the (fuzzy)
$(max,min)$ transitive closure for a distance graph? Perhaps more
interestingly, what is the proximity equivalent of the metric closure of
a distance graph, which is ubiquitous in network science as the APSP/Dijkstra algorithm? Which closures preserve important characteristics
of real complex networks and observe good axiomatic requirements?
These questions are important because all the applications of complex networks that use transitivity produce different results depending on the specific T-Norm/T-Conorm pair $\langle \lor, \land \rangle$ used.
Not only do we want intuitive connectives (e.g. a metric closure), we want those that lead to best results in
specific applications. In the following sections (\S \ref{max_closures}, \S \ref{diffusion_section}) we study in detail the specific closure cases that arise from constraining algebraic structures $I$ or $II$ in different ways.
But before that, in the next subsection we discuss additional constraints on algebraic structures $I$ and $II$ which allow the computation of closures in finite time.

\subsection{Convergence of Distance Closures}
\label{convergence}

As defined above (\S \ref{isomorphism_definition}), the transitive closure of proximity graphs utilizes the algebraic structure $I=\{[0,1],\lor,\land\}$ where $\lor,\land$ are, respectively, T-Conorm and T-Norm binary operations, whereas the distance closure of distance graphs utilizes the algebraic structure $II=\{[0,+\infty],f,g\}$ where $f,g: [0,+\infty] \times [0,+\infty] \rightarrow [0,+\infty]$ are TD-Norm and TD-Conorm binary operations.
It has been known for a while \cite{Klir1995} that if the T-Conorm in $I$ is $\lor = \texttt{maximum}$, with any T-Norm $\wedge$, then the transitive closure of a finite graph converges for a finite $\kappa$ in equation \ref{TC1} or Definition \ref{def3} (\S \ref{background_fuzzy} and \S \ref{isomorphism_definition}), moreover, $\kappa \leq |X|-1$ is the diameter of the graph.
In other words, the transitive closure converges in finite time and can be easily computed using Algorithm 1 defined in appendix A.

In the last decade, the convergence requirements of transitive closure using algebraic structure $I$ have received much attention \cite{Han2007,Gondran2007,Han2004,Dombi2013,Bertoluzza2004,Pang2003}.
It is now known that if $I$ is a dioid, then $\kappa$ is also finite \cite{Gondran2007}.
A diod is a special case of semiring, where, in addition to $\{[0,1], \land\}$ and $\{[0,1], \lor \}$ being monoids (see \S \ref{background_fuzzy}), the T-Conorm/T-Norm pair in $I$ also needs to satisfy the distributive property (in addition to the monotonicity requirements of the T-Norm and T-Conorm monoids).
Not all pairs of T-Conorms/T-Norms satisfy the distributive property \cite{Han2007,Gondran2007,Han2004,Dombi2013,Bertoluzza2004,Pang2003}. However, there is an infinite variety of dioids that do (see \cite{Gondran2007} for an overview), and therefore, an infinite variety of distinct transitive closures that can be computed in finite time.

%[TS - Revision R2.17]
\begin{theorem}
\label{theorem3a}
Given a finite proximity graph $G_P(X,P)$, and an algebraic structure $I=\{[0,1], \lor, \land \}$, with a T-Conorm/T-Norm pair $\langle \land, \lor \rangle$ used to compute the transitive closure of $G_P$, if $I$ is a dioid, then the transitive closure $G_P^T(X,P^T)$ can be computed by equation \ref{TC1} for a finite $\kappa$.
\end{theorem}

See \cite{Gondran2007} for proof; further discussion and examples also see \cite{Han2004, Han2007,Pang2003,Klir1995}.

\begin{theorem}
\label{theorem3c}
Given a finite distance graph $G_D(X,D)$, and an algebraic structure $II=\{[0,+\infty],f,g \}$, with a TD-Conorm/TD-Norm pair $\langle f, g \rangle$ used to compute the distance closure of $G_D$, if $II$ is a dioid, then the distance closure $G_D^T(X,D^T)$ can be computed in finite time via the transitive closure of isomorphic graph $G_P(X,P)$ with algebraic structure $I$ obtained by an isomorphism satisfying Theorem \ref{theorem11}. In other words, if $II$ is a dioid, via an isomorphism satisfying Theorem \ref{theorem11} we obtain an algebraic structure $I$ which is also a dioid.
\end{theorem}

This theorem is easily proven from theorems \ref{theorem5}, \ref{theorem11} and \ref{theorem3a}, by evoking the isomorphism to proximity space.

\section{Shortest-path ($\lor = \texttt{maximum}$) closures}
\label{max_closures}

Of particular interest to current work on complex networks, is the relationship between the metric closure of distance graphs computed via the APSP/Dijkstra algorithm (see \S \ref{semimetric_pairs_section}), and some transitive closure of (fuzzy) proximity graphs, which has not been previously identified.
Furthermore, it is also very worthwhile to understand what other forms of closure exist and are meaningful for complex network analysis.
The general isomorphism (theorem \ref{theorem11} and the constraints of formulae \ref{eq_constraint_g} to \ref{eq_constraint_and}) presented in section \ref{isomorph_section} gives us the ability to identify all these forms of distance closure, and thus, the distinct measures of path length that ensue. We can also study their convergence and axiomatic characteristics.

\subsection{Metric Closure}
\label{metric_example}

Let us start with the metric closure. As described in section \ref{semimetric_pairs_section}, this distance closure can be computed  with pair $\langle f = min, g=+ \rangle$, using eq. \ref{DC1} from definition \ref{def6b} (\S \ref{isomorphism_definition}).
It yields a distance graph $D^{mc}(X,X)$, whose edges $d^{mc}_{ij}$ between any two elements $x_i$ and $x_j$ are defined by the shortest (direct or indirect) distance between them in the original distance graph $D(X,X)$.
In other words, it computes the shortest ($f=\texttt{minimum}$) paths between any pair of elements in the original graph, where path length is computed by summing ($g = +$) the distance weights of every edge in path.

\vspace{5 mm}

%\begin{example}
\emph{\bf{Example 1 (Metric Closure)}}
\label{ex_metric_closure}
Let $\varphi:[0,1]\rightarrow [0,+\infty]$, $d_{ij} = \varphi(p_{ij}) = \frac{1}{p_{ij}}-1$ (as in equation \ref{distance_measures4}, \S \ref{semimetric_background}). Let also $f(x,y) = \min(x,y)$  and $g(x,y) = x + y$, where $x, y \in [0, +\infty]$ represent distance weights from algebraic structure $II$ (see \S \ref{isomorph_section}).
From theorem \ref{theorem11}, eq. \ref{eq_constraint_or}:

\[a\lor b=\varphi^{-1}(f(\varphi(a),\varphi(b)))\]

\noindent where $a,b \in [0,1]$ represent proximity weights from semi-ring $I$ (see $\S$ \ref{isomorph_section}), $a=\varphi^{-1}(x)$ and $b=\varphi^{-1}(y)$.
If $a\leq b$, without loss of generality, then

\[a\lor b=\varphi^{-1}(min(\varphi(a),\varphi(b)))=b=max(a,b)\]

\noindent since $\varphi$ is strictly monotonic decreasing. Therefore,

\[a\lor b= \max(a,b)\]

\noindent To obtain $\land$ we use eq. \ref{eq_constraint_and}:

\[a\land b=\varphi{-1}(g(\varphi(a),\varphi(b)))\]

\[a\land b=\varphi^{-1}(\varphi(a)+\varphi(b))=\varphi^{-1}(\frac{a+b-2ab}{ab})\]
\noindent and since $\varphi^{-1}(x)=\frac{1}{x+1}$ we obtain,

\[a\land b= \left\{ \begin{array}{ccc}
0 & for & a = b = 0\\
\frac{ab}{a+b-ab} & for & a,b \in ]0,1] \end{array} \right.
\]

This conjunction is very well-known in fuzzy graph theory; it is the Dombi family of T-Norms for $\lambda=1$ \cite{dombi82}, which we denote by $DT_{\land}^{1}$---this makes sense since the isomorphism map $\varphi$ used in example 1 (equation \ref{distance_measures4}, \S \ref{semimetric_background}) is also the Dombi T-Norm generator with $\lambda=1$
%---which is the same as the Hammacher family conjunction for $p=0$
(see \cite{Klir1995} and also section \ref{other_closures}).
Therefore, the distance closure of a distance graph with algebraic structure $II$ where  $\langle f,g \rangle \equiv \langle \min, + \rangle$, is isomorphic to the transitive closure with algebraic structure $I$ where $\langle \lor, \land \rangle \equiv \langle \max, DT_{\land}^{1} \rangle$ in the proximity space.
Moreover, because $\lor = \max$, the closure (in both spaces) converges in finite time (\S \ref{convergence})---the same is true for all examples covered in this section.
Indeed, this is the metric closure of distance graphs, also known as the APSP and typically computed using the Dijkstra algorithm \cite{dijkstra} or the distance product \cite{Zwick,Rocha2002,Rocha2005}.
%\end{example}
%

\subsection{Generalized Metric Closure and Shortest Path Length with APSP/Dijkstra}
\label{other_closures}

One way to explore the isomorphism space is to fix $f\equiv \min$ and $g \equiv +$, and let the isomorphism function $\varphi$ vary.
In the proximity space, this means that $\lor \equiv \max$, as shown in example 1 (\S \ref{metric_example}), which guarantees convergence of the transitive closure in finite time.
%, and which can be computed using the very common APSP/Dijkstra algorithm or the metric product.
%
However, because we vary the isomorphism map $\varphi$, as we show below, we are effectively sweeping the space of possible T-Norm ($\land$) operations, since $\varphi$ is their generator function.
In other words, by varying $\varphi$, we can use the canonical metric closure (computed via APSP/Dijkstra) to sweep an infinite space of possible distance closures.

\begin{definition}
\label{def8}
The pseudo-inverse of a decreasing generator $\varphi$ is defined by
\[
\varphi^{(-1)}(a)= \left\{ \begin{array}{ccc}
1 & for & a \in (-\infty,0) \\
\varphi^{-1}(a) & for & a \in [0,\varphi(0)] \\
0 & for & a \in (\varphi(0),\infty) \end{array} \right.
\]
\end{definition}

\begin{theorem}
\emph{(Characterization Theorem of T-Norms)}
\label{characterization}
Let $\land$ be a binary operation on the unit interval. Then, $\land$ is an Archimedean T-Norm iff there exist a decreasing generator $\varphi$ such \[a\land b=\varphi^{(-1)}(\varphi(a)+\varphi(b))\] for all $a,b \in [0,1]$.
\end{theorem}

Both definition \ref{def8} and the proof of theorem \ref{characterization} are provided in \cite{Klir1995}. The next corollary (proof in appendix B) follows from theorem \ref{theorem11}.

\begin{corollary}
\label{theorem12}
Given the isomorphism constraint on the T-Norm from algebraic structure $I$ (eq. \ref{eq_constraint_and}) from theorem \ref{theorem11}, let  $f \equiv \min$,  $g \equiv +$  and  $\varphi$  a distance function per Definition  \ref{def1} (\S \ref{isomorphism_definition}). If  $\lor \equiv \max$  as T-Conorm, then the T-Norm operator $\land$ exists and $\varphi$ is its generator function.
\end{corollary}

Corollary ~\ref{theorem12} states that when we fix the T-Conorm $\lor=\max$ and $\langle f,g \rangle = \langle \min, + \rangle$, there exists a T-Norm $\land$, which preserves the isomorphism between proximity and distance graphs, as well as their closures with the respective operators. Moreover, the isomorphism function $\varphi$ is in fact the T-Norm generator.
Thus, as we fix the TD-Norm and TD-Conorm $\langle f, g \rangle = \langle \min, + \rangle$ which define the metric closure of distance graphs, we can vary the isomorphism $\varphi$ yielding distinct transitive closures in the proximity space which is thus constrained to the $\lor=\max$ T-Conorm, and the T-Norms $\land$ generated by $\varphi$ (using theorem \ref{characterization}).
This generalizes the metric closure as we are free to sweep the space of T-Norm generator functions $\varphi$ that satisfy definition \ref{def1}.
Importantly, we can do this using the very common algorithms developed for APSP, such as the Dijkstra algorithm \cite{dijkstra} or the distance product \cite{Zwick,Rocha2002,Rocha2005}, because the operators $\langle f,g \rangle = \langle \min, + \rangle$ are fixed in distance space.

We can think of this space of generalized metric closures as the different ways we have to compute (shortest) path length in distance graphs. The canonical metric closure, obtained via the simplest map $\varphi$ given by eq. \ref{distance_measures4}, computes path length as the sum ($g \equiv +$) of all edges in the path. As we vary $\varphi$, we can compute an infinite set of different measures of path length (e.g. the ultra-metric closure in subsection \ref{ultra_metric_example} below). %
Still, because $f \equiv min$ for all these cases, we are always computing the \emph{shortest} of some kind of path length---choosing the minimum path. Every possible closure results from choosing a shortest path; what changes is how path length is computed.
Thus, we refer to this class of generalized metric closures as \emph{shortest-path distance closures}. Moreover, since via the isomorphism we obtain $\lor=\max$, these distance closures are guaranteed to converge in finite time, just like their isomorphic transitive closures in proximity space (\S \ref{convergence}).
Notice that closures which do not fix $f \equiv min$ and $\lor=\max$, integrate path lengths in other ways other than the shortest path. Indeed, we study the different case of diffusion distance closure in section \ref{diffusion_section}

Since different measures of path length can be computed via the generalized metric closure, we can investigate, for instance, the desirable variation of shortest path length.
For the empirical analysis of complex networks it is desirable that properties such as \emph{average shortest path} be simultaneously characteristic in both spaces (proximity and distance), for each distance closure chosen. That is, the fluctuations of the mean should behave similarly in both spaces (average shortest path length in distance graphs and average strongest path in proximity graphs).
We estimated the variation of the shortest path distribution when the isomorphism $\varphi$ is parameterized by the Dombi family of T-Norm generators, controlled by a parameter $\lambda$ \cite{dombi82}.
The details of this estimation are provided in Appendix C of the supplementary materials (see also section \ref{diffusion_section} for the Dombi family T-Norm/T-Conorm formulae).
We concluded that when we assume a small variation of the mean for the distribution of shortest path length in distance graphs, the optimal distance closure to preserve small variations of path strength in the proximity space is the metric closure (example 1, $\S$ \ref{metric_example}), where $\lambda = 1$, thus $\land = DT^{1}_{\land}$ .
%
%value of $\lambda = 1$, thus $\land = DT^{1}_{\land}$ leading to the case of the metric closure .
%
However, when variation is allowed to increase, the optimal value closure occurs for other closures with T-Norms with $\lambda > 1$.
This suggests that the metric closure typically computed in network science (using the APSP/Dijkstra) is very appropriate if we assume small variation in the distribution of shortest paths. If, instead, we observe larger fluctuations of that distribution, it may be more appropriate to employ distance closures isomorphic to the transitive closure obtained via a Dombi T-Norm with $\lambda > 1$.

\subsection{Ultra-Metric Closure}
\label{ultra_metric_example}

In Fuzzy logic/set theory T-Conorm/T-Norm which obey a generalization of De Morgan's laws with an involutive complement are called \emph{dual} (\S \ref{background_fuzzy}).
See appendix A for more details about T-Conorm/T-Norm pairs and the dual property; also we develop this concept further in section \ref{axiomatics}.
Within the entire space of shortest-path distance closures, where $\lor \equiv max$, the only dual pair of T-Conorm/T-Norms
%that also obeys the distributive property necessary to have a finite transitive closure (to make algebraic structure $I$ a dioid, \S \ref{convergence})
is $\langle \lor \equiv max, \land \equiv min \rangle$ \cite{klement}.
In other words, the only shortest-path distance closure which is based on a conjunction/disjunction pair that establishes a logic with the reasonable and expected logical axioms of De Morgan's laws is the ultra-metric closure we describe next.
Thus, the metric closure (example 1) is based on an algebraic structure that is too poor to define a reasonable logic.
\vspace{5 mm}

%\begin{example}
\emph{\bf{Example 2 (Ultra-Metric Closure)}}
\label{ex4.2}
Let $\varphi:[0,1]\rightarrow [0,+\infty]$, $d_{ij} = \varphi(p_{ij})$ be any function that obeys the axioms of definition \ref{def1}. Let also $\land (a,b) = \min (a,b)$  and $\lor (a,b) = \max (a,b)$, where $a,b \in [0,1]$ represent proximity weights from semi-ring $I$  (\S \ref{isomorph_section}).
Following the same reasoning as with example 1, via the constraints of the isomorphism (theorem \ref{theorem11} and the fact that $\varphi$ is monotonic decreasing per definition \ref{def1}), it is easy to show that:

\[f(x, y) \equiv \min(x,y)\]

\noindent and

\[g(x,y) \equiv \max(x, y)\]

\noindent where $x,y \in [0, +\infty]$ represent distance weights from semi-ring $II$ (\S \ref{isomorph_section}).
Therefore, the distance closure of a distance graph with algebraic structure $II$ where  $\langle f,g \rangle \equiv \langle \min, \max \rangle$, is isomorphic to the transitive closure with algebraic structure $I$ where $\langle \lor, \land \rangle \equiv \langle \max, \min \rangle$ in the proximity space---the most common transitive closure in fuzzy graphs (\S \ref{background_fuzzy}), based on a dual T-Conorm/T-Norm pair.

The $\langle \lor, \land \rangle \equiv \langle \max, \min \rangle$ closure of a fuzzy graph is equivalent to the \emph{ultra-metric} closure of a distance graph, where instead of the triangle inequality, a stronger inequality is enforced: $d_{ij} \leq \max (d_{ik}, d_{kj}), \forall k$.
Ding et al \cite{ding} have previously shown this simple relationship, which derives easily for any $\varphi$ (per definition \ref{def1}) in our framework. %
Ding et al further used this closure to compute cliques in protein interaction networks---a complex network problem relevant for computational Biology.

Because in this case $\lor = \max$, the ultra-metric closure is still a shortest-path distance closure (\S \ref{other_closures}),  and therefore converges (in both spaces) in finite time (\S \ref{convergence}).
However, in the ultra-metric closure, instead of path length being computed by summing the edges in a path (as the canonical metric closure, \S \ref{metric_example}), path length is measured exclusively by the ``weakest link'' in the path: the largest distance edge-weight or the smallest proximity edge-weight, of distance or proximity graphs, respectively.

\begin{figure}[!th]
\centerline{\includegraphics[width=1\textwidth]{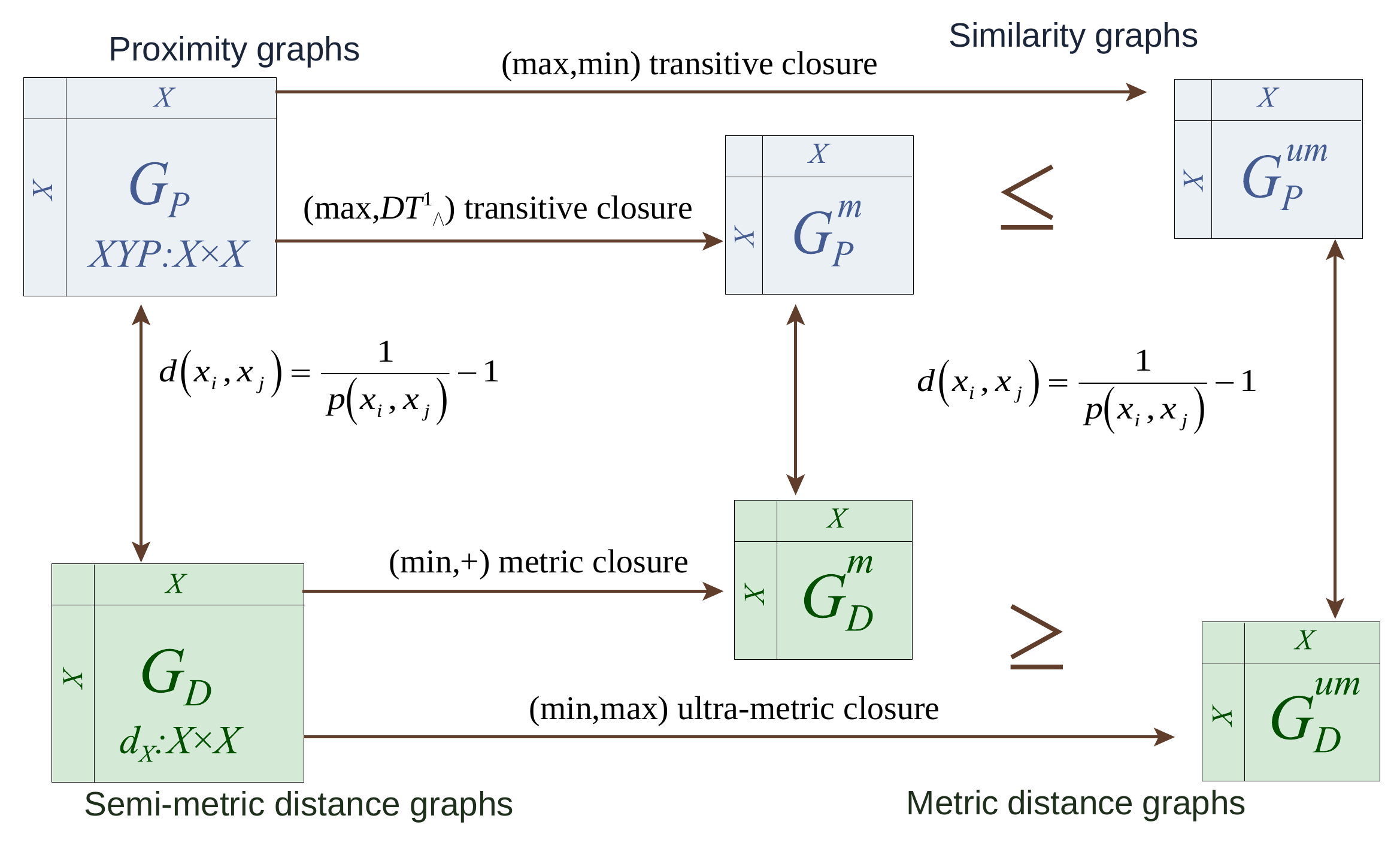}}
\caption{\small Metric and ultra-metric distance closures, and their
fuzzy proximity graph counterparts for $\varphi:
distance=\frac{1}{proximity}-1$. The ultra-metric distance closure
is equivalent to the $\langle \max, \min \rangle$ transitive closure of a fuzzy graph. The
metric closure is equivalent to the $\langle \max,DT_{\wedge}^{1} \rangle$ closure of a
fuzzy graph, where $DT_{\wedge}^{1}$ is the Dombi conjunction for $\lambda=1$
%\cite{Klir1995}
.} \label{semimetric_closure_space}
\end{figure}

Figure \ref{semimetric_closure_space} depicts the closures of examples 1 and 2 above, for the
proximity-to-distance isomorphism $\varphi$ of formulae \ref{distance_measures4}.
The or ultra-metric closure
 of distance graphs (or $\langle \max, \min \rangle$ closure of proximity graphs) imposes quite a strong distortion of the original graph.
After closure, every item tends to become highly
related to every other indirectly linked item, however many edges far away.
When using it to infer indirect relationships (shortest paths, cliques, clusters), the assumption is that the strength or proximity of connection between any two items is equal only to the weakest edge on the path between both items---irrespectively of how many edges that path may be comprised of.
For instance, in a social network, any two people are as strongly connected as the weakest social connection in the chain of indirect social connections that links them, with no penalty for the number of indirect connections that exist.
A catholic who is very close to a priest who is close to a bishop who is close to a cardinal who is close to the Pope, becomes automatically close to the Pope---a scenario that runs against our perception of the reality of that social connection.

This intuition is also validated in more testable scenarios in information retrieval applications.
We have observed in our work with recommender systems \cite{rocha02WCCI,Rocha2005,Simas2012}, as well as in our analysis of social and knowledge
networks \cite{Rocha2002,rocha02Terr,verspoor05BMC,abihaidar_GB08}, that the metric closure, $\langle \lor, \land \rangle \equiv \langle \max, DT_{\land}^{1} \rangle$, produces better and more intuitive results than the ultra-metric closure, $\langle \lor, \land \rangle \equiv \langle \max, \min \rangle$---insofar as the search for relevant indirect associations is concerned.

In the metric closure case, because we sum the distance weight of every edge in a path ($g \equiv +$), there is a built-in penalty for the number of indirect edges in the path.
This matches our intuition that, in reality, the catholic in our example will have a harder time influencing the Pope if the communication chain involves a hierarchy of many levels, no matter how strong each connection between levels is.
This means that the metric closure results in significantly fewer edges being altered in the original graph; only those indirect paths comprised of a few edges, with every distance edge-weight relatively small, may provide a shorter indirect connection than the original direct connection.
In other words, the metric closure imposes a weaker distortion of the original graph.
Theorem \ref{theorem9} below (proof in appendix B) shows that the ultra-metric distance closure always leads to a larger distortion of the original graph, than what we get from the metric closure of the same graph: $\Delta_{um} \geq \Delta_{mc}$. These results are also depicted in Figure \ref{semimetric_closure_space}.

\begin{theorem}
\label{theorem9}
Given the isomorphism $\varphi $ from definition \ref{def1} and theorem \ref{theorem11}, if $D^{mc}$ is the metric closure with $f \equiv \min$ and $g_{1}\equiv +$, and $D^{um}$ is the ultra-metric closure with $f\equiv \min$ and $g_{2}\equiv \max$ then $D^{mc}\dot \supseteq D^{um}$ is equivalent to $P^{mc}\subseteq P^{um}$, where $D^{mc}=\Phi(P^{mc})$ and $D^{um}=\Phi(P^{um})$. Therefore,  $\Delta(P^{um}) \geq \Delta(P^{mc})$, where distortion is computed using eq. \ref{distortion}.
\end{theorem}

%%%%%

We can see from theorem ~\ref{theorem11} and corollary ~\ref{theorem12} that the transitive closure of proximity graphs and the isomorphic distance closure of distance graphs, entails a very wide space of possibilities, which include the metric and ultra-metric closures. In the generalized metric closure case, each variant implies a distinct way of computing shortest path lengths---as well as assumptions about constraints on the variation of the distribution of shortest paths (\S \ref{other_closures}).
Consequently these closures are not unique as already known in the theory of fuzzy graphs, but not so well-known in the field of network science. For a given application, it is important to pay attention to the distortion created by the distance or transitive closure computed on the original relational information extracted from data.
Next we look at forms of distance closure which step outside the notion of shortest path, and search for distance closures with good axiomatic characteristics from a logical point of view, but which are intuitively close to the canonical metric closure.

\section{Diffusion Distance (Dombi Transitive) Closure}
\label{diffusion_section}

\subsection{Axiomatics of Distance Closure and Network Approximate Reasoning}
\label{axiomatics}

In the Fuzzy logic community, considerable work has been done to
identify pairs of operations and complements that satisfy
desirable axiomatic characteristics (e.g. De Morgan's laws
\cite{dombi82}). These pairs of general (fuzzy) logic conjunction
and disjunction operations are known as conjugate or \emph{dual} T-Norms and
T-Conorms \cite{Klir1995,Klement2004}. As
discussed above, each distinct
conjunction/disjunction pair leads to a specific transitive closure
of an initial proximity graph---with isomorphic distance closures.  However,
only some of these entail intuitive logical operations.
To pursue logical reasoning, it is reasonable to expect a complement
to be \emph{involutive}, so that $\bar{\bar{x}}=x$. It is also reasonable
that disjunction, conjunction and complement follow De Morgan's
laws: $\overline{a \vee b} = \bar{a} \wedge \bar{b}, \, \overline{a \wedge b} =
\bar{a} \vee \bar{b}$.
For instance, the $\langle \lor, \land \rangle \equiv \langle \max, \min \rangle$ operations, with
the standard fuzzy complement ($\bar{x}=1-x$), follow De Morgan's
laws. So do many other operations and complements, see
\cite{Klir1995} for a good overview.

Such desirable logical axiomatic constraints are also
important for graphs, especially when we use them to model
knowledge networks.
Since, as we have shown (\S \ref{knowledge_networks_section}),
proximity networks are good knowledge representations for many
applications, it is desirable to be able to combine or fuse networks obtained from different data sources and compute compound logical statements from the knowledge they store.
%
%when we use proximity networks as knowledge representations
%(see \S \ref{knowledge_networks_section}),
%%in a given application,
%it may be useful to have an intuitive understanding of what is the
%complement of a given network, or better, be able to compute the
%conjunction and disjunctions of various networks obtained from
%distinct data sources.
%
For instance, in the recommender system
developed for MyLibrary@LANL \cite{Rocha2005}, it may be useful
to issue recommendations on an aggregate journal network built from a
conjunction of two constituent networks (e.g. journal proximity
obtained from user access data \emph{and} journal proximity obtained
from citation data).
This calculus of fuzzy graphs \cite{Zadeh1999} allows us to perform a \emph{network approximate reasoning}, the development of which is beyond the scope of this article, but necessarily requires that algebraic structures $I$ and $II$ in our isomorphism (\S \ref{isomorphism_definition}) possess the reasonable (duality) constraints outlined above.

The metric closure of a distance graph corresponds to the transitive closure with algebraic structure $I$ where $\langle \lor, \land \rangle \equiv \langle \max, DT_{\land}^{1} \rangle$ in the proximity space (\S \ref{metric_example}).
As shown in example 1, the T-Norm $\land$ associated with this closure is a special case of the Dombi \cite{dombi82} conjunction \cite{Klir1995}:

\begin{equation}
\label{eq_dombi_and_lambda}
DT^{\lambda}_{\wedge}(a,b)=\frac{1}{1+\left[\left(\frac{1}{a}-1\right)^{\lambda}
+ \left(\frac{1}{b}-1\right)^{\lambda}\right]^{\frac{1}{\lambda}}} \quad \forall a, b \in [0,1], \; \lambda \in [0, +\infty]
\end{equation}

\noindent which, when $\lambda = 1$ becomes:

\begin{equation}
\label{eq_dombi_and_1}
DT^{1}_{\wedge}(a,b)=\frac{ab}{a+b-ab} \quad \forall a, b \in [0,1]
\end{equation}

Unfortunately, the T-Conorm/T-Norm pair $\langle \max, DT_{\land}^{1} \rangle$ used on the metric closure leads to an algebraic structure $I$ with very poor
axiomatic characteristics for logical reasoning. It can be easily shown that this pair of
operations does not satisfy De Morgan's laws for any involutive complement (see theorem \ref{complement} in appendix B).
Since no involutive complement exists that can
satisfy De Morgan's laws for this pair, we now ask what
is the $\langle \lor,\land \rangle$ pair closest to it, that with an involutive
complement obeys De Morgan's laws?

In section \ref{max_closures},  we fixed the T-Conorm $\lor \equiv \max$, and varied the isomorphism map $\varphi$, which is the same as varying the space of possible T-Norms $\land$ in proximity space. This led to the generalized metric closure, the entire space of which (shortest-path distance closures) contains a single dual T-Conorm/T-Norm pair: the ultra-metric closure, $\langle \lor, \land \rangle \equiv \langle \max, \min \rangle$.
In distance space, this means that we fixed the concept of shortest path ($f \equiv \min$), generalizing the computation of path length (via different $g$ binary operations).

Here, also starting from the metric closure $\langle \lor, \land \rangle \equiv \langle \max, DT_{\land}^{1} \rangle$, we fix the T-Norm $\land \equiv DT_{\land}^{1}$ and let the T-Conorm $\lor$ vary instead.
In this case, we also fix the isomorphism to the simplest, intuitive and most used $\varphi$ function of equation \ref{distance_measures4}.
In distance space, this means that we preserve the computation of path length to the summation of every edge weight in a path---because with this $\varphi$, fixing $\land \equiv DT_{\land}^{1}$ in proximity space, is equivalent to fixing $g \equiv +$ in distance space (\S \ref{metric_example}).
In the present work, we restrict the search of T-Conorms to the same Dombi family \cite{dombi82,Klir1995}:

\begin{equation}
\label{eq_dombi_or_lambda}
DT^{\lambda}_{\vee}(a,b)=\frac{1}{1+\left[\left(\frac{1}{a}-1\right)^{-\lambda}
+ \left(\frac{1}{b}-1\right)^{-\lambda}\right]^{-\frac{1}{\lambda}}} \quad \forall a, b \in [0,1], \; \lambda \in [0, +\infty]
\end{equation}

\noindent which, when $\lambda = 1$ becomes:

\begin{equation}
\label{eq_dombi_or_1}
DT^{1}_{\vee}(a,b)=\frac{a+b-2ab}{1-ab} \quad \forall a, b \in [0,1]
\end{equation}

\noindent Therefore, in distance space, we are no longer computing the shortest path ($f \equiv \min$), but something else, where $f$ is given by eq. \ref{eq_constraint_f} using $\varphi$ from eq. \ref{distance_measures4} and $\lor$ from eq. \ref{eq_dombi_or_lambda}.

Using the general Dombi T-Conorm equation (\ref{eq_dombi_or_lambda}) we can find a
$\lambda$ which satisfies De-Morgan's laws when paired with the Dombi T-Norm used in the metric closure: $\langle DT^{\lambda}_{\lor},DT^{1}_{\wedge} \rangle$.
We perform this search using the Dombi T-Conorm $DT^{\lambda}_{\lor}$ because it is one of the most well-known parametric T-Norm/T-Conorm families which includes the widest range possible of such operations \cite{dombi82,Klir1995}.
It includes the T-Conorm used in the metric closure (\S \ref{metric_example}), since $DT^{\lambda \rightarrow \infty}_{\lor} (a,b) = \max (a,b)$. Therefore, we can investigate the properties of metric closure in this formulation.

Because the Dombi family of T-Norms and T-Conorms is dual when the same $\lambda$ is used for the T-Norm and T-Conorm \cite{dombi82, Klir1995}, $\lambda = 1$ obviously yields a dual pair $\langle DT^{1}_{\lor},DT^{1}_{\wedge} \rangle$ with the characteristics we seek. This pair is used to compute a \emph{diffusion distance closure} studied in detail below (\S \ref{diffusion_closure_subsection}).
However, is this pair the closest to the metric closure in this formulation? Can we find those values of $\lambda$ near satisfying the requisite laws of logic? How far is the metric closure from satisfying these laws?

Notice that the T-Conorm/T-Norm pair used by the ultra-metric closure (\S \ref{ultra_metric_example}) is not included in this search because it does not share the T-Norm  $\land = DT_{\land}^{1}$.
In any case, because the ultra-metric closure is based on the dual T-Norm/T-Conorm pair $\langle \lor,\land \rangle = \langle \max,\min \rangle$, it is already based on an algebraic structure $I$ with all the desirable axiomatic characteristics---the only one in the set of shortest-path distance closures (\S \ref{other_closures}).
%
%We are now perform an alternative search also starting from the metric closure, but fixing the T-Norm to $\land \equiv DT^{1}_{\wedge}$ and varying the T-Conorm with the general Dombi T-Conorm $DT^{\lambda}_{\lor}$ of formula \ref{eq_dombi_or_lambda}.

Let us then investigate if the De-Morgan's Laws, with the standard complement $C_{1}(x)=1-x$, are satisfied for the pair $\langle DT^{\lambda}_{\lor},DT^{1}_{\wedge} \rangle$;

\[ \bar a \lor \bar b | \lor \equiv D_{\lor}^{\lambda}(\overline a, \overline b) = \frac{1}{1+\left[ \left( \frac{1}{a}-1\right)^{\lambda}+\left(\frac{1}{b}-1\right)^{\lambda}\right]^{-\frac{1}{\lambda}}}\]

\[ \overline{a \land b} | \land \equiv \overline D_{\land}^{1}(a,b)  =1 - \frac{ab}{a+b-ab}=\frac{a+b-2ab}{a+b-ab}\]

For De-Morgan's Law to hold,  $\overline{a \land b}=\overline a \lor \overline b$:
\[ -ab\left[\left(\frac{1}{a}-1\right)^{\lambda}
+ \left(\frac{1}{b}-1\right)^{\lambda}\right]^{\frac{1}{\lambda}}+a+b-2ab=0\]

This equation has $\lambda =1$ as a straightforward unique solution which
is not at all surprising;
%
%, because the $\langle DT_{\lor}^{1},DT_{\land}^{1} \rangle$, with the standard complement naturally form a dual T-Norm/T-Conorm pair \cite{Klir1995}.
%
this merely shows that in the Dombi family the unique dual T-Conorm for the $DT_{\land}^{1}$ T-Norm is $DT_{\lor}^{1}$.
While the unique error-free solution is trivial, we can use this equation to understand how far from desirable axiomatic characteristics the pair used in metric closure is.
We can think of the left side of the equation as the \emph{error} or \emph{deviation} from a T-Norm/T-Conorm pair that obeys De-Morgan's laws with standard complement. An
integral of the left side of the above equation yields an estimate of
the total deviation $F(\lambda)$ from ideal axiomatic characteristics over the entire domain of
the operations:

\[F(\lambda)= \int_0^1{\int_0^1 {\left(-xy\left[\left(\frac{1}{x}-1\right)^{\lambda} + \left(\frac{1}{y}-1\right)^{\lambda}\right]^{\frac{1}{\lambda}} + x +y -2xy\right)}dx}dy\]

\begin{figure}[!th]
\centerline{\includegraphics[width=0.8\textwidth]{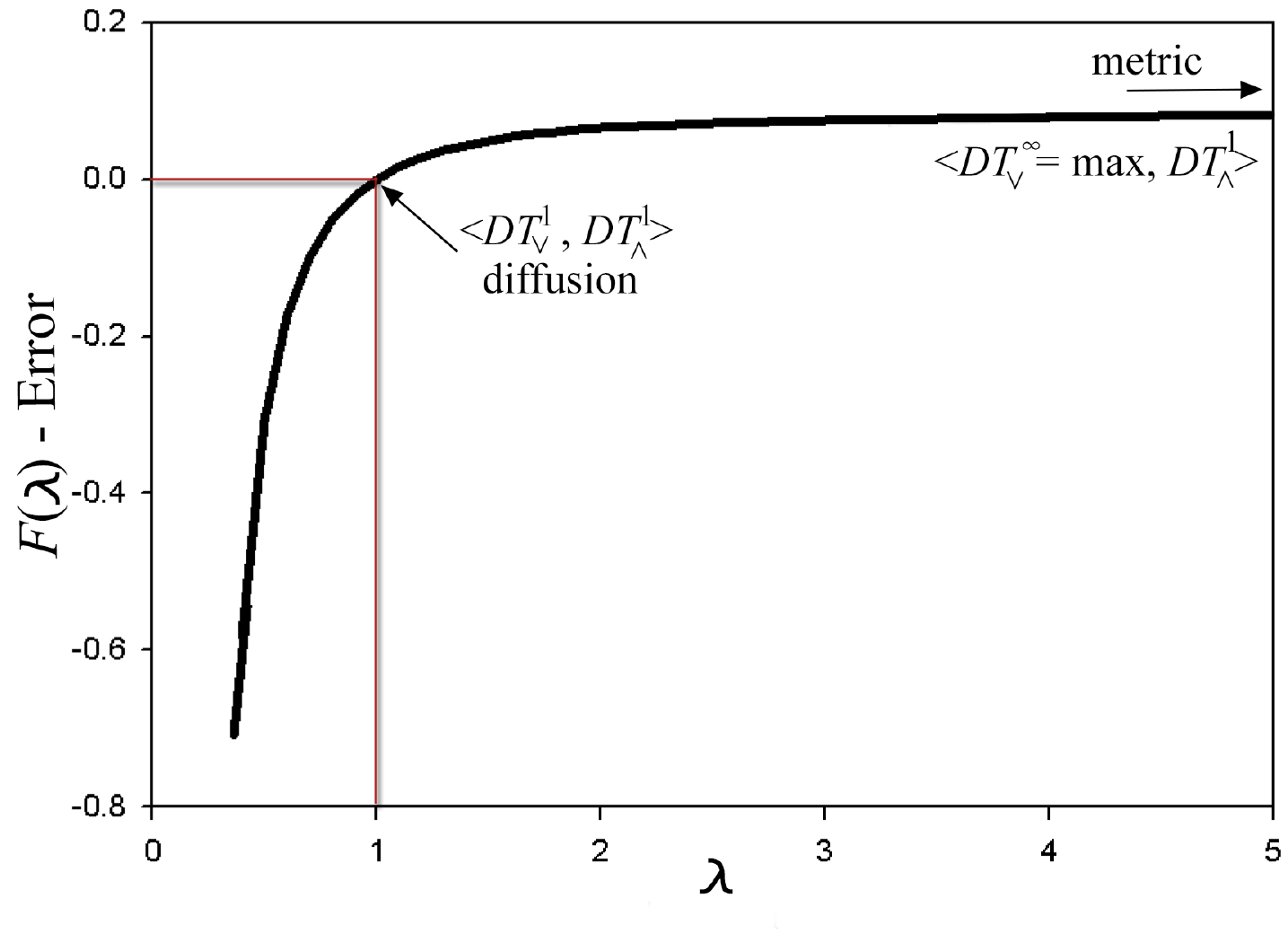}}
\caption{\small Error  between the surface established by the
desired axiomatic constraints (De-Morgan's laws with standard complement), and $\langle DT^{\lambda}_{\lor},DT^{1}_{\wedge} \rangle$ as $\lambda$ varies. }
\label{dombi_error}
\end{figure}

Figure \ref{dombi_error} shows the error from ideal axiomatic characteristics  (computed as the double
integral, above) that ensues from using the pair $\langle \lor,\land \rangle = \langle DT_{\lor}^{\lambda},DT_{\land}^{1} \rangle$ for a given $\lambda$.
The unique, error-free solution exists for $\lambda=1$; $\langle DT^{1}_{\lor},DT^{1}_{\wedge} \rangle$ is the T-Conorm/T-Norm pair used in the diffusion distance closure (\S \ref{diffusion_closure_subsection}).
This pair allows De Morgan's and involution rules to be systematically applied without error when logically combining graphs (in network approximate reasoning).
As $\lambda\rightarrow +\infty$, we reach the T-Conorm/T-Norm pair $\langle \max, DT^{1}_{\land} \rangle$ used in the metric closure.
In contrast, logically combining graphs with this pair will result in the systematic accumulation of errors, meaning that we cannot recover the original values of a graph by involution or by applying De Morgan's laws.
When $\lambda\rightarrow 0$, the Dombi T-Conorm approaches the drastic disjunction\footnote{$u(a,b)=1$, except when $a=0$ or $b=0$, where $u(a,b)$ is $b$ or $a$, respectively \cite{Klir1995}.}, which is revealed to be very far from any desirable characteristics, with unbounded error as $\lambda\rightarrow 0$.

Interestingly, while the pair $\langle \max,DT^{1}_{\land} \rangle$ employed by the metric closure does not possess perfect axiomatic characteristics, its error is bunded, as the curve in Figure
\ref{dombi_error} asymptotically approaches 0.1 when
$\lambda\rightarrow +\infty$.
The relatively small error of this pair may be
acceptable if we do not intend to \emph{frequently} combine our graphs using logical expressions---using approximate reasoning on networks based on T-Norms, T-Conorms, and the complement.
%
%Naturally, if all we are interested in is the computation of shortest paths, as commonly done in complex networks, obeying De Morgan's laws is largely irrelevant---though shortest paths as computed via the metric closure do not have a dual formulation, whereas the indirect distances between any two nodes computed by the diffusion or ultra-metric closures do have dual formulations via DeMorgan's laws.

There are thus two solutions available if we are interested in algebraic structures $I$ and $II$ (\S \ref{isomorph_section}) capable of logical reasoning with an involutive complement---when we intend to use proximity or distance graphs as knowledge representations and manipulate them with network approximate reasoning.
We can either preserve the notion of shortest path ($f \equiv \min$) or the notion of path length as the sum of edges ($g \equiv +$), but not both simultaneously, because there is no isomorphic T-Norm/T-Conorm pair that obeys De Morgan's laws and simultaneously satisfies those two notions---which define the metric closure.
When we preserve the notion of shortest path ($f \equiv \min$) the \underline{only} alternative is the ultra-metric T-Conorm/T-Norm pair $\langle \lor \equiv \max, \land \equiv \min \rangle$ (\S \ref{ultra_metric_example}).
When we preserve the notion of path length as the sum of edges ($g \equiv +$), then \underline{only} the diffusion T-Norm/T-Conorm pair $\langle \lor \equiv DT^{1}_{\lor}, \land \equiv DT^{1}_{\wedge} \rangle$ obeys De Morgan's laws with the most intuitive and simple isomorphism (eq. \ref{distance_measures4}).
Next we study this second alternative in more detail and show that it leads to a notion of distance closure potentially useful for network science in its own rightç that is, even if we are not interested in network approximate reasoning.

\subsection{Diffusion Closure}
\label{diffusion_closure_subsection}

The T-Conorm/T-Norm pair $\langle \lor \equiv DT^{1}_{\lor}, \land \equiv DT^{1}_{\wedge} \rangle$ obtained above, obeys De Morgan's laws and preserves the notion of path length as the sum of edges ($g \equiv +$). However, when used to compute a distance closure via our isomorphism (\S \ref{isomorphism_definition}), it relaxes the notion of shortest path because: $\lor \neq \max \Leftrightarrow f \neq \min$. We now study what ensues from this pair when used in algebraic structure $II$ to compute a \emph{diffusion distance closure}. Moreover, because $\lor \neq \max$, convergence in finite time is no longer guaranteed (\S \ref{convergence}), and so we also need to understand how to use this diffusion closure computationally.

\vspace{5mm}

%\begin{example}
\emph{\bf{Example 3 (Diffusion Closure)}}
\label{ex_diff_clos}
Let $\varphi:[0,1]\rightarrow [0,+\infty]$, $d_{ij} = \varphi(p_{ij}) = \frac{1}{p_{ij}}-1$ (as in equation \ref{distance_measures4}, \S \ref{semimetric_background}).
Let also $\land \equiv DT^{1}_{\land}(a,b)$ (eq. \ref{eq_dombi_and_1}), and $\lor \equiv DT^{1}_{\lor}(a,b)$ (eq. \ref{eq_dombi_or_1}), where $a,b \in [0, 1]$ represent proximity weights from semi-ring $I$ (\S \ref{isomorph_section}).
We know from theorem \ref{theorem11}, eq. \ref{eq_constraint_g}:

\[g(x,y)=\varphi(\land(\varphi^{-1}(x),\varphi^{-1}(y)))\]

\noindent where $x,y \in [0, +\infty)$ represent distance weights from semi-ring $II$ (\S \ref{isomorph_section}), $x=\varphi(a)$ and $y=\varphi(b)$.
Since $\varphi^{-1}(x)=\frac{1}{x+1}$, by substitution of $\varphi, \varphi^{-1}$, and $\land = DT^{1}_{\land}$, we obtain:

\[g(x,y)\equiv x+y\]

We apply the same reasoning to $f$, using eq. \ref{eq_constraint_f}:

\[f(x,y) \equiv \varphi (\lor( \varphi^{-1}(x),\varphi^{-1}(y)))\]

\[f(x,y)=\frac{1-DT^{1}_{\lor}(\varphi^{-1}(x),\varphi^{-1}(y))}{DT^{1}_{\lor}(\varphi^{-1}(x),\varphi^{-1}(y))}\]
\[f(x,y)= \frac{1-DT^{1}_{\lor}(\frac{1}{x+1},\frac{1}{y+1})}{DT^{1}_{\lor}(\frac{1}{x+1},\frac{1}{y+1})}\]
\noindent yielding,

\[f(x,y)= \left\{ \begin{array}{ccc}
y & for & x=+\infty \\
x & for & y=+\infty \\
\frac{xy}{x+y} & for & x,y \in ]0,+\infty[ \\
0 & for & x = y =0 \end{array} \right.
\]
%\end{example}

Therefore, the transitive closure of a proximity graph with algebraic structure $I$ where $\langle \lor, \land \rangle \equiv \langle DT^{1}_{\lor}, DT^{1}_{\land} \rangle$, is isomorphic to the distance closure using algebraic structure $II$ where $f(x,y)=\frac{xy}{x+y}=\frac{1}{\frac{1}{x}+\frac{1}{y}}$ and $g(x,y)=x+y$ in the distance space. With this algebraic structure $II$, the composition of distance graphs $G_D = (X, D)$ (definition \ref{def4}, \S \ref{isomorphism_definition}) is given by this specific TD-Conorm/TD-Norm pair:

\[ D \circ D = \mathop{f}\limits_{k} \, (d_{ik} + d_{kj}) = \frac{1}{\displaystyle\sum_{k} \frac{1}{d_{ik} + d_{kj}}} = d^2_{ij}, \quad \forall x_i,x_j,x_k \in X \]

\noindent because $g(d_{ik},d_{kj})=d_{ik}+d_{kj}$.
Since $p_k = d_{ik} + d_{kj}$ is the length of the path between vertices $x_i$ and $x_j$, via vertex $x_k$, the distance between vertices after composition with $\langle f, g \rangle$ becomes:

\begin{equation}
\label{diffusio_dist_eq}
d^2_{ij} = \frac{1}{\frac{1}{p_1}+\dots+\frac{1}{p_{\nu}}}=\frac{HM(p_1,\dots,p_n)}{\nu}
\end{equation}

\noindent where $\nu$ is the number of distinct paths $p_k$ that exist between $x_i$ and $x_j$, via some vertex $x_k$, and $HM$ is the harmonic mean of the lengths of such paths.
This means that the operations $\langle \lor, \land \rangle \equiv \langle DT_{\lor}^{1}, DT_{\land}^{1} \rangle$ of proximity graphs yield a \emph{diffusion distance} \cite{Coifman_etal_2005} in distance space; thus, the transitive closure with the Dombi operators (with $\lambda = 1$) yield a \emph{diffusion closure} of distance graphs.

%\footnote{It would be interesting to identify, via the isomorphism constraint of theorem  \ref{theorem11}, what algebraic structure $I$ would result in proximity space by using the plain harmonic mean in the distance closure (instead of the harmonic mean divided by number of paths $n$ of equation \ref{diffusio_dist_eq}). However, this would result in operations that do not obey all the axioms required for a TD-Norm and TD-Conorm pair for algebraic structure $II$ in definition \ref{def4a}.}.

%LMR: Added this footnote, see if you agree.

The algebraic structure $I=([0,1],DT_{\lor}^{1}, DT_{\land}^{1})$ does not form a dioid (it is a pre-ordered bounded lattice \cite{Han2004}). This means that the conditions of theorems \ref{theorem3a} and \ref{theorem3c} are not met, and therefore the transitive and distance closures are not guaranteed to converge in a finite time (\S \ref{convergence}).
Indeed, the transitive closure $G_P^T (X, P^T)$ with $I=([0,1],DT_{\lor}^{1}, DT_{\land}^{1})$ converges asymptotically to the T-Norm neutral element $e=1$ as $\kappa \rightarrow +\infty$ in Definition \ref{def3}, but not in finite time.
Likewise, via the isomorphism, the diffusion distance closure $G_D^T (X, D^T)$ with algebraic structure $II=([0,+\infty], f(x,y)=\frac{xy}{x+y}, g(x,y)=x+y)$ converges asymptotically to the TD-Conorm neutral element $\dot e=0$ as $\kappa \rightarrow +\infty$ in Definition \ref{def6b}.

From eq. \ref{diffusio_dist_eq}, it is easy to see that the diffusion closure of a connected distance graph converges to a fully connected distance graph where every edge weight is near zero. As $\kappa$ increases, the distance between every pair of vertices is computed over and over as the harmonic mean divided by the (growing) number of paths of $\kappa$ (repeating) edges, leading to a quick convergence to zero.
Indeed, while this diffusion distance closure does not converge in finite time, it does quickly converge to arbitrarily near its limit ($\dot e=0$) in just a few (composition) steps of $\kappa$ in Definition \ref{def6b}.

It is precisely what happens to the original distance graph $G_D (X, D)$ in the first few $\kappa$ steps of the diffusion closure computation that makes it interesting for network science. In other words, the limit to which the diffusion distance closure converges (all edges with zero distance weight) is trivial---information is in the limit completely diffused to all connected vertices. But as we compute $D^n$, the $n$-Power of distance graph $G_D$ (Definition \ref{def4}), for small values of $n$, we can study \emph{diffusion processes} of $n$ steps on networks.
The indirect distances computed via diffusion, offer an altogether different quantification of indirect distances on networks from what we can obtain via the shortest-path distance closures (\S \ref{max_closures}), such as the metric closure (via APSP/Dijkstra).
Moreover, this approach to studying diffusion distances naturally derives from the algebraic formulation we have outlined here, rather than via stochastic algorithms such as random walks---commonly used in network science to study diffusion processes \cite{noh2004random,fronczak2009biased}.

$D^n$ (and isomorphically $P^n$) yields a graph whose edges measure the diffusion distance between vertices of the original graph $G_D(X,D)$.
More specifically, the edges between vertices $x_i$ and $x_j$ are computed as  the harmonic mean of the lengths of all paths of $n$ edges (computed as the sum of constituent edge weights) between $x_i$ and $x_j$, divided by the total number of distinct such paths ($\nu$, eq. \ref{diffusio_dist_eq}).
We can think of this as a measurement of how near are (indirectly connected) vertices $x_i$ and $x_j$ to each other, if information is allowed to traverse all paths of $n$ edges between them---where the same edge can repeat in the formation of a path.
Therefore, we can think of the $n$-Power of a distance graph, $D^n$, as a $n$-\emph{diffusion} process. Rather than computing a distance closure (as $\kappa \rightarrow +\infty$, Definition \ref{def6b}), we are thus interested in such contained diffusion processes.

An interesting feature of $D^n$ is that the distance edge weights, in addition to being semimetric (breaking the triangle inequality, \S \ref{semimetric_background}), can also break the symmetry axiom of a metric function when $n > 2$. In other words, the distance function of graphs $D^{n > 2}$ only obeys the nonnegative and anti-reflexive axioms (\S \ref{semimetric_background}) and is therefore a \emph{premetric} (inducing a \emph{pretopology} \cite{STADLER2001241}).
This means that $G_D^n(X,D^n)$ can be a directed, weighted graph, even though $G_D (X,D)$ is undirected.
The symmetry breaking occurs in the computation of the $n$-Power of Distance Graph via Definition \ref{def4} because of the composition of adjacency matrices of graphs with different powers. At $n=3$, the symmetry breaking point, the edge weights between nodes $x_i$ and $x_j$ are obtained as:

\[ d_{ij}^3 = \mathop{f}\limits_{k} \, (d^2_{ik} + d_{kj}) = \frac{1}{\displaystyle\sum_{k} \frac{1}{d^2_{ik} + d_{kj}}}, \quad \forall x_k \in X \]

\[ d_{ji}^3 = \mathop{f}\limits_{k} \, (d^2_{jk} + d_{ki}) = \frac{1}{\displaystyle\sum_{k} \frac{1}{d^2_{jk} + d_{ki}}}, \quad \forall x_k \in X \]

Because $d^2_{ik}$ and $d^2_{jk}$ are computed via the harmonic mean of eq. \ref{diffusio_dist_eq}, and the degree of a node $x$ can be larger in graph $D^2$ than the degree of $x$ in graph $D$, we have $d_{ij}^3 \neq d_{ji}^3$ and the asymmetry appears\footnote{In the case of shortest-path closures (\S \ref{max_closures}), because $f = \min$, the asymmetry does not occur.}.
If we want to avoid the asymmetry, an alternative is to compute the composition only of a graph with itself. In this case, the $n$-diffusion is computed only for the powers of two: $D^{2^\eta} = D^{2^{\eta-1}} \circ D^{2^{\eta-1}}, \eta=1,2,3...$
Each $n=2^\eta$ power of $D$ represents the distance between vertices that arises from diffusion on paths of $2^\eta$ edges.
This approach to computing the $n-$diffusion is reasonable, also because as $\kappa \rightarrow +\infty$, the distance closure naturally converges to the trivial, undirected graph where all edges have zero distance weight. Therefore, in the limit, the transitive closure retains symmetry.
%
%Interestingly for network science, this symmetry breaking process results from asymmetries in the structure (connectivity) of the original graph. If $G_D (X,D)$ is a fully symmetric graph, such as a circle, then there is no edge-symmetry breaking for any $n$-diffusion. But if the original graph structure is asymmetric, there is edge-symmetry breaking on the graph edges that exist in between non-symmetric subgraphs. The asymmetry can only be observed for $n > 2$ because we can only measure it when the diffusion information pertains to paths that return to originating vertices, having gone to other vertices that are not directly connected to the originating one (a minimum of 3 edges).
%
Edge-symmetry breaking is exemplified in more detail in the next subsection (\S \ref{diffusion_examples}), where the utility of $n-$diffusion to network science is also discussed further.

To highlight how different the measurement of indirect distances on networks is between shortest-path closures and $n$-diffusion processes, let us return to our social example.
The ultra-metric closure is based on the idea that the indirect distance between two vertices (in a distance graph) is a shortest-path computed as the smallest of the weakest links (largest edge distance weight) of all paths---the distance between the catholic and the Pope is only the weakest social link found in the strongest possible indirect social chain up the hierarchy.
The metric closure is based on the idea that there is a penalty for the number of edges in the strongest path up the hierarchy.
In contrast, a diffusion process assumes that the ability to ``influence'' a distant node depends (via the harmonic mean) on how many strong paths exist to that node.
Whereas the metric closure assumes the distance between two indirectly connected nodes is the single shortest path between them, the $n$-diffusion assumes that having more or fewer short indirect paths is important in computing indirect path distances.
Our catholic has a higher ability to influence the Pope if there are many alternative strong paths (measured by summing the edges just as in the metric case) up the hierarchy, than if there is only one strong path.

Because diffusion distances automatically account for number of indirect connections, they can be very useful for community detection in weighted graphs and segmentation of topological data \cite{goes_etal_2005, Coifman_etal_2005, Lafon_Lee2005}. As we can see in figure \ref{diffusion}, inside a community the diffusion distance is shorter (from vertex C to B) because there are many possible strong indirect paths. In contrast, from one community to the other the diffusion distance is larger (from vertex B to A) because there are only a few possible indirect paths (bridges).
In the next subsection (\S \ref{diffusion_examples}) we demonstrate with examples the utility of $n$-diffusion for community detection.

\begin{figure}[!th]
\centerline{\includegraphics[width=0.9\textwidth]{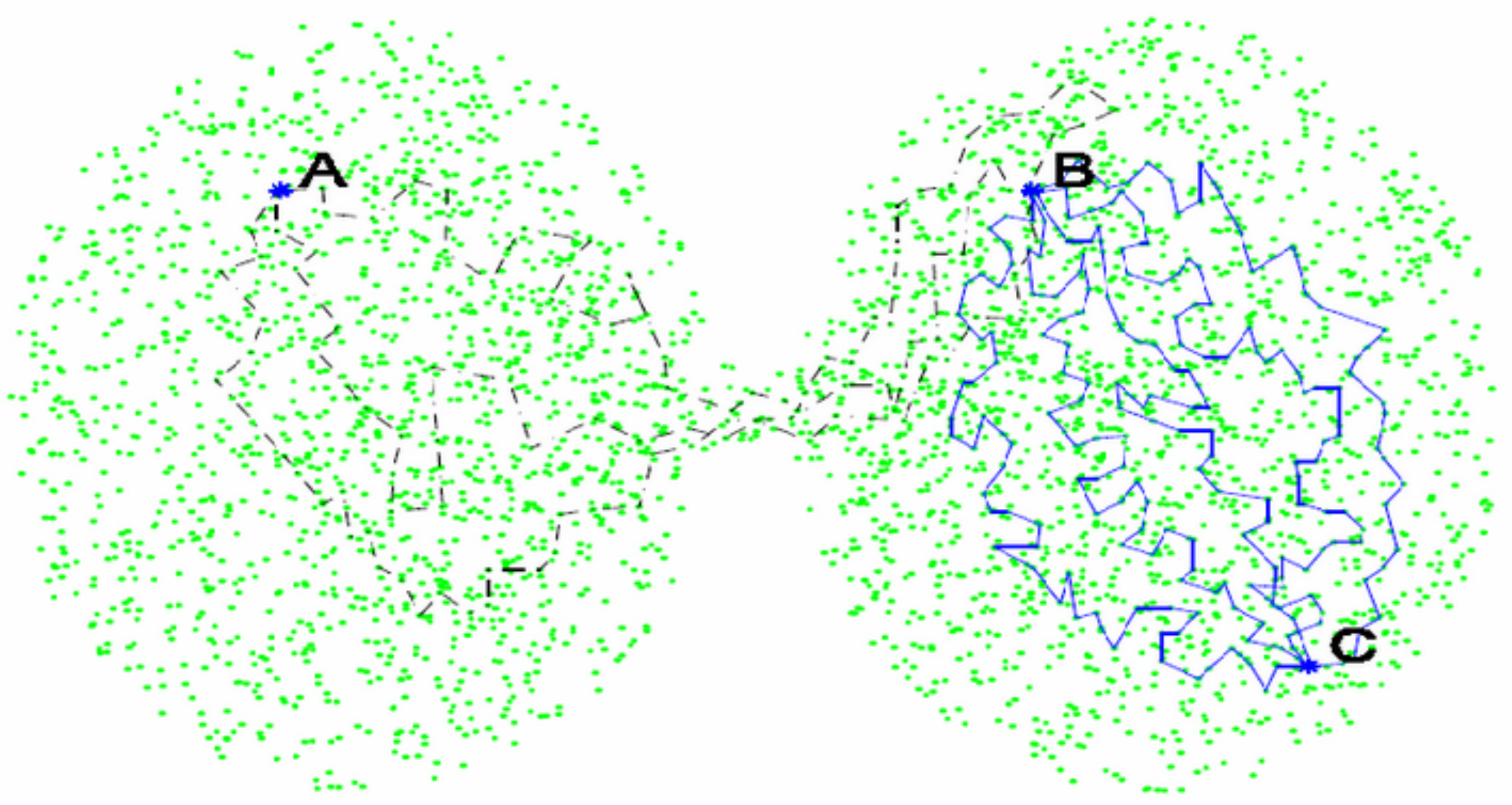}}
\caption{\small Diffusion distance in community detection. From http://www.math.duke.edu/~mauro/diffusiongeometries.html.}
\label{diffusion}
\end{figure}

%
%LMR: Is the figure from that paper? I don't think so, I saw it at http://www.math.duke.edu/~mauro/diffusiongeometries.html. Please make sure it is properly acknowledged, we may not have permission to use it...

Theorem \ref{theorem9} (\S \ref{ultra_metric_example}) shows that the ultra-metric closure always leads to smaller distances than the metric closure---larger distortion of the original graph. But how do $n$-diffusion processes relate to the metric and ultra-metric closures?
It is trivial to see that

\[ \min(p_1,\dots,p_{\nu}) \leq HM(p_1,\dots,p_{\nu}) \]

\noindent therefore the metric closure, which is the shortest path between every pair of nodes, always leads to a smaller or equal distance than the harmonic mean of the lengths of all possible paths between a pair of nodes.
However, the $n$-diffusion computes the distance between vertices according to equation \ref{diffusio_dist_eq}, which is the harmonic mean of the lengths of all paths, divided by the number $\nu$ of such paths.
When $\nu=1$, which happens only in the rare case of a single path $p_1$ (a bridge) between two vertices $x_i$ and $x_j$, the metric closure ($d^{mc}$) and $n$-diffusion ($d^n$) closure yield the same value:

\[ d_{ij}^{mc} = \min(p_1) = p_1 = HM(p_1) = d_{ij}^n. \]

\begin{figure}[!th]
\centerline{\includegraphics[width=1\textwidth]{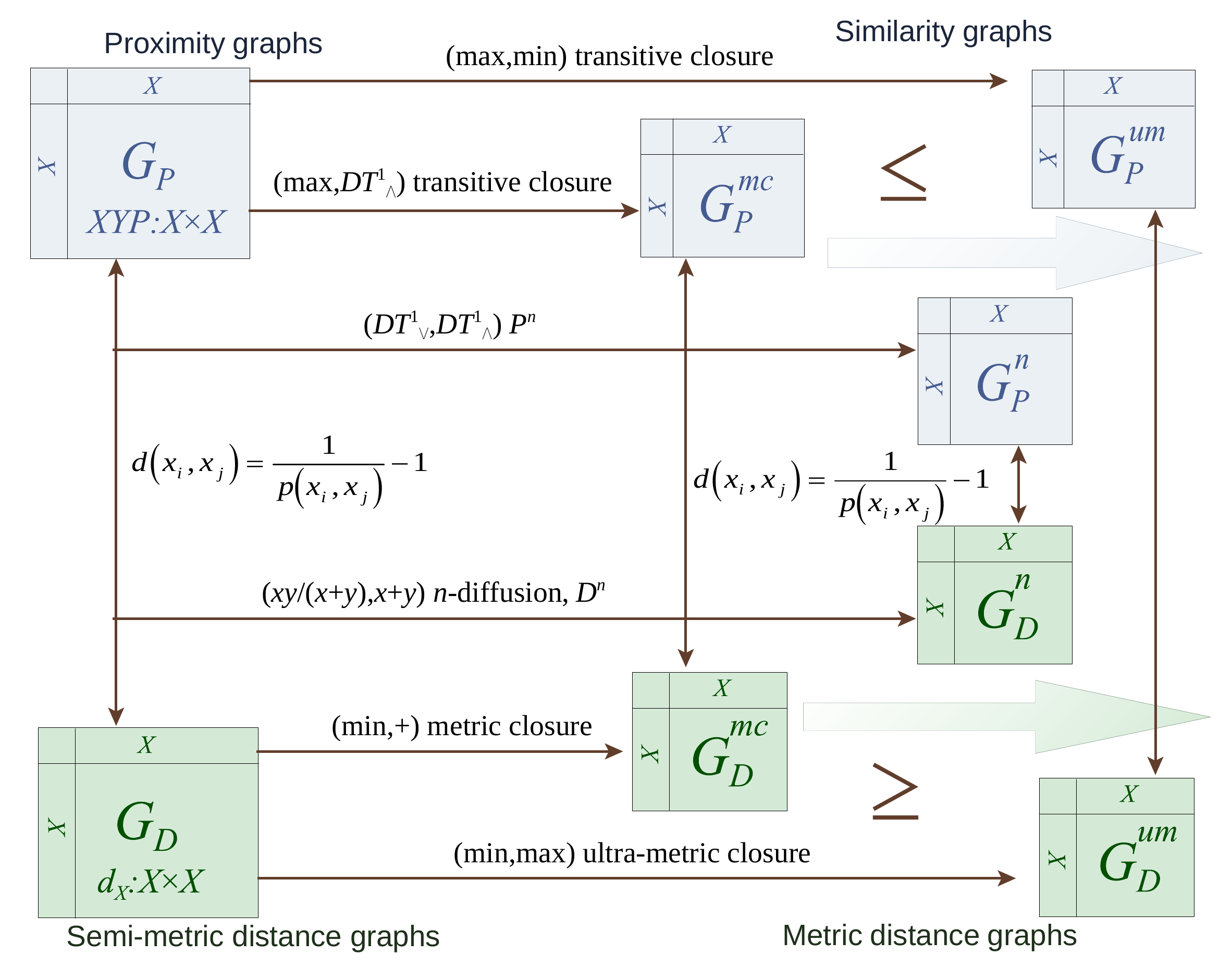}}
\caption{\small $n$-diffusion is isomorphic to the $n$-power of proximity graphs based on the Dombi T-Norm and T-Conorm for $\lambda=1$. The $n$-diffusion is shown with the metric and ultra-metric distance closures, and their fuzzy proximity graph counterparts for $\varphi: distance=\frac{1}{proximity}-1$. It can yield values larger or equal to the metric closure.} \label{diffusion_space}
\end{figure}

\noindent But as $\nu \rightarrow \infty$, we have $d_{ij}^n \rightarrow 0$. Depending on the number of paths that exist between a pair of nodes, the $n$-diffusion leads to a distance value always smaller or equal than the metric closure; and a value that can be larger or smaller than the ultra-metric closure. In other words, the $n$-diffusion distance varies in the interval $(0, \min(p_1,\dots,p_{\nu})]$.
Thus, it is always guaranteed to be metric, but the distance between some vertices (those with many paths between them) can be smaller than ultra-metric, while the distance between others (such as bridges or with very few paths between them) can be very close to metric, a space of variation depicted in Figure \ref{diffusion_space}.
%\footnote{Naturally, when $n \rightarrow \infty$, we also have $\nu \rightarrow \infty$, and so the $n$-diffusion approaches the diffusion closure and all distances become smaller than even the ultra-metric closure.}.
%
Indeed, the fact that the $n$-diffusion depends on the number of indirect paths between any two nodes, is what makes it a natural candidate for community detection as we exemplify next.

\subsection{Applying $n$-Diffusion}
\label{diffusion_examples}

In section \ref{diffusion_closure_subsection} we showed how the concept of distance closure allows us to study diffusion processes on networks using an algebraic formulation---rather than stochastic simulations.
Furthermore, the $n$-\emph{diffusion} process is based on an algebraic structure with good axiomatic characteristics to pursue logical or approximate reasoning on networks (\S \ref{axiomatics}).
In proximity space the algebraic structure $I=([0,1],DT_{\lor}^{1}, DT_{\land}^{1})$ (Dombi T-Conorm/T-Norm pair for $\lambda=1$) is employed, whereas in distance space we have the isomorphic $II=([0,+\infty], f(x,y)=\frac{xy}{x+y}, g(x,y)=x+y)$.

Similarly to what was pursued in section \ref{other_closures} to define generalized metric closures, we can fix the T-Conorm ($DT_{\lor}^{1}$) and vary map $\varphi$ in the isomorphism of theorem \ref{theorem11}, in effect varying all possible T-Norms $\land$. This would result on sweeping the space of possible \emph{generalized diffusion processes}, whereby the length of paths would be computed differently as $g$ would change in the algebraic structure $II$ of distance space. For instance, in the $n$-diffusion case we present here, we have $g(x,y)=x+y$ because path length is computed by summing edges, but if we use T-Norm $DT_{\land}^{\rightarrow +\infty} \rightarrow \min$ (eq. \ref{eq_dombi_and_lambda}), path length would be computed by the weakest link, the largest distance edge weight in a path (because we would obtain $g(x,y)=\max(x,y)$).
Thus, we could compute all variations of diffusion distances as the harmonic mean of path length computed by some measure (divided by number of paths).
The exploration of the space of such generalized diffusion closures and processes is left for future work. Here we want to emphasize the utility of $n$-diffusion processes computed via the algebraic distance closure of sections \ref{axiomatics} and \ref{diffusion_closure_subsection} on two network examples we pursue next.

\subsubsection*{Toy Network}

\begin{figure}[!th]
\centerline{\includegraphics[width=1\textwidth]{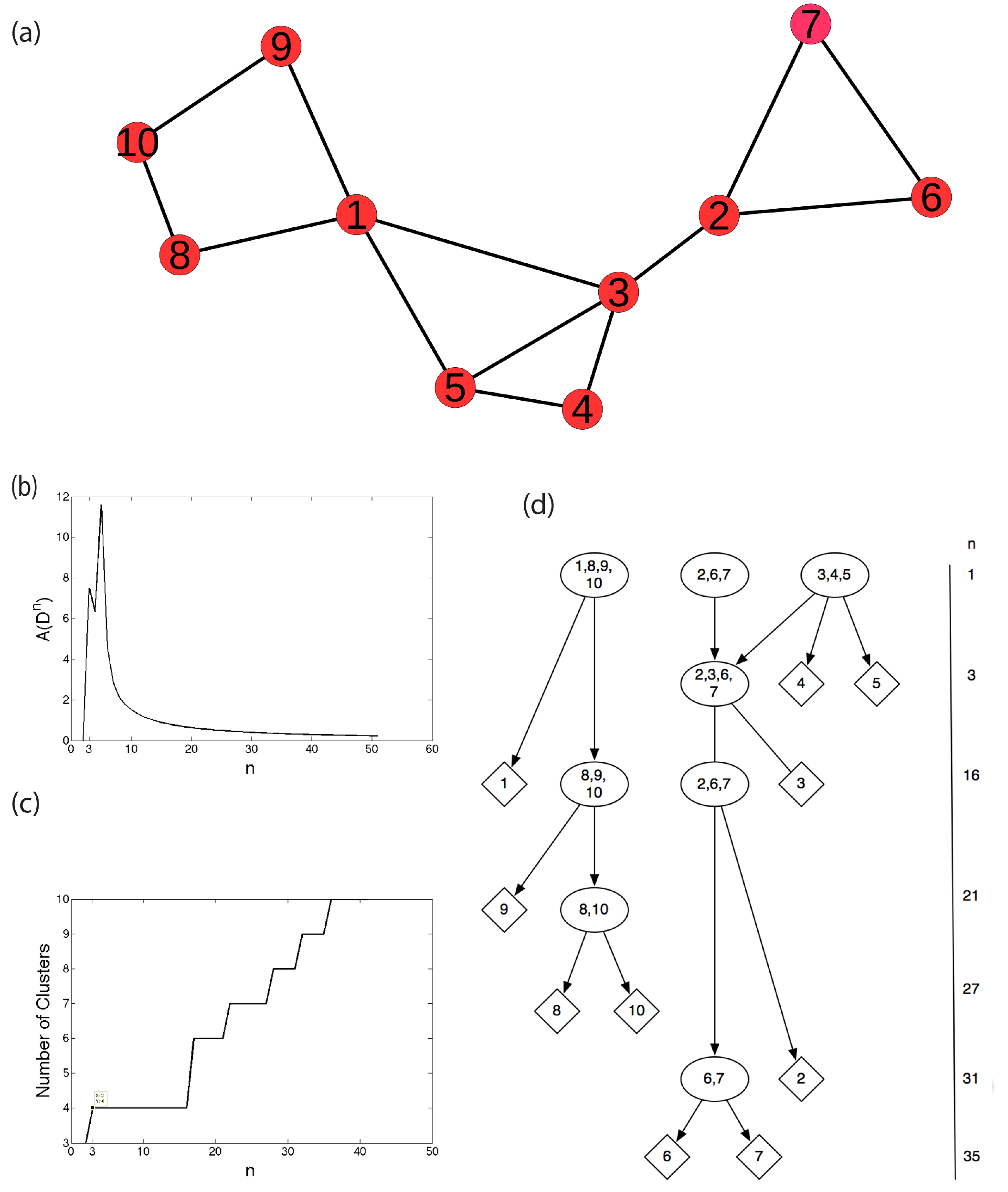}}
\caption{\small $n$-diffusion in Toy Network. (a) the network structure with 3 communities. (b) Asymmetry as defined in equation \ref{asymmetry}. (c) Number of communities in each $D^n$, as $n$ increases. (d) Hierarchy of communities as $n$ increases.  See text for details.
%(e) number of symmetry breaks between one node and the remaining nodes.
}
\label{fig_example1}
\end{figure}

Figure \ref{fig_example1} (a) depicts a toy network, defined by a simple graph where edge weights are as follows: $d_{ii}=0, d_{ij}=1$. When an edge does not exist between $x_i$ and $x_j$, we have $d_{ij}=+\infty$.
This network is designed to display three communities with two types of bridges: a node (1) and an edge (between nodes 2 and 3). Furthermore, one of the communities ($\{1,3,4,5\}$) is a ``bridge community'' as it sits between the other two peripheral communities ($\{1,8,9,10\}$ and $\{2,6,7\}$).

Let us now compute the $n$-diffusion of this graph for $n=1,2,\cdots$, which is given by the respective $n$-power of the distance graph $D^n$ (Definition \ref{def4}, \S \ref{isomorphism_definition}).
Naturally, $n=1$ refers to the original distance network with connectivity matrix $D$.
For each $n$ we compute the community structure of the $n$-diffusion graph $D^n$ using the Louvain method adapted for directed graphs as described in \cite{rubinov2010}. We use a method for directed graphs due to the asymmetry that arises in $n$-diffusion (\S \ref{diffusion_closure_subsection}).
%
%
%We use this algorithm because it is applicable to directed graphs, which is important here because the $n$-diffusion breaks the symmetry of the original (semimetric) distance graph, producing premetric graphs (\S \ref{diffusion_closure_subsection}) for $n > 2$.
%
Other algorithms can be employed, but this one suffices to provide the reader with an intuitive understanding of the effect $n$-diffusion processes.
The $n$-diffusion yields the distance graph $D^n$ whose edges $d_{ij}^n$ denote the indirect distance between nodes $x_i$ and $x_j$ measured by the diffusion that occurs via all paths of $n$ edges between those nodes (including repetition of edges in paths). Diffusion here is measured by the harmonic mean of the length of all these paths (sum of all edge weights), divided by the number of such paths (\S \ref{diffusion_closure_subsection}).
%
%When the initial connectivity structure of the graph is not symmetric, as in this toy example, the length of paths of $n>2$ edges that go from $x_i$ to $x_j$ can be different from the the length of paths that go from $x_j$ to $x_i$. This happens if the graph is not perfectly symmetric about edge $d_{ij}$.
%%
%For instance, in our toy network example, the graph is not symmetric for an axis centered on edge $d_{34} = d_{43}$ (same applies for $d_{35} = d_{53}$). For $n=3$, there exist 6 possible paths of 3 edges that go from node 3 to node 4 $\{\overline{3434}, \overline{3234}, \overline{3134}, \overline{3534}, \overline{3154}, \overline{3454}\}$, but there are only 4 possible such paths that go from node 4 to 3 $\{\overline{4343}, \overline{4323}, \overline{4513}, \overline{4543}\}$.
%
%This means, given eq. \ref{diffusio_dist_eq}, that $d_{34}^3 \neq d_{43}^3$, and so the symmetry is broken for this edge in this toy network  in an $n$-diffusion process at $n=3$ (same occurs for the edge between node 3 and 5).
%
To measure the total amount of asymmetry of the network at a given $n$-diffusion process (\S \ref{diffusion_closure_subsection}), we sum the difference between the upper and lower diagonals of the connectivity matrix of $D^n$:

\begin{equation}
\label{asymmetry}
A(D^n) = \sum_i \sum_{j=i+1} |d_{ij}^n - d_{ji}^n|
\end{equation}

In Figure \ref{fig_example1} (b) we can see how symmetry is broken for this network at $n=3$. Afterwards, asymmetry increases in the $n$-diffusion process until $n=5$, and then quickly decreases asymptotically as $n$ increases and the network converges to its diffusion closure.

Interestingly, the computation of the community structure of each $D^n$ as $n$ increases, produces a form of hierarchical clustering.
In Figure \ref{fig_example1} (d) the communities uncovered as $n$ increases are depicted in a tree with $n$ represented vertically. We can see, for instance, that when symmetry is broken at $n=3$, node 3 moves from a community with nodes 4 and 5, to the community with nodes 2,6, and 7.
In this case, the bridge community is broken, with node 3 of bridge edge $d_{23}$ joining community $\{2,6,7\}$, while the remaining nodes of the bridge community break into single-node communities $\{4\}$ and $\{5\}$.
Only later, at $n=16$, do bridge nodes $\{1\}$ and $\{3\}$ separate.
%
%In general, what happens is that bridge edges break from communities first, before the other members of a community. Also, the bridge communities are broken before the peripheral communities, as we can see happening to $\{1,3,4,5\}$.
%
In summary, the bridge edge and the bridge community break apart well before the bridge nodes do.
As diffusion with longer paths progresses, eventually all communities break into the individual 10 nodes in the network by $n=35$.
This overall process can also be seen in \ref{fig_example1}(c), which depicts the number of communities detected in each $D^n$ network as $n$ increases.
The hierarchical breaking of the original community structure reflects how the diffusion via all possible paths up to a given $n$ connect the nodes in the network. In the limit, all nodes are near all other nodes, but in the diffusion process the nodes of some communities, for a while, remain closer to each other than to the rest of the network. The first nodes to be near all other nodes are 4 and 5, while the nodes in the most peripheral communities and edges remain nearer each other than the rest of the network for longer.

% COMPARE TO NEWMAN FAST AND HIERACHICAL bridges come out first. The process is as ndescribed above. the diffusion of information.

It is clear that $n$-diffusion peels of the community structure of this toy example in a reasonable way. It may also be useful to compare it to more traditional modularity detection algorithms. In Appendix D of supporting materials, the results of the analysis of this network using Newman's Fast algorithm \cite{NewmansFast_PhysRevE} and Hierarchical clustering \cite{day1984efficient} are depicted.
While the results are similar across all methods (e.g. there are clearly three main clusters), some differences are noteworthy. The $n$-diffusion handles the bridges and bridge communities differently. While the other two methods clearly locate nodes 1 and 3 as part of a cluster with nodes 4 and 4, the $n-$diffusion (plus Louvain community detection) distinguish such bridge nodes more emphatically: node 3 is first moved to the $\{2,6,7\}$ cluster and subsequently the first to be broken from this cluster, while 1 is first a part of cluster $\{8,9,10\}$ and then broken from it first.

These subtle differences make sense when we consider that $n-$diffusion is akin to a process of information diffusion on a network via all paths of $n$ edges. As paths of different $n$ are combined, the asymmetry arises and grows, because the diffusion distances computed for a given size path are distinct from those computed with another (\S \ref{diffusion_closure_subsection}). In a sense, the information about the network topology that is transmitted by paths of a certain size, is different from what is transmitted by paths of a different size. As longer paths are considered, information is transmitted across the whole network, the asymmetry disappears (see \ref{fig_example1}(c)) and every node is as near as possible to every other (connected) node.

It is also interesting to compare the results of applying $n$-diffusion to this network, with those obtained via the metric and ultra-metric closures.
Because the edges between distinct nodes of the original toy network all possess the same weight $d_{ij}=1$, the edge-eights that appear from indirect metric paths are all larger than the weights of original edges ($d \in \{2,3,4,5 \}$).
This means that the community structure of the metric closure $D^{m}$ is still best represented by the same 3 communities as the original network.
In contrast, the ultra-metric closure in this case produces a fully-connected graph of the same nodes with same edge-weights $d_{ij}=1$ (everything is connected to everything by the weakest link in a path).This means that the ultra-metric closure of this network $D^{um}$ contains a single community that includes all nodes.
Therefore, the $n$-diffusion analysis reveals a different community structure than the shortest-path metric and ultra-metric closures. In particular, it identifies the bridges and bridge communities more clearly.

\subsubsection*{Co-Activity Flu Network}

%%%%AQUI
%Materials to suplemental
%clustering references
% the other methods in appendix

We repeated the same analysis for a network obtained from real-world data about Flu co-activity between countries.
This network was built from time-series data obtained from \texttt{Google Flu Trends}\footnote{Data Source: Google Flu Trends (http://www.google.org/flutrends)}.
The time-series represent the number of queries that contain the term ``flu'' for a set of 29 countries from both hemispheres. Such time-series have been shown to be correlated with seasonal flu pandemics \cite{Ginsberg:2009qv}.
More specifically, the data collected corresponds to all pandemic seasons in the period of 2004 to 2013.
The country network shown in Figure \ref{fig_example2} (a) was built by computing Pearson's correlation between the time-series of pairs of countries for each edge, similarly to what is commonly done in Neuroscience to compute functional Brain networks from fMRI time-series signals \cite{rubinov2010}.
The correlation weights can be directly interpreted as proximity weights and converted to distance weights via the isomorphism map $\varphi$ of eq. \ref{distance_measures4}.
Because of the isomorphism of theorem \ref{theorem11}, the $n$-diffusion analysis can be performed with either the $n$-power of the proximity graph, $P^n$ (Definition \ref{def2}) using algebraic structure $I=([0,1],DT_{\lor}^{1}, DT_{\land}^{1})$, or the $n$-power of the distance graph $P^n$ (Definition \ref{def4}) using algebraic structure $II=([0,+\infty], f(x,y)=\frac{xy}{x+y}, g(x,y)=x+y)$.

\begin{figure}[!th]
\centerline{\includegraphics[width=1\textwidth]{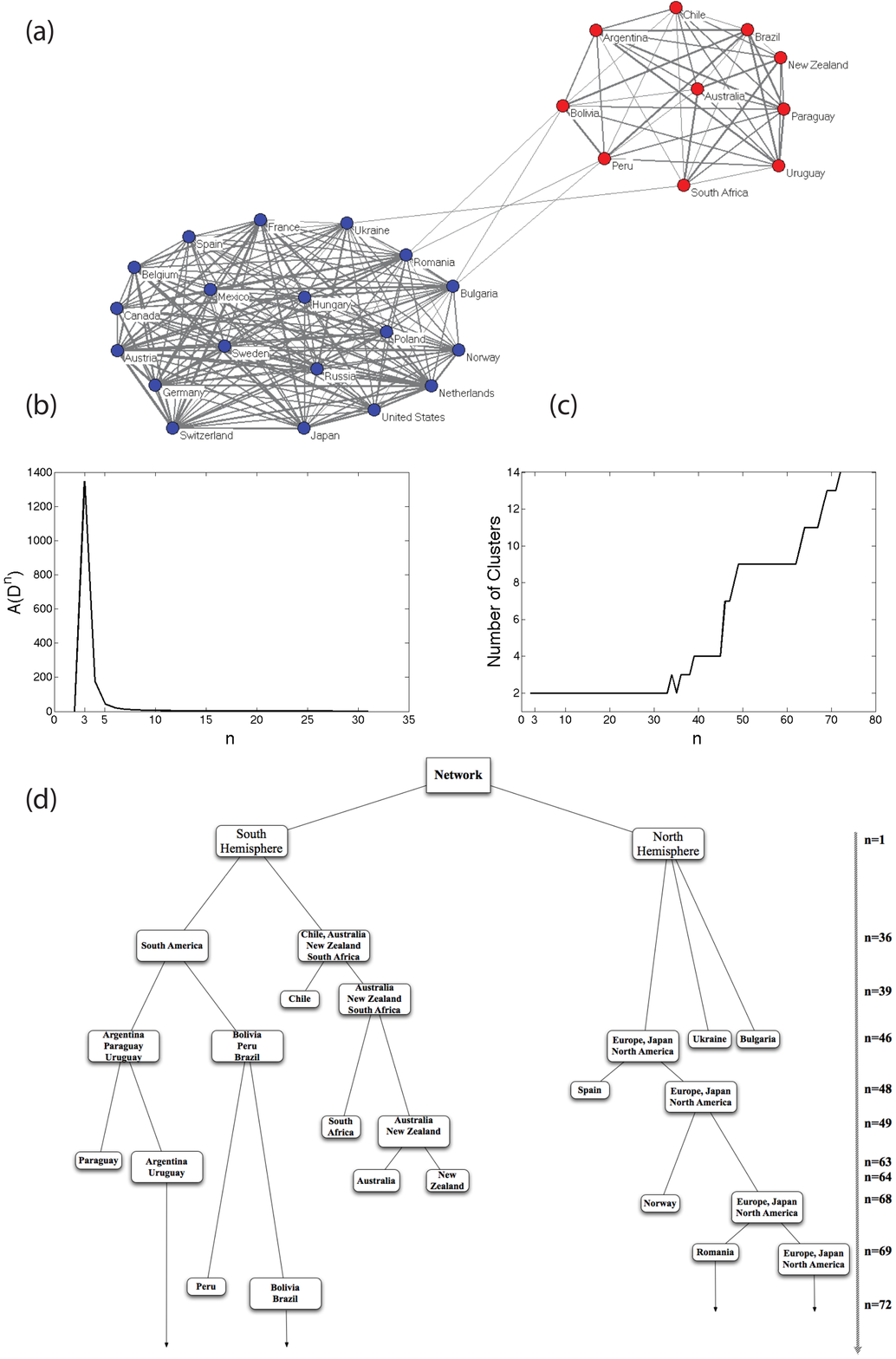}}
\caption{\small $n$-diffusion in flu co-activity network. (a) the 29 country network and its two main groups of countries dividing the north and south hemispheres. (b) Asymmetry as defined in equation \ref{asymmetry}. (c) Number of communities in each $D^n$, as $n$ increases.  (d) Hierarchy of communities as $n$ increases. See text for details.}
\label{fig_example2}
\end{figure}

The results of $n$-diffusion for this real network mirror what we observed for the Toy example.
Again, symmetry breaking occurs at $n=3$, decreasing rapidly after $n\geq5$, as seen in Figure \ref{fig_example2} (b).
In this case, the number of communities remains constant at 2, corresponding to the Northern and Southern hemisphere countries, for a long number of iterations until $n=36$, as seen in Figure \ref{fig_example2} (c) and (d).
This makes much sense because countries in each hemisphere are strongly correlated seasonally with one another.
As $n$ increases, the clusters break apart into constituent countries, peeling the bridges off first.
Indeed, the unfolding of communities in the $n$-diffusion analysis very much reflects the geographical and socio-economic nearness of the countries represented in the network.
This makes sense because we know that flu pandemics are in essence diffusive processes, whereby countries with greater geographical or socio-economic ties are more correlated. Our analysis nicely reproduces those ties.
For instance, notice how the southern hemisphere first breaks South America from the pacific nations, taking Chile (the southern Pacific nation from South America) first into that cluster. But being a bridge node, Chile is the first country node to be isolated into its own community (similarly to node 3 in the Toy example).
Later, south Africa is peeled off from the same South Pacific community, leaving Australia and New Zealand together.
Thus, like in the Toy example, we see bridge nodes and bridge communities in the graph breaking off first.

In Appendix D of supporting materials, the results of the analysis of this network using Newman's Fast algorithm \cite{NewmansFast_PhysRevE} and Hierarchical clustering \cite{day1984efficient} are depicted.
All methods extract a very similar community structure.
A more detailed analysis of $n$-diffusion in real-World networks such as the Flu co-activity network is beyond the scope of this paper and we leave it for forthcoming work.
However

\section{Conclusions}

We mapped and explored the isomorphism between distance and fuzzy (proximity or strength) graphs.
More specifically, we formalized the isomorphic constraints between transitive closure in fuzzy graphs and the distance closure in distance graphs.
In complex networks, the computation of path length and shortest paths is essential for structural analysis of graphs. However, given the isomorphism we explored, it is clear that there is an infinite number of ways to compute indirect distances on graphs, or distance closures, which are isomorphic to transitive closures in fuzzy graphs.
Therefore, the canonical shortest path (the metric closure typically computed via the APSP/Dijkstra algorithm) is just one way of looking at indirect associations in network data.
We have characterized the set of generalized metric closures, which includes all possible shortest path variations, where the length of each path can be computed in an infinite variety of ways---including the ultra-metric closure we also exemplified.

In addition to generalized shortest paths, there are many other ways to compute indirect distances or closures, leading to widely different properties useful for network science.
In particular, we identified a diffusion distance closure which is isomorphic to the transitive closure of fuzzy graphs based on the Dombi T-Conorm/T-Norm pair $\langle \lor, \land \rangle = \langle DT^{1}_{\land}$, $DT^{1}_{\lor} \rangle$.
While this distance closure, in the limit, is trivial, the intermediate steps towards closure, which we named $n$-diffusion, wre shown to be potentially useful for analysis of communities and diffusion processes on networks.
It also offers a simple algebraic means to compute diffusion processes on networks (via matrix products), rather than the traditional stochastic simulations commonly used in the literature.

Whereas the metric closure relates indirectly linked items via the length of the shortest path between them, the $n$-diffusion relates indirectly linked items via the harmonic mean of the length of all distinct paths between them. In other words, the number of available paths is a factor in discerning closeness, which we argued to be intuitively useful in network analysis.
Moreover, unlike the traditional shortest-path method (metric closure), it is based on desirable axioms for logical inference or approximate reasoning on networks.
This means that distance graphs (or their isomorphic fuzzy graphs) obtained from distinct data sources can be logically combined and manipulated with the diffusion distance T-Norm/T-Conorm pair, while obeying De Morgan's laws with any involutive complement.
In contrast, if we use the T-Norm/T-Conorm pair from the metric closure, no involutive complement exists that can satisfy De Morgan's laws, leading to information loss if we were to perform logical combination of graphs.
In summary, while the diffusion distance operators can be used for logical inference---a network approximate reasoning---the metric distance operators cannot.

The isomorphism allowed us to understand and relate other forms of closure, such as the ultra-metric closure and the infinite number of closures parameterized by the Dombi T-Norm/T-Conorm family. This further allowed us to propose a simple method to compute alternative distance closures using existing algorithms for the APSP, as well as estimate the optimal parameter ranges of the Dombi family to constrain variation of a graph's average shortest-path distribution. We showed that the metric closure assumes very small fluctuations in this distribution, therefore, when larger fluctuations exist we should consider alternative closures (namely those with higher values of the Dombi family $\lambda$ parameter).

Based on these results, we argue that different distance closures can lead to different conclusions about indirect associations in network data, as well as the (community) structure of complex networks. Therefore, our exploration of the isomorphism between fuzzy and distance graphs expands the toolbox available to understand complex networks.

\section*{Acknowledgements}

We would like to thank Bharat Dravid for help with the Mathematica calculations of section \ref{axiomatics}, and Joana Gon\c{c}alves S\'{a} for discussions pertaining to the Flu network used in section \ref{diffusion_examples}. We are also very thankful for the constructive comments we received from the reviewers of this article.
This work was partially supported by the Active Recommendation Project at the Los Alamos National Laboratory, National Science Foundation award "Dynamics of Information Flow and Decisions in Social Networks" BCS-0527249, and Funda\c{c}\~ao para a Ci\^{e}ncia e a Tecnologia, Portugal.

%LMR: Precisamos do numero do project FCT eventualmente

\bibliographystyle{apalike}
\bibliography{gbib}

%%%%%%%%%%%%%%%%%%%%%%%%%%%%%%%%%%%%%%%%%%%%%%%%%%%%%%%%%%%%%%%%%%%%%%%%%%%%%%%

\newpage

 %%%% AQUI

%
%\appendix{Appendix - Supplement Material}

%%%%%%%%%%%%%%%%%%%%%%%%%%%%%%%%%%%%%%%%%%%%%%%%%%%%%%%%%%%%%%%%%%%%%%%
% [TS] ALL NEW APPENDIX

\newpage
\appendix
\section{Mathematical Background}
\label{appendix1}

\setcounter{definition}{0}
\setcounter{theorem}{0}
\setcounter{corollary}{0}

In this appendix we present definitions that will be useful for the understanding of the paper.

\subsection{A brief overview on Fuzzy Sets Theory}

First we introduce the definition of T-Norms and T-Conorms first introduced by Menger et al. in \cite{Menger,schweizer83}.

\begin{definition}[T-Norm]
\label{T-Norm}
A \textit{triangular norm (T-Norm for short)} is a binary operation $\land$ on the unit interval $[0,1]$, i.e., a function $\land:[0,1]^{2}\to[0,1]$, such that for all $x,y,z \in [0,1]$ the following four axioms are satisfied:

(T1) $x \land y=y \land x$.

(T2) $x \land (y \land z)= (x \land y) \land z$.

(T3) $x \land y \leq x \land z $ wherever $y \leq z$.

(T4) $x \land 1=x$.
\end{definition}

A \emph{T-Norm} is a generalisation of intersection in set theory and conjunction in logic. It was first defined in the context of probabilistic metric spaces \cite{schweizer83}.

\begin{definition}[T-Conorm]
\label{T-Conorm}
A \textit{triangular conorm (T-Conorm for short)} is a binary operation $\lor$ on the unit interval $[0,1]$, i.e., a function $\lor:[0,1]^{2}\to[0,1]$, such that for all $x,y,z \in [0,1]$, satisfies (T1)-(T3) and

(S4) $x \lor 0=x$.
\end{definition}

A \emph{T-Conorm} is a generalisation of union in set theory and disjunction in logic.

There is an innumerable number of T-Norms and T-Conorms. In the following examples \cite{klement} we present the four basic T-Norms and T-Conorms.
\\
\\
%\begin{example}
%\emph{\bf{Example 2} (T-Norms)}
%\label{ex2}
The following are the four basic T-Norms $\land_{M}, \land_{P}, \land_{L}$ and $\land_{D}$ given by, respectively:
\begin{example}[Basic T-Norms]
(1) $x \land_{M} y=min(x,y)$  (minimum),\\ \\
(2) $x \land_{P} y=x \cdot y$  (product),\\ \\
(3) $x \land_{L} y=max(x+y-1,0)$  (Lukasiewicz T-Norm),\\ \\
(4) $x \land_{D} y=\left\{\begin{array}{cc}
 0, & \mbox{if $(x,y) \in [0,1[^{2}$;} \\
 min(x,y), & \mbox{otherwise.} \\
\end{array} \right.$ (drastic product)
%\end{example}
\end{example}

These T-Norms cover the range for T-Norms, from the strongest T-Norm $\land_{M}$ to the weakest T-Norm $\land_{D}$. There are other T-Norms, namely \emph{parametric T-Norms}, which range the spectrum of all possible T-Norms. Examples of these T-Norms are the \emph{Dombi} T-Norms.

\begin{definition}[Dombi T-Norm]
\label{Dombi-def}
The definition of Dombi \emph{T-Norm} is the following:

\[ DT^{\lambda}_{\wedge}(a,b)=\left\{ 1+\left[ \left( \frac{1}{a}-1 \right)^{\lambda} +\left( \frac{1}{b}-1 \right)^{\lambda}\right]^{\frac{1}{\lambda}} \right\}^{-1} \]

\noindent Where the parameter $\lambda \in \left]0,+\infty \right[$.

\end{definition}

%\begin{example}
\begin{example}[Basic T-Conorms]
The following are the four basic T-Conorms $\lor_{M}, \lor_{P}, \lor_{L}$ and $\lor_{D}$ given by, respectively:\\
(1) $x \lor_{M} y=max(x,y)$  (maximum), \\ \\
(2) $x \lor_{P} y=x+y-x \cdot y$  (probabilistic sum), \\ \\
(3) $x \lor_{L} y=min(x+y,1)$  (Lukasiewicz T-Conorm), \\ \\
(4) $x \lor_{D} y=\left\{\begin{array}{cc}
 1, & \mbox{if $(x,y) \in [0,1[^{2}$;} \\
 max(x,y), & \mbox{otherwise.} \\
\end{array} \right.$ (drastic sum)
\end{example}

These T-Conorms define the specific range of T-Conorms, from the strongest T-Conorm $\lor_{D}$ to the weakest T-Conorm $\lor_{M}$.

\begin{definition}[Dombi T-Conorm]
\label{Dombico-def}
The definition of Dombi \emph{T-Conorm} is the following:

\[ DT^{\lambda}_{\vee}(a,b)=\left\{ 1+\left[ \left( \frac{1}{a}-1 \right)^{\lambda} +\left( \frac{1}{b}-1 \right)^{\lambda}\right]^{-\frac{1}{\lambda}} \right\}^{-1} \]

\noindent Where the parameter $\lambda \in \left]0,+\infty \right[$.

\end{definition}

Now we are able to define the transitivity property of a fuzzy relation.

\begin{definition}[Transitivity]
\label{transitivity}
A fuzzy relation $R(X,X)$ is transitive if \[R(x,y) \geq \lor_z(R(x,z) \land R(z,y))\] is satisfied $\forall x,y,z \in X$.
\end{definition}

Definition ~\ref{transitivity} entails that transitivity depends on the pairs T-Conorm/T-Norm chosen.

\begin{definition}[Fuzzy Complement]
\label{complement}
A \textit{complement} c of a fuzzy set satisfies the following axioms:

(c1) $c(0)=1$  $c(1)=0$ (boundary conditions).

(c2) $\forall a,b \in [0,1]$ if $a \leq b$, then $c(a) \geq c(b)$ (monotonicity).

\end{definition}

The \emph{Complement} of a fuzzy set measures the degree to which a given element of the fuzzy set does not belong to the fuzzy set. Two most desirable requirements, which are usually among of fuzzy complements are:

\begin{definition}[Fuzzy Complement)(cont)]
\label{complement2}
A \textit{complement} c of a fuzzy set satisfies the following axioms:

(c3) $c$ is a continuous function.

(c4) $c$ is involutive, which means that $c(c(a))=a$ for each $a\in[0,1]$.

\end{definition}

In classical set theory, the operations of intersection and union are dual with respect to the complement in the sense that they satisfy the De Morgan laws. It is desirable that this duality be satisfied for fuzzy set as well. We say that a T-Norm $\land$ and a T-Conorm $\lor$ are dual with respect to a fuzzy complement $c$ if and only if \[c(a \land b)=c(a) \lor c(b)\] and \[c(a \lor b)=c(a) \land c(b).\]

Examples of dual T-Norms and T-Conorms with respect to the complement $c_{s}(a)=(1-a)^{s}$ are: \[<min(a,b),max(a,b),c_{s}>\] \[<DT^1(a,b),DT^1(a,b),c_{s}>\].
We can have weaker complements, which only obey to the first two axioms in definition \ref{complement} to allow other T-Norm and T-Conorm operators.

Next we follow with composition of fuzzy relations.

\begin{definition}[Relation Composition]
\label{composition}
Consider two binary fuzzy relations, $P(X,Z)$  and $Q(Z,Y)$ with a common set of $Z$. The standard composition of these relations, which is denoted by $P(X,Z) \circ Q(Z,Y)$ produces a binary fuzzy relation $R(X,Y)$ on $X \times Y$ defined by \[ R(X,Y)=[P \circ Q] = \lor_z(P(x,z) \land Q(z,y)),\]$ \forall x \in X$ and $\forall y \in Y$ and $\forall z \in Z$ .
\end{definition}

When the transitive closure $R^T(X, X)$ uses the T-Conorm $\vee = \texttt{maximum}$, with any T-Norm $\wedge$,
%that obeys the Archimedean property,
$\kappa$ in eq. \ref{TC1} is finite and not larger than $|X|-1$
%\cite{Pang2003, Klir1995}
\cite{Klir1995}.
In other words, the transitive closure converges in finite time and can be easily computed using Algorithm 1 \cite{Klir1995}:

\begin{Algorithm}
\label{tca2}
\begin{enumerate}
  \item $R' = R \cup (R \circ R)$
  \item If $R'\neq R$, make $R = R'$ and go back to step 1.
  \item  Stop: $R^T = R'$
\end{enumerate}
\end{Algorithm}

It has also been shown  that if the semiring formed by $\langle \lor, \land \rangle$ on the unit interval is a dioid or a bounded preordered lattice\cite{Gondran2007}, then $\kappa$ in eq. \ref{TC1} is also finite \cite{Han2004,Han2007} (see conditions below).
%
%LMR: THE SENTENCE ABOVE MUST BE CONFIRMED!!!
%
In this case, the transitive closure can be computed in finite time using Algorithm 2:

\begin{Algorithm}
\label{tca2_v2}
\begin{enumerate}
  \item $R' = R, R_p = R, p=1$
  \item $R_p = R \circ R_p, p=p+1$
  \item If $R' \neq (R' \bigcup R_p)$, make $R' = R' \bigcup R_p $ and go back to step 2.
  \item  Stop: $ R^T = R' $
\end{enumerate}
\end{Algorithm}

%LMR: I am not convinced that we need everything that is below. It should be summarized to required results and it should be coherent with notation in paper.

The union in step 1 must be in accordance with the T-Conorm defined in the relation composition. The resulting relation in step 3 is transitive with respect to the T-Norm, T-Conorm operations used. Moreover, given the last algorithm, a fuzzy graph is transitive if the algorithm stops at the first step.
A reflexive, symmetric and transitive fuzzy relation is denominated as a \textit{Similarity} or \textit{Equivalence} relation.

Next we give a more detailed description of T-Norms and T-Conorms.

The intersection of two fuzzy sets $A$ and $B$ is performed by a binary operation closed on the unit interval. There are an infinite number of T-Norms from definition \ref{T-Norm}. One important class is that of \emph{Archimedean T-Norms}, see \cite{Klir1995}. Before we introduce one of the fundamental theorems of T-Norms, which provides us a method for generating Archimedean T-Norms we introduce the following definitions:

\begin{definition}[Decreasing Generator]
\label{def-generator}
A decreasing generator $\varphi$ is a continuous decreasing function from the unit interval $[0,1]$ into the real extended interval $[-\infty,+\infty]$.
\end{definition}

\begin{definition}[Pseudo-Inverse of a decreasing generator]
 \label{def-pseudo-inverse}
The pseudo-inverse of a decreasing generator $\varphi$ is defined by
\[
\varphi^{(-1)}(a)= \left\{ \begin{array}{ccc}
1 & for & a \in (-\infty,0) \\
\varphi^{-1}(a) & for & a \in [0,\varphi(0)] \\
0 & for & a \in (\varphi(0),\infty) \end{array} \right.
\]

Where $\varphi^{-1}$ is the inverse function of $\varphi$.

\end{definition}

\begin{theorem}[Characterization Theorem of T-Norms]
\label{characterization1}
Let i be a binary operation closed on the unit interval. Then, i is an Archimidean T-Norm iff there exists a decreasing generator $\varphi$ such \[a \land b=\varphi^{(-1)}(\varphi(a)+\varphi(b))\] for all $a,b \in [0,1]$.
\end{theorem}

\noindent With this last theorem we can generate an infinite class of T-Norms. Among many decreasing generators is the \emph{Dombi T-Norm generator}, (see definition \ref{Dombi-def}): \[ \varphi(x)=\left(\frac{1-x}{x}\right)^{\lambda}\] Parameter $\lambda$ allow us to obtain the range from the $\land_{D}$ T-Norm $(\lambda \to 0)$ to the $\land_{M}$ T-Norm $(\lambda \to+\infty)$. For many other decreasing generators, see \cite{klement}.

Set unions are generalized by the T-Conorms in definition \ref{T-Conorm}. There are an infinite number of T-Conorms and ways to generate new T-Conorms. One important class of T-Conorms is the \emph{Archimedean T-Conorms}, see \cite{Klir1995}.

\begin{definition}[Increasing Generator]
 \label{def-generator2}
A increasing generator $\theta$ is a continuous increasing function from the unit interval $[0,1]$ into the real extended interval $[-\infty,+\infty]$.
\end{definition}

\begin{definition}[Pseudo-Inverse of a increasing generator]
 \label{def-pseudo-inverse_inc}
The pseudo-inverse of a increasing generator $\theta$ is defined by
\[
\theta^{(-1)}(a)= \left\{ \begin{array}{ccc}
0 & for & a \in (-\infty,0) \\
\theta^{-1}(a) & for & a \in [0,\theta(0)] \\
1 & for & a \in (\theta(0),\infty) \end{array} \right.
\]

Where $\theta^{-1}$ is the inverse function of $\theta$.
\end{definition}

\begin{theorem}[Characterization Theorem of T-Conorms]
\label{characterization1_TCO}
Let u be a binary operation closed on the unit interval. Then, u is an Archimidean T-Conorm iff there exists an increasing generator $\theta$ such \[a \lor b=\theta^{(-1)}(\theta(a)+\theta(b))\] for all $a,b \in [0,1]$.
\end{theorem}

\noindent With this last theorem we can generate an infinite class of T-Conorms. Among many increasing generators is the Dombi T-Conorm generator: \[ \theta(x)=\left(\frac{x}{1-x}\right)^{\lambda}\] Parameter $\lambda$ allow us to obtain the range from the $\lor_{M}$ T-Conorm $(\lambda \to0)$ to $lor_{D}$ T-Conorm $(\lambda \to+\infty)$. For many other decreasing generators the reader can see, \cite{klement}.

%A T-Norm is called archimedean iff $\forall
%if for each x, y in the open interval (0, 1) there is a natural number n such that x * ... * x (n times) is less than or equal to y.

% REMOVE ALL THIS SUBSECTION ON ALGEBRAIC STRUCTURES
%%%%%%%%%%%%%%%%%%%%%%%%%%%%%%%%%%%%%%%%%%%%%%%%%%%%%%%%%%%%%%%%%%%%%%
\subsection{Algebraic Structures Basics}
Here we present the basic definitions on algebraic structures used in this work.

The whole class of semirings splits into main disjoint subclasses: (a) rings and (b) canonical ordered semirings or \emph{dioids}.
On the following, we consider algebraic structures consisting of a basic set $E$, with two internal operations $\oplus$ and $\otimes$. All these definitions can be found on \cite{Gondran2007}.

\begin{definition}[Semi-Ring]
\label{semiring}
Let consider the following algebraic structure $L=$($E$, $\oplus$, $\otimes$). $L$ is called a \emph{semiring} if the following properties hold:\\ \\
(i) ($E$, $\oplus$) is a commutative monoid with zero element $\epsilon$,\\
(ii) ($E$, $\otimes$) is a monoid with unit element $e$,\\
(iii) $\otimes$ is right and left distributive with respect to $\oplus$,\\
(iv) $\epsilon$ is absorbing, i.e. $\epsilon \otimes a = a \otimes \epsilon = \epsilon$, $\forall a \in E$.
\end{definition}

\begin{definition}[Canonical Order]
\label{canonicalorder}
$L=$($E$, $\oplus$) being a monoid, the binary relation $\leq$ on $E$ is defined as: $a\leq b$ iff $\exists c \in E$ such that $b=a\oplus c$, is a \emph{preorder relation} (reflexive and transitive) called the \emph{canonical preorder}. A monoid is called \emph{canonically ordered} iff the canonical preorder is order, or equivalently iff $\leq$ is antisymmetric ($a\leq b$ and $b\leq a$ $\Rightarrow$ $a=b$).
\end{definition}

\begin{definition}[Dioid]
\label{dioid}
A semiring ($E$, $\oplus$, $\otimes$) such that ($E$, $\oplus$) is canonically orderd is called a \emph{dioid}.
\end{definition}

The algebraic structure $I=$($[0, 1]$, $\lor$, $\land$), where $\lor$ and $\land$ are general T-Conorm/T-Norm, respectively, are not in general a dioid, since they fail property $(iii)$ (distributivity) of definition \ref{semiring} (semiring). However, there are subclasses of the algebraic structure $I=$($[0, 1]$, $\lor$, $\land$), which are dioids.

For more details about algebraic structures see for example \cite{Gondran2007, Han2004} or any book about Abstract Algebra.

%%%%%%%%%%%%%%%%%%%%%%%%%%%%%%%%%%%%%%%%%%%%%%%%%%%%%%%%%%%%%%%%%%%%%%%
\newpage
\section{Proofs to the theorems}
\label{appendix2}

In this section we provide the proofs to the theorems in the main text.

\setcounter{theorem}{0}
\setcounter{corollary}{0}

\begin{theorem}
\label{theorem1_proof}
 Let $G_{P}=(X,P)$ be a proximity (symmetric and reflexive) graph and $\Phi$ the graph distance function in definition ~\ref{def1}, then $G_{D}=(X,D)$, where $D=\Phi(P)$, is symmetric and anti-reflexive.
\end{theorem}

\begin{proof}
Since $G_{P}$ is reflexive then $p_{x,x}=1$ and from definition ~\ref{def1} we have $d_{x,x}=\varphi(p_{x,x})=\varphi(1)=0$, therefore $G_{D}$ is anti-reflexive. Let $x$ and $y$ be two vertices of $G_{P}$, because a proximity graph is symmetric we have $p_{x,y}=p_{y,x}$, since $\varphi$ is bijective $d_{x,y}=\varphi(p_{x,y})=\varphi(p_{y,x})=d_{y,x}$, therefore $G_{D}$ is symmetric.
\end{proof}

\begin{theorem}
\label{theorem2_proof}
 If $\varphi$ is a distance function as in definition ~\ref{def1}. For every pair of T-Norm/T-Conorm operations $\langle \land,\lor \rangle$, there exists a pair of operations  $\langle f, g \rangle$ a TD-conorm/TD-norm (definition \ref{def4a}) and vice versa, obtained via the following constraints:\\
\\(1) $\varphi(a \land b) = g(\varphi(a),\varphi(b))$;
\\(2) $\varphi(a \lor b)= f(\varphi(a),\varphi(b))$.
\\
\\ Where $a,b \in [0,1]$.
\end{theorem}

\begin{proof}
Let us assume $a\leq b$.
\\(1) Suppose $\varphi(a \land b) > g(\varphi(a),\varphi(b))$, thus the inequality is true if the maxima of $\varphi(a \land b)$ (must be maximum) is bigger than the minimum of $g(\varphi(a),\varphi(b))$ (must be minimum). $\varphi(a \land b)$ is maximum for $\land \equiv T_{D}$ (drastic product, see \cite{klement} \cite{Klir1995}) and $g(\varphi(a),\varphi(b))$ is minimum for $\varphi(b)=0$, thus for $\varphi(b)=0$ we obtain, $\varphi(a)\leq g(\varphi(a),\varphi(b))$, from the other side $\varphi(a \land b)\leq \varphi(min(a,b))=\varphi(a)$. Therefore, $\varphi(a \land b) \leq g(\varphi(a),\varphi(b))$.

Suppose $\varphi(a \land b) < g(\varphi(a),\varphi(b))$, thus $\varphi(a \land b)$ must be minimum and $g(\varphi(a),\varphi(b))$ must be maximum. $\varphi(a \land b)$ is minimum for $\land \equiv min$ and $g(\varphi(a),\varphi(b))$ is maximum for $a=0$, thus for $a=0$ we obtain, $g(\varphi(a),\varphi(b)) \leq \varphi(a)$, from the other side $\varphi(a \land b)\geq \varphi(min(a,b))=\varphi(a)$. Therefore, $\varphi(a \land b) \geq g(\varphi(a),\varphi(b))$, and from above this implies $\varphi(a \land b) = g(\varphi(a),\varphi(b))$, which proves statement (1).
%$\varphi(a \land b)=\varphi(a)$ because $a \land b=a$ and $\varphi(a) \lor \varphi(b)=\varphi(a)$ because $\varphi$ is monotonic decreasing we have, $a \leq b \Rightarrow \varphi(a) \geq \varphi(b)$, thus, $\varphi(a \land b)=\varphi(a) \lor \varphi(b)$.
\\(2) Suppose $\varphi(a \lor b) > f(\varphi(a),\varphi(b))$, thus $\varphi(a \lor b)$ must be maximum and $f(\varphi(a),\varphi(b))$ must be minimum. $\varphi(a \lor b)$ is maximum for $\lor \equiv max$ and $f(\varphi(a),\varphi(b))$ is minimum for $\varphi(a)=0$, thus for $\varphi(a)=0$ we obtain, $f(\varphi(a),\varphi(b))\geq 0$, from the other side $\varphi(a \lor b)\leq \varphi(max(1,b))=\varphi(1)=0$. Therefore, $\varphi(a \lor b) \leq f(\varphi(a),\varphi(b))$.

Suppose $\varphi(a \lor b) < f(\varphi(a),\varphi(b))$, thus $\varphi(a \lor b)$ must be minimum and $f(\varphi(a),\varphi(b))$ must be maximum. $\varphi(a \lor b)$ is minimum for $\lor \equiv S_{D}$ (drastic sum, see \cite{klement} \cite{Klir1995}) and $f(\varphi(a),\varphi(b))$ is maximum for $b=0$, thus for $b=0$ we obtain, $f(\varphi(a),\varphi(b)) \leq \varphi(a)$, from the other side $\varphi(a \lor b)\geq \varphi(max(a,b))=\varphi(a)$. Therefore, $\varphi(a \lor b) \geq f(\varphi(a),\varphi(b))$, and from above this implies $\varphi(a \lor b) = f(\varphi(a),\varphi(b))$, which proves statement (2).
%in the same way as (1), $\varphi(a \lor b)=\varphi(b)$, since $a \lor b=b$ and $\varphi(a) \land \varphi(b)=\varphi(b)$, thus $a \leq b \Rightarrow \varphi(a) \geq \varphi(b)$. Then, $\varphi(a \lor b)=\varphi(a) \land \varphi(b)$.
\end{proof}

\setcounter{theorem}{2}

\vspace{10mm}

\begin{theorem}
\label{theorem5_proof}
If $G_{P}=(X,P)$ is a fuzzy proximity graph and $G_{D}=(X,D)$ is the distance graph obtained from $G_P$ via $D=\Phi(P)$, where $\Phi$ is the isomorphism (distance function) in definition ~\ref{def1}, then the following statements are true:

1) $\Phi (P)\dot \supseteq \Phi (P^{2} )\dot \supseteq \Phi (P^{3} )\dot \supseteq \cdots \supseteq \Phi (P^{\infty } )$ ;

2) $D\dot \supseteq D^{2} \dot \supseteq D^{3} \dot \supseteq \cdots \dot \supseteq D^{\infty } $.
\end{theorem}

\noindent where $\Phi(P^{n})\dot \supseteq \Phi(P^{n+1})$ means that: $\forall x_i,x_j \in X: \varphi(p^{n}_{ij}) \geq \varphi(p^{n+1}_{ij})$, and $D^{n} \dot \supseteq D^{n+1}$ means that: $\forall x_i,x_j \in X: d^{n}_{ij} \geq d^{n+1}_{ij}$.

\begin{proof}
1)  $\varphi $  is a monotonic decreasing function and because P is reflexive, from \cite{Mordeson2000} we have  $P\subseteq P^{2} \subseteq P^{3} \subseteq \cdots \subseteq P^{\infty } \Rightarrow \Phi (P)\dot \supseteq \Phi (P^{2} )\dot \supseteq \Phi (P^{3} )\dot \supseteq \cdots \dot \supseteq \Phi (P^{\infty } )$  which proves the statement.\\
2) To prove the second statement we first need to prove that $D\dot \supseteq D^{2}$, which is equivalent to showing that, $\forall x,y,z\in X: d^{2}_{x,y}=\mathop{f}\limits_{z}\{ g(d_{x,z},d_{z,y})\} \leq d_{x,y}$.
%From theorem \ref{theorem3}: $P^{2} \supseteq P$, thus $\mathop{\vee}\limits_{z} (\{\wedge(p_{x,z},p_{z,y})\}\geq p_{x,y}$.
%Lets choose a $z$ for which $\mathop{f}\limits_{z}\{ g(d_{x,z},d_{z,y})\}$ is maximum, thus $d^{2}_{x,y}\leq \mathop{min}\limits_{z}\{g(d_{x,z},d_{z,y})\}$ by definition \ref{def4}. The second member of the equation is minimum at $z=x$ or $z=y$ and since $g(a,0)=a$ and $D$ anti-reflexive ($d_{x,x}=d_{y,y}=0$), thus $d^{2}_{x,y}\leq \mathop{min}\limits_{z}\{g(d_{x,x},d_{x,y})\}=\mathop{min}\limits_{z}\{g(d_{x,y},d_{y,y})\}=d_{x,y}$, which proves $D\supseteq D^{2}$.  % For this $z$, $\wedge(p_{x,z},p_{z,y})\}\geq p_{x,y}$ is minimum for $\wedge \equiv T_{D}$ ( where $T_{D}$ is \emph{drastic product}), thus if $p_{x,z}, p_{z,y} \in [0,1[$ implies $\wedge(p_{x,z},p_{z,y})=0=p_{x,y}$. Therefore, $d^{2}_{x,y}=\varphi(p_{x,y}) \varphi(p_{x,y}) =\infty=\varphi(p_{x,y})$. If $p_{x,z}=1$ or $p_{z,y}=1$ then $min(p_{x,z},p_{z,y})\geq p_{x,y}$, which implies $\varphi(min(p_{x,z},p_{z,y}))<\varphi(p_{x,y})$ and from theorem \ref{theorem2} we obtain, $max(\varphi(p_{x,z}),\varphi(p_{z,y}))<\varphi(p_{x,y})$, thus $max(d_{x,z},d_{z,y}) = g(d_{x,z},d_{z,y}) <d_{x,y}$
Lets prove by absurd this statement: suppose $d_{x,y}^2 > d_{x,y}$ then the minimum of $\mathop{f}\limits_{z}\{ g(d_{x,z},d_{z,y})\}$ must be $> d_{x,y}$. $\mathop{f}\limits_{z}\{ g(d_{x,z},d_{z,y})\}$ is minimum if $f$ and $g$ are minimum. $g$ is minimum if $d_{z,y}=0$ for all $z\in X-\{x\}$, then $g(d_{x,z},d_{z,y})\geq d_{x,z}$. $f$ is minimum if $d_{x,z}\geq d_{x,y}$ for all $z\in X-\{y\}$ then $f(d_{x,y},d_{x,z})\leq f(d_{x,y},+\infty)\leq d_{x,y}$, which contradicts our assumption, $d_{,x,y}^2 > d_{x,y}$. Therefore, $d_{x,y}^2 \leq d_{x,y}$.

By induction we can prove the general result.

 $\forall x,y,z\in X: d_{x,y}^{n+1}=\mathop{f}\limits_{z}\{g(d^{n} _{x,z},d_{z,y})\}$ by hypothesis $d_{x,y}^{n}\leq d^{n-1}_{x,y}$, thus $d_{x,y}^{n+1} \leq \mathop{f}\limits_{z}\{g(d^{n-1} _{x,z},d_{z,y})\}= d^{n}_{x,y}$, which proves the second statement.
\end{proof}

\begin{theorem}
\label{theorem11_proof}
Given a proximity graph $G_{P}=(X,P)$, a distance graph $G_{D}=(X,D)$, and the isomorphism $\varphi$ and $\Phi$ of definition  \ref{def1},  for any algebraic structure $I = ([0,1], \land, \lor )$ with a T-Conorm/T-Norm pair $\langle \land, \lor \rangle$ used to compute the transitive closure of $P$, there exists an algebraic structure $II = ([0,+\infty], f, g )$ with a TD-conorm/TD-norm pair $\langle f, g \rangle$ to compute the isomorphic distance closure of $D$, $D^{T} = \Phi(P^{T})$, which obeys the condition:

\[\forall x_i,x_j,x_k \in X:\mathop{f}\limits_{k} (g(\varphi (p_{ik}),\varphi (p_{kj})) ) = \varphi(\mathop{\lor}\limits_{k} ( (p_{ik}\land p_{kj})))\]

\noindent and vice-versa if we fix $\langle f, g \rangle$ (TD-norm/TD-Conorm) and isomorphism $\varphi$, to obtain $\langle \lor, \land \rangle$:

\[ \forall x_i,x_j,x_k \in X: \mathop{\lor}\limits_{k}( \varphi^{-1} (d_{ik}) \land \varphi^{-1} (d_{kj})) = \varphi^{-1} ( \mathop{f}\limits_{k} (g(d_{ik},d_{kj})) ) \]

\noindent where $\varphi ^{-1}$ is the inverse function of $\varphi$.
\end{theorem}

\begin{proof}
The transitive closure of P is given by $P^{k_1}$ and the distance closure of D by $D^{k_2}$, with $k_1$ and $k_2$ integers. Let $n=max(k_1,k_2)$, thus for $\Phi(P^{n})=D^{n}$ to be true, the following must also be true:  \[\forall x,y,z \in X:\mathop{f}\limits_{z} \{g(\varphi (p_{x,z}),\varphi (p_{z,y})\}=\varphi(\mathop{\lor}\limits_{z}\{p_{x,z}\land p_{z,y}\})\]
%\noindent Rewritting $\Phi(P^{n})=D^{n}=\underbrace{\Phi (P)\circ \cdots \circ \Phi (P)}_{n} \equiv \Phi ^{n} (P)=\Phi (P^{n} )$, which can be simplified to \[\Phi(P^{n})=\Phi (P^{n} )\]

\noindent We can prove by induction that $\Phi(P^{n})=D^{n}$ is true if we assume that the condition in this theorem is true.%, \[\forall x,y,z \in X:\varphi ^{-1} (\mathop{f}\limits_{z} \{g(\varphi (p_{x,z}),\varphi (p_{z,y}))\}=\mathop{\lor}\limits_{z}\{ \land (p_{x,z},p_{z,y})\}\],

 \noindent The condition in this theorem is equivalent to:

 \[\Phi ^{-1} (\Phi (P)\circ \Phi (P))=P^{2} =P\circ P\]

\noindent Where $\Phi(P) \circ \Phi(P)$ is the distance composition using $f$ and $g$, and $P \circ P$ is the transitive composition using $\land$ and $\lor$. We also can define $D^n$ in function of $\Phi$ and $P$.

\noindent \[D^{n}=\underbrace{D\circ \cdots \circ D}_{n}=\underbrace{\Phi (P)\circ \cdots \circ \Phi (P)}_{n}\] Therefore, what we want to prove is: \[\Phi ^{n} (P)=\Phi (P^{n} )\] \noindent given the condition on this theorem is true.

\noindent by induction:\\
\\(1) $\Phi(P)\circ \Phi(P) = \Phi(P^{2})$ (Basis);\\
(2) $\Phi ^{n} (P)=\Phi (P^{n} )$ (Hypothesis);\\
(3) $\Phi ^{n+1} (P)=\Phi (P^{n+1} )$ (Thesis).\\

\noindent Assuming the condition on this theorem $\Phi ^{-1} (\Phi (P)\circ \Phi (P))=P^{2}$ is true, then it is also true that $\Phi (P)\circ \Phi (P)=\Phi(P^{2})$. Thus, $\Phi ^{n+1} (P)= \Phi^{n}(P)\circ \Phi(P)=\Phi(P^{n})\circ \Phi(P)=\Phi(P^{n+1})$ from statements (1) and (2), which proves the theorem.

Let us prove that there exist a pair of binary functions $f$ and $g$ per definition  \ref{def4a}. From theorem \ref{theorem2} we have \[g(\varphi(p_{x,z}),\varphi(p_{z,y}))=\varphi(p_{x,z}\land p_{z,y})\]
\noindent and from the condition in this theorem, we have \[\mathop{f}\limits_{z}\{g(\varphi(p_{x,z}),\varphi(p_{z,y}))\}=\varphi( \mathop{\lor}\limits_{z}\{p_{x,z}\land p_{z,y}\})\] \[\mathop{f}\limits_{z}\{\varphi(p_{x,z}\land p_{z,y})\}=\varphi( \mathop{\lor}\limits_{z}\{p_{x,z}\land p_{z,y}\})\] \noindent %which implies, \[\mathop{f}\limits_{z}\varphi \equiv \varphi \mathop{\lor}\limits_{z}\].
\noindent Therefore, \[f(d_{x,z},d_{z,y}) \equiv \varphi (\varphi^{-1}(d_{x,z})\lor \varphi^{-1}(d_{z,y}))\] \noindent

The conditions of this theorem leads to the equations of theorem \ref{theorem2}:\[g(d_{x,z},d_{z,y})=\varphi(\varphi^{-1}(d_{x,z})\land \varphi^{-1}(d_{z,y}))\]  \[f(d_{x,z},d_{z,y}) \equiv \varphi (\varphi^{-1}(d_{x,z})\lor \varphi^{-1}(d_{z,y}))\].

From these last equations we can also find $\lor$ and $\land$ given $f$, $g$ and the isomorphism $\varphi$: \[p_{x,z}\lor p_{z,y}=\varphi^{-1}(f(\varphi(p_{x,z}),\varphi(p_{z,y})))\]  \[p_{x,z}\land p_{z,y}=\varphi^{-1}(g(\varphi(p_{x,z}),\varphi(p_{z,y})))\]

%Let us prove by induction,

%(1) $n=1:$  $\Phi ^{-1} (\Phi (P)\circ \Phi (P))=P\circ P$  is true by hypothesis

%(2) $n\Rightarrow n+1$ :  $\Phi ^{-1} (\underbrace{\Phi (P)\circ \cdots \circ \Phi (P)}_{n} )=\underbrace{P\circ \cdots \circ P}_{n} \Rightarrow \Phi ^{-1} (\underbrace{\Phi (P)\circ \cdots \circ \Phi (P)}_{n+1} )=\underbrace{P\circ \cdots \circ P}_{n+1} $

%By hypothesis  $\underbrace{\Phi (P)\circ \cdots \circ \Phi (P)}_{n} \equiv \Phi ^{n} (P)=\Phi (P^{n} )$  then let's compose this result with  $\Phi (P)$ ,

%\[\begin{array}{l} {\Phi (P)\circ \Phi ^{n} (P)=\Phi (P)\circ \Phi (P^{n} )} \\ {\Phi ^{n+1} (P)=\Phi (P)\circ \Phi (P^{n} )} \end{array}\]

%By hypothesis   $\Phi (P)\circ \Phi (P)=\Phi (P\circ P)$  then  $\Phi (P)\circ \Phi (P^{n} )=\Phi (P\circ P^{n} )=\Phi (P^{n+1} )$ , which completes our proof,

%\[\Phi ^{-1} (\Phi ^{n+1} (P))=P^{n+1} .\]

%Therefore  $\Phi$ exists and consequently $\varphi$ exists and is a isomorphism between the two generated closure graphs.
\end{proof}

\begin{theorem}
\label{theorem3a_proof}
Given a finite proximity graph $G_P(X,P)$, and an algebraic structure $I=([0,1], \lor, \land)$, with a T-Conorm/T-Norm pair $\langle \land, \lor \rangle$ used to compute the transitive closure of $G_P$, if $I$ is a dioid, then the transitive closure $G_P^T(X,P^T)$ can be computed by equation \ref{TC1} for a finite $\kappa$.
\end{theorem}

See \cite{Gondran2007} for proof; further discussion and examples also see \cite{Han2004, Han2007,Pang2003,Klir1995}.

\begin{theorem}
\label{theorem3c_ap}
Given a finite distance graph $G_D(X,D)$, and an algebraic structure $II=(\{[0,+\infty],f,g)$, with a TD-Conorm/TD-Norm pair $\langle f, f \rangle$ used to compute the distance closure of $G_D$, if $II$ is a dioid, then the distance closure $G_D^T(X,D^T)$ can be computed in finite time via the transitive closure of isomorphic graph $G_P(X,P)$ with algebraic structure $I$ obtained by an isomorphism satisfying Theorem \ref{theorem11}. In other words, if $II$ is a dioid, via an isomorphism satisfying Theorem \ref{theorem11} we obtain an algebraic structure $I$ which is also a dioid.
\end{theorem}

This theorem can be easily proven from theorems \ref{theorem5}, \ref{theorem11} and \ref{theorem3a}, by evoking the isomorphism to proximity space.

\begin{corollary}
\label{theorem12_proof}
Given the isomorphism constraint on the T-Norm from algebraic structure $I$ (eq. \ref{eq_constraint_and}) from theorem \ref{theorem11}, let  $f \equiv \min$,  $g \equiv +$  and  $\varphi$  a distance function, per definition  \ref{def1}. If  $\lor \equiv \max$  as T-Conorm, then the T-Norm operator $\land$ exists and $\varphi$ is its generator function.
\end{corollary}

\begin{proof}
We have seen in theorem ~\ref{theorem2} that  $\varphi (x\wedge y)=g(\varphi (x),\varphi (y))$  therefore $\forall x,y,z \in P$ and by theorem \ref{theorem11}: \[\begin{array}{l} {\varphi ^{-1} (\mathop{min}\limits_{z} \{\varphi (p_{x,z})+\varphi (p_{z,y})\})=\mathop{max}\limits_{z}\{p_{x,z}\land p_{z,y}\}} \\ {\mathop{max}\limits_{z} \{\varphi ^{-1} (\varphi (p_{x,z})+\varphi (p_{z,y}))\}=\mathop{max}\limits_{z} \{p_{x,z}\land p_{z,y}\}} \\ {\Rightarrow } \\ {\varphi ^{-1} (\varphi (p_{x,z})+\varphi (p_{z,y}))=p_{x,z}\land p_{z,y}} \end{array}\] This last result is the characterisation function of T-Norms, according to theorem ~\ref{characterization} \cite{Klir1995}, which states that  $\land $  is a T-Norm and  $\varphi $ is  the decreasing generator function (obeying definition \ref{def1}).
\end{proof}

\setcounter{theorem}{7}

\begin{theorem}
\label{theorem9_proof}
Given the isomorphism $\varphi $, if $D^{mc}$ is the metric closure with $f \equiv \min$ and $g_{1}\equiv +$, and $D^{um}$ is the ultra-metric closure with $f\equiv \min$ and $g_{2}\equiv \max$ then $D^{mc}\dot \supseteq D^{um}$ is equivalent to $P^{mc}\subseteq P^{um}$, where $D^{mc}=\Phi(P^{mc})$ and $D^{um}=\Phi(P^{um})$. Therefore, $\Delta(P^{um}) \geq \Delta(P^{mc}).$
\end{theorem}

\begin{proof}
We can prove by induction that:

1)  $D^{2} \dot \supseteq \Phi (P^{2} )$ ;

2)  $\left\{\begin{array}{c} {H:D^{n}\dot \supseteq \Phi (P^{n} )} \\ {T:D^{n+1} \dot \supseteq \Phi (P^{n+1} )} \end{array}\right. $

Let's prove 1)

 $\forall x,y,z\in X:D_{mc}^{2}=\mathop{f }\limits_{z} (d_{x,z}+d_{z,y})=\mathop{f }\limits_{z} (\varphi (p_{x,z})+\varphi (p_{z,y}))\ge \mathop{f }\limits_{z} (g_{2}(\varphi (p_{x,z}),\varphi (p_{z,y})))=D_{um}^{2}$, therefore $D_{mc}^{2} \dot \supseteq D_{um}^{2}$.

2) by the hypothesis we know that  $\forall x,y,z\in X:D^{n}\ge \Phi (P^{n})$ , then using this result we have  $\forall x,y,z\in X:D^{n+1}=\mathop{f}\limits_{z} \{ d^{n}_{x,z}+d _{z,y}\} \ge \mathop{f}\limits_{z} \{ \varphi (p^{n}_{x,z})+\varphi (p _{z,y})\} $ , because   $\mathop{f}\limits_{z} \{ \varphi (p^{n}_{x,z})+\varphi (p_{z,y})\} \ge \mathop{f}\limits_{z} \{ \varphi (p^{n}_{x,z})\vee \varphi (p_{z,y})\} $  and using theorem \ref{theorem2},   $\mathop{f}\limits_{z} \{ g_{2}(\varphi (p^{n}_{x,z}),\varphi (p_{z,y}))\} =\varphi (\mathop{\vee}\limits_{z} \{ p^{n}_{x,z}\wedge p_{z,y}\} )=\Phi (P^{n+1})$ , so

 $\forall x,y,z\in X:D^{n+1}\ge \Phi (P^{n+1})$ , which proves that  $D^{mc}\equiv D^{n} \dot \supseteq \Phi (P^{n} ) \equiv D^{um}$.
\end{proof}

\begin{theorem}
\label{complement_theorem}
Given a fuzzy complement $c(x)$, a T-Norm $DT_{\land}^1=\frac{ab}{a+b-ab}$ and a T-Conorm $max(a,b)$, then the triple has no involutive complement, which satisfies the De Morgan's laws.
\end{theorem}

\begin{proof}
A complement is involutive if $c(c(x))=x$. If the complement $c(x)$ satisfies the De Morgan's laws we have: \[\overline{a \lor b}=\bar a \land \bar b\] \[c(max(a,b))=\frac{c(a)c(b)}{c(a)+c(b)-c(a)c(b)}\] \noindent without loss of generality let $a\geq b$ \[c(a)=\frac{c(a)c(b)}{c(a)+c(b)-c(a)c(b)} \] \[c(a)(1-c(b))=0\] \[c(a)=0 \lor c(b)=1\] \noindent the only function that satisfies this condition is the threshold function, which is not involutive \cite{Klir1995}.
\end{proof}

%%%%%%%%%%%%%%%%%%%%%%%%%%%%%%%%%%%%%%%%%%%%%%%%%%%%

\newpage
\section{Optimal Dombi T-Norm for a characteristic path length}

%\section{Exploring the isomorphism with the Dombi T-Norm}

%In the previous section, we fixed the isomorphism $\varphi$ to the simplest and most common distance function given by formula \ref{distance_measures4}. This allowed us to search the T-Norm/T-Conorm pairs that better preserve logical axiomatics and are closed to our intuitive metric closure operator. Here, we fix the metric closure operators instead, and search the space of possible functions $\varphi$.

We have seen that we can apply an infinity of pairs of T-Norms and T-Conorms to calculate distance closure, and compute shortest paths in distance graphs. In this formulation (see corollary \ref{theorem12}), we fix the T-Conorm with $\lor \equiv max$, allowing us to explore many options for the T-Norm $\land$. The T-Norm is defined via the T-Norm generator isomorphism $\varphi$ (corollary \ref{theorem12}). Then, using $\langle f \equiv min,g\equiv + \rangle$ as the TD-norm/TD-conorm pair for computing the metric closure, via the APSP/Dijkstra, distance product or equivalent, we can sweep the space of possible T-Norms, thus simultaneously exploring the range of possible isomorphisms. %This poses us the following question: which T-Norm/isomorphism is optimal, given a set of assumptions, for the shortest paths calculation?
%
%LMR: optimal in what sense? what set of assumptions, this is very vague here and should be more clear from the start
%
%We answer this question here for the
In this section we explore the Dombi T-Norm family,
%which provides the range of T-Norms between the lower and upper bounds through the $\lambda$ parameter. Recall that
where the Dombi T-Norm generator is:

\begin{equation}
\label{dombi4}
\varphi(x)=\left(\frac{1}{x}-1\right)^{\lambda}
\end{equation}

\noindent where $\lambda$ is the sweeping parameter. The parameter $\lambda$ in the T-Norm generator takes values in $]0,+\infty[$: $\lambda \to 0$ lower bound (\emph{drastic product}) and $\lambda \to \infty$ is the upper bound (\emph{minimum}). The reason we choose this T-Norm generator is because it yields the more commonly used isomorphism from proximity to distance; when $\lambda = 1$, \cite{Eckhardt2009} \cite{Strehl}, the generator of eq. \ref{dombi4}, becomes the isomorphism of formulae \ref{distance_measures4}, which we have used in the previous section:

 \[\varphi(x)=\frac{1}{x}-1.\]

We have seen that when T-Norm and T-Conorm ($\lor,\land$) are fixed, the transitive closure and the distance closure are equivalent via isomorphism $\varphi$.

For empirical analysis of complex networks it is desirable that properties of the graphs obtained via specific closures, such as \emph{average shortest path}, be simultaneously characteristic in both spaces (proximity and distance). That is, the fluctuations of the mean, must be constrained on both spaces (average shortest path and average strongest path). In order to have a characteristic average path length, the shortest paths distribution must follow approximately a normal distribution. We want to find the best $\lambda$, using the Dombi T-Norm generator, which guarantees these assumptions, while fixing $\lor = \max$.

Assuming that the shortest path distribution of a distance graph follows a normal distribution, the probability density function for a normal random variable $X$, here the shortest path, is given by:

\begin{equation}
\label{density}
h_{X}(x)=\frac{1}{\sqrt{2\pi \sigma^{2}}} e^{\frac{-(x-\mu)^{2}}{2\sigma^{2}}}
\end{equation}

\noindent where $\mu$ and $\sigma$ are the mean and standard deviation of the normal distribution.

\begin{figure}[!th]
\centering
\includegraphics[width=80mm,height=60mm]{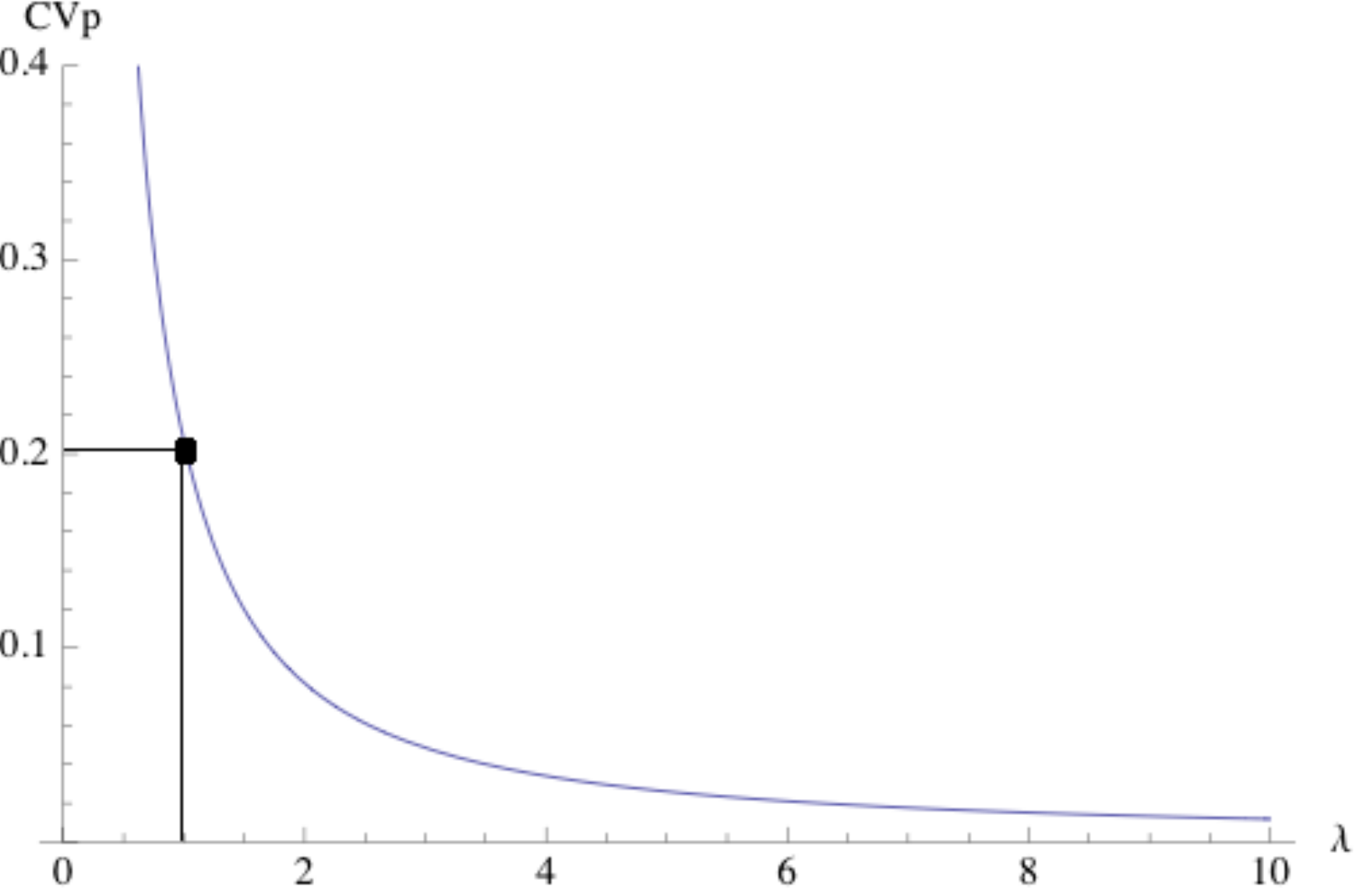}
\caption{Study of the fluctuations in proximity space, $CV_{p}$ as function of $\lambda$ for $\mu=10$ (average path length in distance space) with $CV_{d}=0.2$.\label{cvp}}
\end{figure}

The mean of a random variable $Y=j(X)$, which is a monotonic function of X, where X is the random variable representing shortest path in a distance graph, and Y the random variable representing the strongest path in the isomorphic distance graph, is given by:

\begin{equation}
\label{avg}
<Y>=\int_{0}^{\infty} j(x)h_{X}(x)dx
\end{equation}

\noindent In our case, \[j(x)=\varphi^{-1}(x)=\frac{1}{x^{\frac{1}{\lambda}}+1}\]

Therefore, the fluctuations of the mean, in the proximity space are given by:

\begin{equation}
\label{fluctuations}
CV_{p}=\frac{\sigma_{p}}{\mu_{p}}=\frac{\sqrt{<Y^{2}>-<Y>^{2}}}{<Y>}
\end{equation}

\noindent where $CV_{p}$ is the \emph{coefficient of variability}\footnote{The coefficient of variability is scale invariant.}, and $\sigma_{p}$ and $\mu_{p}$ are the standard deviation and mean of the strongest path in the proximity space and $<Y^{2}>$ is given by:

\begin{equation}
\label{varY}
<Y^2>=\int_{0}^{\infty} j^{2}(x)h_{X}(x)dx
\end{equation}

The fluctuations in the distance space of the shortest path, are given by the \emph{coefficient of variability}, $CV_{d}$:

\begin{equation}
\label{fluctuations2}
CV_{d}=\frac{\sigma}{\mu}
\end{equation}

\begin{figure}[!th]
\centering
\includegraphics[width=100mm,height=65mm]{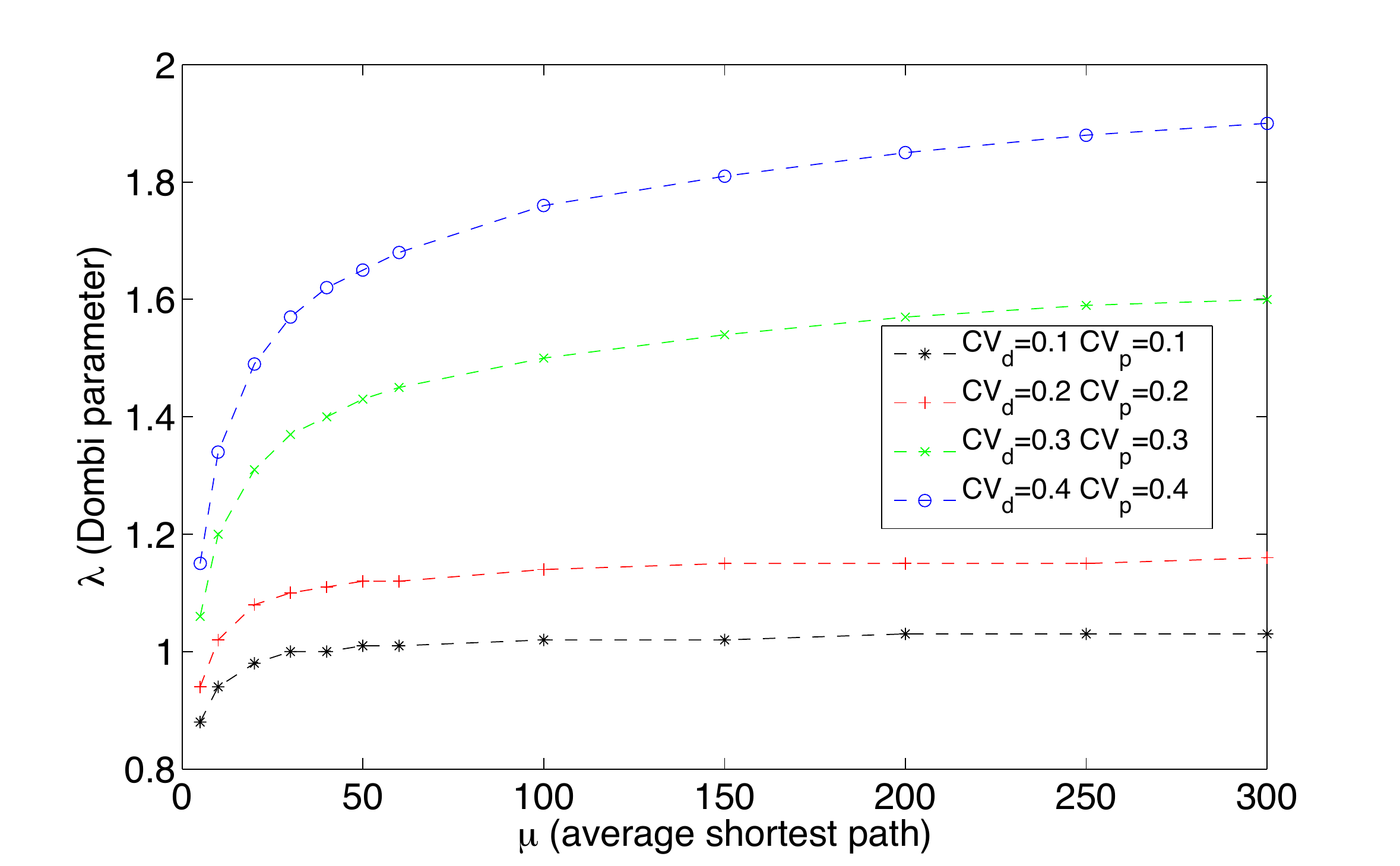}
\caption{$\lambda$ versus $\mu$ for several coefficients of variability $CV_{d}$ and $CV_{p}$\label{cv}}
\end{figure}

The dependence of $CV_p$ on $CV_d$ comes from equations \ref{density}, \ref{avg} and \ref{varY}. In figure \ref{cvp} we plot the theoretical relation between $\lambda$ and $CV_{p}$ for $\mu=10$ (average shortest path in distance space is normally distributed) and $CV_{d}=0.2$, using equation \ref{fluctuations}; the shape is preserved for different parameter values. We can see from this figure that the coefficient of variability in the proximity space is minimum when $\lambda$ converges to the \emph{min} T-Norm ($\lambda \rightarrow +\infty$); the ultra-metric closure. However, from our assumptions we require that $CV_{p}\approx CV_{d}=0.2$, in this case. The marked point in the figure \ref{cvp} shows the point where the assumptions are met. We observe that $\lambda \approx 1$ in this scenario. %Moreover, we know that applying different distance closures we observe a distortion $\Delta$, which is maximum at the ultra-metric closure. As a result the fluctuations to the average path length in the distance space increases when we approximate the ultra-metric closure, changing the shape of the distribution. Ergo, there must be an equilibrium point between the drastic product and the \emph{min} T-Norm, where fluctuations are acceptable (intermediate value for $\lambda$), which in this case the equilibrium is is at $\lambda \approx 1$.

To inspect in more detail the best value or values for $\lambda$, using the metric closure we plot, in figure \ref{cv} the theoretical $\lambda$ versus $\mu$ (average shortest path), for several acceptable coefficients of variability in both spaces, assuming that the optimal value should share a controlled $CV_{d} \approx CV_{p} \leq 0.6$. The results from this figure are obtained by finding the root ($\lambda$) of the equation: \[CV^{theoretical}_{p}(\lambda)-CV_{p}=0\] \[CV^{theoretical}_{p}(\lambda)=\frac{\sqrt{<Y^{2}>-<Y>^{2}}}{<Y>}\]  \noindent Where $<Y^2>$ and $<Y>$ are given by equations \ref{density}, \ref{avg} and \ref{varY} with $j(x)=\frac{1}{x^{\frac{1}{\lambda}}+1}$ and we assume $h_{X}(\mu,\sigma)$ is normally distributed with $\sigma=\mu\times CV_{d}$ ($\mu$ is the average shortest path) with $CV_{d} \approx CV_{p}$ the real data fluctuations. We use \emph{Mathematica 7} to find the roots of this equation. From this figure we can see that when we increase the coefficients of variability, $\lambda$ also increases. However, $\lambda$ remains contained in the interval $[0.8,1.9]$. For small average shortest paths the best $\lambda \in [0.8,1.2]$, where after a transient ($\mu \approx 25$), $\lambda$ reaches an equilibrium, independent of scale factors ($\lambda$ becomes invariant). The scale factor associated to the average shortest path length (characteristic for each network), depends mainly on the weights distribution.  We can also observe that for very small fluctuations ($CV_{d}=CV_{p}=0.1$), $\lambda$ becomes invariant for values $\approx 1$. $\lambda=1$ is an optimal asymptotic value for small fluctuations, since $CV\geq 0$. In real data in order to guarantee a characteristic mean (average strongest and shortest path), in both spaces (proximity and distance), the fluctuations should be as small as possible. However in real data the shortest path distribution only approximates to the normal distribution, which is one of our assumptions, resulting in higher fluctuations, for both spaces. For fluctuations $CV_{d}\approx CV_{p} \in [0,0.4]$ we should use an isomorphism with $\lambda \in [0.8,1.9]$. For $CV \approx 0$ the asymptotical optimal value is $\lambda=1$ (see figure \ref{cv}). This gives us a lower bound to calculate the desired metric closure in a distance graph to minimize fluctuations, $\lambda$ should be larger or equal than 1 ($\lambda \geq 1$). To control fluctuations in both spaces (proximity, distance) we should choose $\lambda$ according to the fluctuations obtained in the distance or proximity spaces (this can be seen as an optimization problem).

In most applications, researchers use mappings between proximity and distance spaces similar to $\lambda=1$, using isomorphisms $\varphi=\frac{1}{x}$ or $\varphi=\frac{1}{x}-1$. We have to alert that the first choice $\varphi=\frac{1}{x}$ is not mathematically correct, since it maps $\varphi:[0,1] \rightarrow [1,+\infty]$, which is not a distance space. $\lambda=1$ leads to the more common $\varphi$ and asymptotical optimal value, assuming small fluctuations. However, to constrain fluctuations we may want to use other values of $\lambda\geq 1$, depending on the level real data fluctuations. %Moreover, as we saw, $\lambda=1$, also leads to an axiomatically desirable T-Norm/T-Conorm pair, when we do not constrain $\lor \equiv max$.

%%%%%%%%%%%%%%%%%%%%%%%%%%%%%%%%%%%%%%%%%%%%%%%%%%%%%%%
%

%%%%%%%%%%%%%%%%%%%%%%%%%%%%%%%%%%%%%%%%%%%%%%%%%%%%

\newpage
\section{Community Structure in Example Networks}

%%%%%%%%%%%%%%%%%%%%%%%%%%%%%%%%%%%%%%%%%%%%%%%%%%%%%%%

\begin{figure}[!th]
\centerline{\includegraphics[width=1\textwidth]{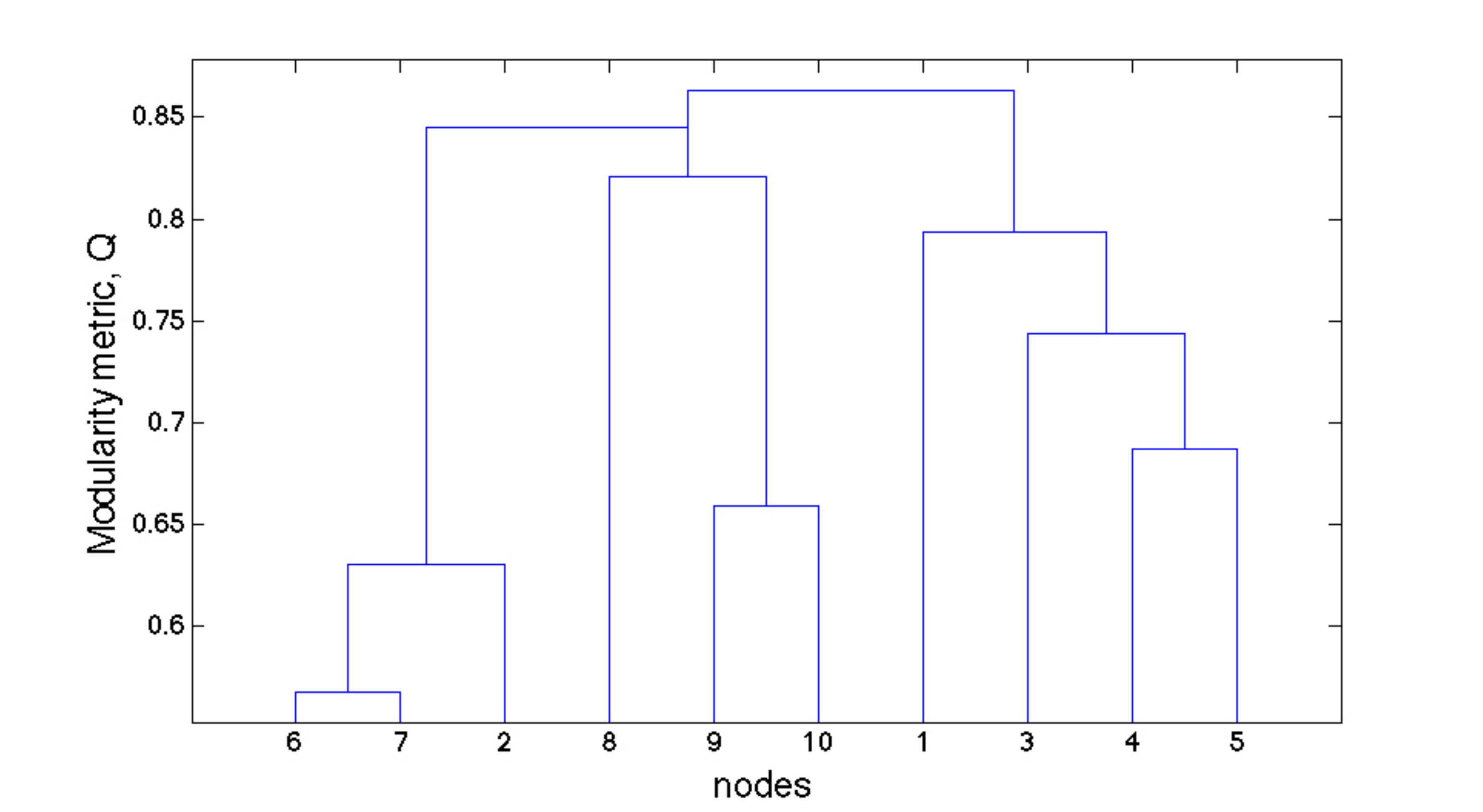}}
\caption{\small Community Structure of Toy Network with Newman's Fast Algorithm.
%(e) number of symmetry breaks between one node and the remaining nodes.
}
\label{fig_toy_NF}
\end{figure}

\begin{figure}[!th]
\centerline{\includegraphics[width=1\textwidth]{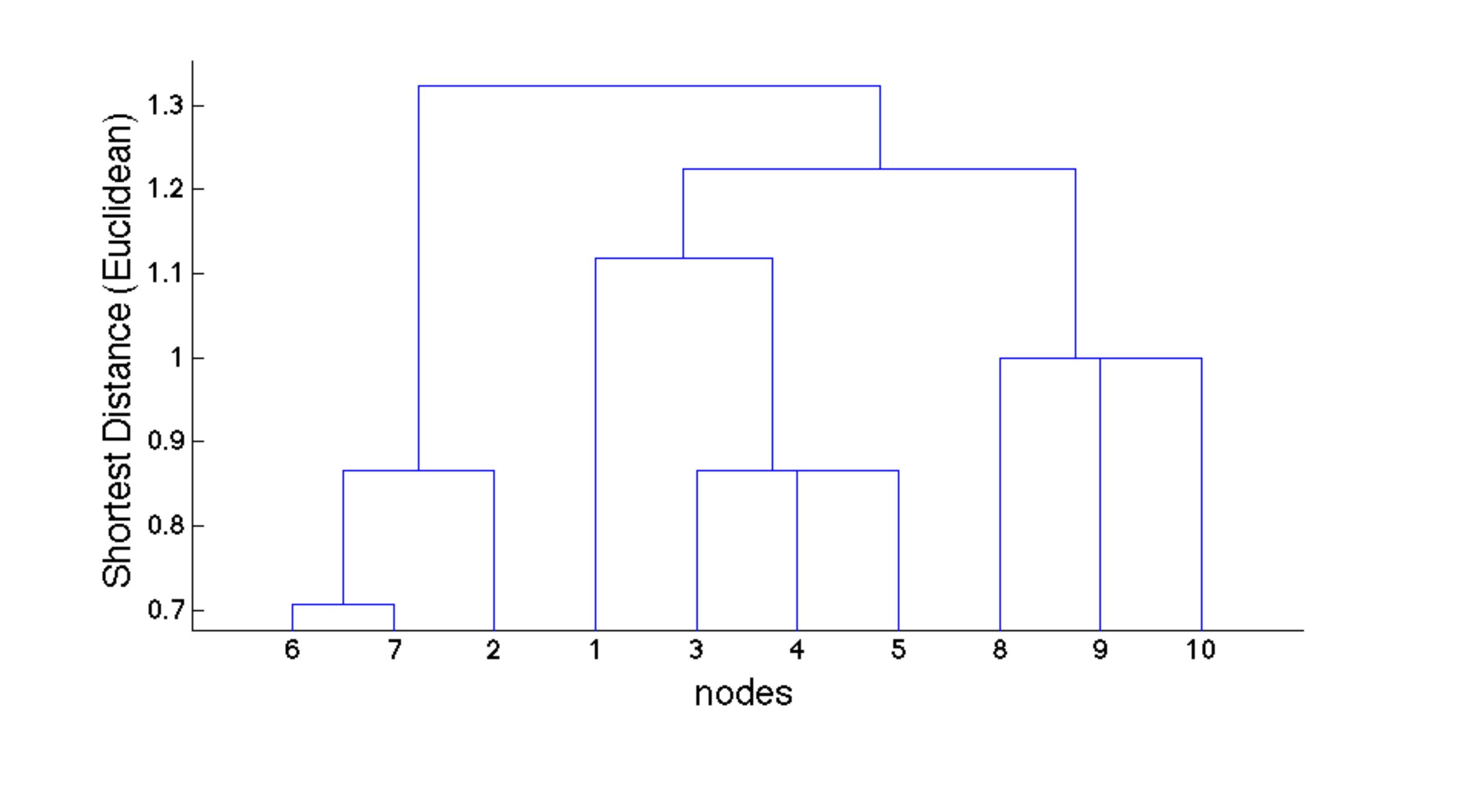}}
\caption{\small Community Structure of Toy Network with Hierarchical Clustering.
%(e) number of symmetry breaks between one node and the remaining nodes.
}
\label{fig_toy_H}
\end{figure}

\begin{figure}[!th]
\centerline{\includegraphics[width=1\textwidth]{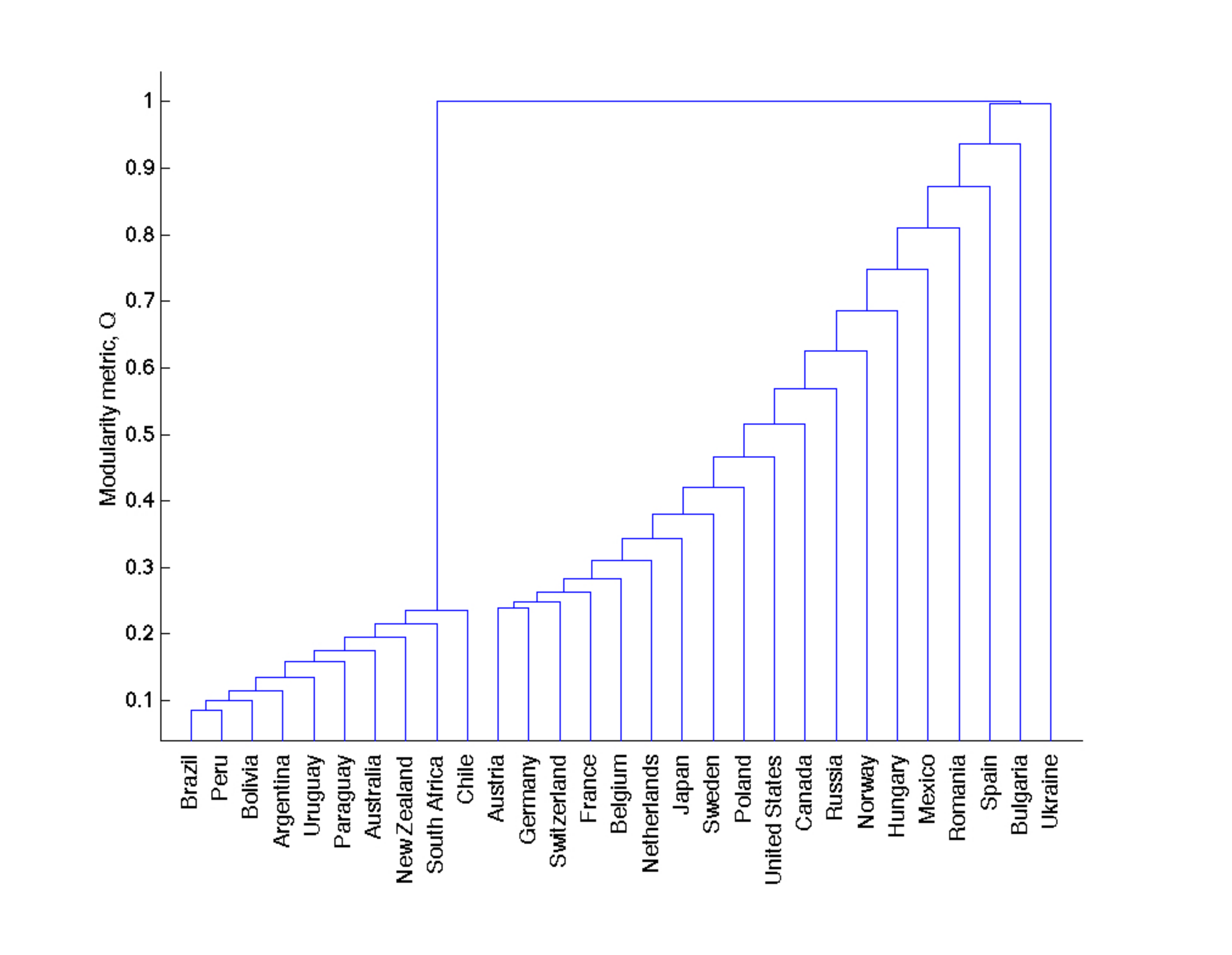}}
\caption{\small Community Structure of Flu Network with Newman's Fast Algorithm.
%(e) number of symmetry breaks between one node and the remaining nodes.
}
\label{fig_toy_NF}
\end{figure}

\begin{figure}[!th]
\centerline{\includegraphics[width=1\textwidth]{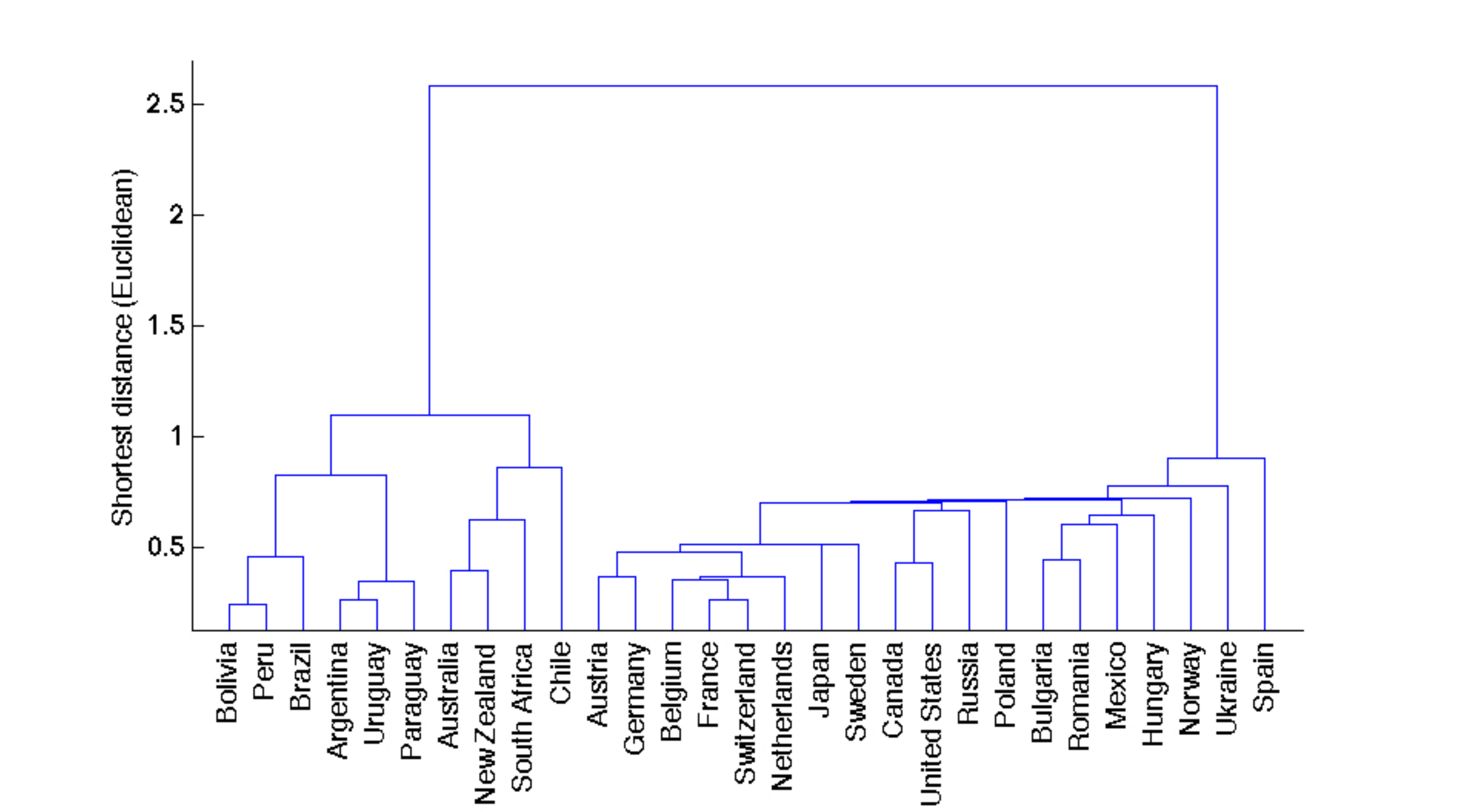}}
\caption{\small Community Structure of Flu Network with Hierarchical Clustering.
%(e) number of symmetry breaks between one node and the remaining nodes.
}
\label{fig_toy_H}
\end{figure}

%\begin{thebibliography}{00}

% \bibitem{label}
% Text of bibliographic item

% notes:
% \bibitem{label} \note

% subbibitems:
% \begin{subbibitems}{label}
% \bibitem{label1}
% \bibitem{label2}
% If there is a note, it should come last:
% \bibitem{label3} \note
% \end{subbibitems}

%\bibitem{}

%\end{thebibliography}

\end{document}